\documentclass[leqno,a4paper]{article}
\usepackage [dvipsnames] { xcolor }
\usepackage{amsfonts}
\usepackage{amsmath, amsfonts, amsthm, amssymb, amscd}
\usepackage[mathcal]{euscript}
\usepackage{mathrsfs}
\usepackage{dsfont}
\usepackage{ulem}
\usepackage{bbm}
\usepackage{pifont}

\usepackage{hyperref}
\hypersetup{linktoc = all}
\hypersetup{hidelinks}

\usepackage{slashed}
\usepackage{a4wide}
\usepackage{graphicx}  
\graphicspath{ {./} }

\usepackage{lineno} 

\modulolinenumbers[5]

\newtheorem{theorem}{Theorem}[section]
\newtheorem{proposition}[theorem]{Proposition}
\newtheorem{lemma}[theorem]{Lemma}
\newtheorem{corollary}[theorem]{Corollary}
\newtheorem{definition}[theorem]{Definition}
\newtheorem{remark}[theorem]{Remark}

\numberwithin{equation}{section}

\usepackage{multirow}
\usepackage{tabularx}
\usepackage{lscape}

\newcommand \del \partial

\newcommand \Kcal {\mathcal K}

\newcommand \Hcal {\mathcal H}

\newcommand \Hinit{\Hcal^{\text{\bf init}}}

\newcommand \Jcal {\mathcal{J}}

\newcommand \Rcal{\mathcal{R}}

\newcommand \delb {\bar {\del}}

\newcommand \Boxt {\widetilde {\Box}}
\newcommand \delu {\uline{\del}}

\newcommand \Tu {\uline{T}}

\newcommand \hu {\uline{h}}

\newcommand \Pu{\uline{P}}
\newcommand \Psiu{\uline{\Psi}}
\newcommand \Phiu{\uline{\Phi}}

\newcommand \Ebf {{\bf E}}

\newcommand \RR{\mathbb{R}}
\newcommand \eps{\varepsilon}

\newcommand {\vep}{\varepsilon}

\newcommand{\diff}{\mathrm{d}}

\newcommand \Mcal {\mathcal{M}}
\newcommand \etau {\underline{\eta}}
\newcommand \gSch {{g_{\mathrm{S}}}}
\newcommand \mSch {{m_{\mathrm{S}}}}

\newcommand \be {\begin{equation}}
\newcommand \ee {\end{equation}}

\newcommand \bea {\be\aligned}
\newcommand \eea {\endaligned\ee}

\makeatletter
\def\hlinew#1{%
	\noalign{\ifnum0=`}\fi\hrule \@height #1 \futurelet
	\reserved@a\@xhline}
\makeatother

\title{Gravitational Radiation Coupled to a Real Massive Scalar Field: Single-Direction Observability \footnote{The present work is supported by Shaanxi Fundamental Science Research Project for Mathematics and Physics (Grant No.22JSZ003)}}

\author{
Yue MA \footnote{School of Mathematics and Statistics, Xi'an Jiaotong University, Xi'an, Shaanxi 710049, P.R. China.\ E-mail: yuemath@xjtu.edu.cn}, Xianghao Wu \footnote{School of Mathematics and Statistics,
Xi'an Jiaotong University, Xi'an, Shaanxi 710049, P.R. China.\
E-mail: xianghaowu@stu.xjtu.edu.cn} and Weidong Zhang \footnote{School of Mathematics and Statistics,
Xi'an Jiaotong University, Xi'an, Shaanxi 710049, P.R. China.\
E-mail: zhangwdmath@gmail.com}
}

\begin{document}
	\maketitle

\begin{abstract}
We study small-data global solutions to the Einstein--Klein--Gordon system in wave coordinates, with a real scalar field of positive mass and initial data that agree with Schwarzschild data outside a compact set. In mass-adapted time $\hat t$, the linearized trace of the metric perturbation has a leading $\hat t^{-1}$ profile whose continuous boundary value is the late-time limit of the trace radiation field. We prove a quantitative observability estimate: the value of this profile in any single direction at future null infinity controls the initial Klein--Gordon energy. In particular, a zero limit in one direction forces the scalar field to vanish identically.
\end{abstract}

\tableofcontents
\newpage
\section{Introduction}
\subsection{Objective}
An astrophysical event far from Earth can perturb an otherwise stationary Schwarzschild exterior. This motivates the question of whether the asymptotic gravitational field can reveal the presence of matter coupled to gravity in the source region. Here we consider a real scalar field of positive mass and small initial data that agree with Schwarzschild data outside a compact set.

Our observable is constructed from the linearized trace of the metric perturbation in the specified wave coordinates,
\begin{equation}
	\hbar:= \eta^{\alpha\beta}h_{\alpha\beta} = \eta^{\alpha\beta}(g_{\alpha\beta} - \eta_{\alpha\beta})
\end{equation}
where $\eta_{\alpha\beta} = \mathrm{diag}(-1,1,1,1)$. We first take the
Friedlander-type radiation limit \cite{Friedlander1962} at fixed
mass-adapted retarded time, as in \eqref{eq1-13-oct-2026}, and then take
the late-time limit of that radiation field. A zero late-time limit in
even one angular direction forces the scalar field to vanish identically.
Conversely, for zero scalar field this limit vanishes in every direction.
This is a mathematical observability result in the stated coordinates:
the trace profile is not the strain measured by a gravitational-wave
detector, and we do not establish a reconstruction from detector data or
invariance under further changes of gauge.

\subsection{Notation and main result}
To be more precise, let $(\Mcal,g,\phi)$ be the spacetime under consideration with $\Mcal\simeq \RR^4_+ = \big\{(x^0,x^a) \mid x^0\geq 2\big\}$. The metric $g$ and the real scalar field $\phi$ solve the Einstein-Klein-Gordon system\footnote{Throughout this work the Greek indices $\alpha,\beta,\gamma,\cdots$ and $\mu,\nu,\lambda,\cdots$ take values in $\{0,1,2,3\}$ for the spacetime components, while the Latin indices $a,b,c,\cdots$ take values in $\{1,2,3\}$ for the space components.}
\begin{equation}\label{eq1-19-08-2026}
G_{\alpha\beta} = 8\pi T[\phi]_{\alpha\beta},
\quad 
T[\phi]_{\alpha\beta} = \nabla_{\alpha}\phi\nabla_{\beta}\phi - \frac{1}{2}\big(\nabla^{\mu}\phi\nabla_{\mu}\phi + 2V(\phi)\big)g_{\alpha\beta}.
\end{equation} 
The nonlinear stability of Minkowski space for the Einstein vacuum
equations was proved in \cite{CK1993}. Global stability in wave
coordinates was proved by Lindblad and Rodnianski
\cite{Lindblad-Rodnianski-05,LR2010}, with a massless scalar field also
included in \cite{LR2010}.

It was shown in \cite{PLF-YM-book2,IP3,PLF-YM-PDE} that a small regular initial data set imposed on $\{t=2\}$ generates a globally defined, globally hyperbolic regular solution that remains close to the Minkowski vacuum. A complementary approach was developed by Wang \cite{Wang}, who constructed an intrinsic hyperboloid framework in curved spacetimes and obtained the geometric and boundary control near the asymptotic light cone needed for the corresponding energy method. In the present work we use the formulation established in \cite{PLF-YM-book2}, where \eqref{eq1-19-08-2026} is written as
\begin{equation}\label{eq1-29-nov-2025}
	\aligned
	&R_{\alpha \beta}  = 8\pi\big(\nabla_\alpha \phi\nabla_{\beta} \phi + V(\phi) \, g_{\alpha \beta} \big),
	\\
	&\Box_g\phi - V'(\phi) = 0,
	\endaligned
\end{equation}
where $R_{\alpha\beta}$ is the Ricci curvature of the spacetime $(\Mcal,g)$, and $V$ is the scalar potential. In the present case, we fix $c>0$ and take the massive quadratic potential
\begin{equation}\label{eq2-29-nov-2025}
	V(\phi) = \frac{c^2}{2}\phi^2.
\end{equation}
It would be natural to ask whether the same conclusions hold when
$V$ includes cubic and higher-order terms. We do not treat such terms here.

In a {\sl wave gauge}, the above system \eqref{eq1-29-nov-2025} is written as
\begin{equation}\label{eq3-29-nov-2025}
	\aligned
	&\Boxt_g h_{\alpha \beta} = F_{\alpha \beta}(h;\del h, \del h) - 16\pi \del_{\alpha} \phi\del_{\beta} \phi - 16\pi V(\phi)g_{\alpha \beta},
	\\
	&\Boxt_g\phi - V'(\phi) = 0,
	\\
	&\Box_g x^{\alpha} = 0,
	\endaligned
\end{equation}
where $\Boxt_g = g^{\mu\nu}\del_{\mu}\del_{\nu}$. The nonlinear term $F_{\alpha\beta}$ is composed of {\sl null terms} and {\sl quasi-null terms}, see \eqref{eq1-02-dec-2025} for the precise expressions.
The Einstein vacuum equations in wave coordinates satisfy the weak null
condition introduced in \cite{LR2003WeakNull}. 

In the present work, to keep the focus on the main mechanisms, we treat compact perturbations around a given Schwarzschild spacetime. For this purpose, let us recall the Schwarzschild metric in standard wave coordinates $(x^0, x^1, x^2, x^3)$ (cf.~\cite{Asanov}):
\begin{equation}\label{eq Sch-wave}
	\aligned
	\gSch_{00} &= - \frac{r-m_{\mathrm{S}}}{r+m_{\mathrm{S}}},\qquad 
	\gSch_{0a} = 0,
	\qquad
	\gSch_{ab} = \frac{r+m_{\mathrm{S}}}{r-m_{\mathrm{S}}} \omega_a \omega_b + \frac{(r+m_{\mathrm{S}})^2}{r^2}(\delta_{ab} - \omega_a \omega_b)
	\endaligned
\end{equation}
with $\omega_a := x_a/r$.  The inverse metric components are given by
\begin{equation}
	\gSch^{00} = -\,\frac{r + \mSch}{r - \mSch}, 
	\qquad
	\gSch^{0a} = 0, 
	\qquad
	\gSch^{ab} = \frac{r - \mSch}{r + \mSch}\,\omega^a\omega^b + \frac{r^2}{(r + \mSch)^2}\,\bigl(\delta^{ab} - \omega^a\omega^b\bigr).
\end{equation}
We now recall the global existence result established in \cite{PLF-YM-book2}.
\begin{theorem}[Theorem~2.4 of \cite{PLF-YM-book2}]
	\label{thm main-PDE}
	Consider the Einstein-massive field equations \eqref{eq3-29-nov-2025} together with an initial data set satisfying the constraints and prescribed on the hypersurface $\{t=2\}$:
	\begin{equation}\label{eq main initial-data}
		\aligned
		& g_{\alpha \beta}|_{\{t=2\}} = g_{0, \alpha \beta},
		\qquad
		&&
		\del_t g_{\alpha \beta}|_{\{t=2\}} = g_{1, \alpha \beta},
		\\
		& \phi|_{\{t=2\}} = \phi_0,
		\qquad
		&&
		\del_t \phi |_{\{t=2\}} = \phi_1, 
		\endaligned
	\end{equation}
	which, on $\{t=2\}$ outside the unit ball $\{ r < 1\}$, is assumed to coincide with the restriction of Schwarzschild spacetime of mass $m_{\mathrm{S}}\geq0$ (in the wave gauge \eqref{eq Sch-wave}), i.e.
	\begin{equation}\label{eq1-07-dec-2025}
		g_{\alpha \beta}(2, \cdot) = {\gSch}_{\alpha \beta}, \qquad \del_tg_{\alpha \beta}(2, \cdot) = \phi(2, \cdot) = \del_t\phi(2, \cdot) = 0
		\quad \text{in } \, \big\{ r=|x| \geq 1 \big\}.
	\end{equation}
	Then, for any sufficiently large integer $N$ ($N\geq 14$ being sufficient), there exist constants $\vep_0,C_1, \delta>0$ such that, provided
	\begin{equation}\label{eq main smallness-condition}
		\sum_{\alpha, \beta} \| \del g_{0, \alpha \beta}, g_{1, \alpha \beta} \|_{H^N(\{r < 1\})} 
		+ \|\phi_0 \|_{H^{N+1}(\{ r <1\})} 
		+ \|\phi_1 \|_{H^{N}(\{ r <1\})} 
		+ m_{\mathrm{S}} \leq \vep\leq \vep_0
	\end{equation}
	holds at the initial time, then the solution associated with the initial data \eqref{eq main initial-data} exists for all times $t \geq 2$ and, furthermore, 
	\begin{equation}\label{eq main energy-bound}
		\aligned
		\Ebf(s, \del^IL^Jh_{\alpha \beta})^{1/2}
		& \leq C_1\vep s^{\delta},
		\qquad&&&& | I | + |J| \leq N,
		\\
		\Ebf_c(s, \del^IL^J\phi)^{1/2}
		& \leq C_1\vep s^{\delta+1/2},
		\qquad &&&&| I | + |J| \leq N,
		\\
		\Ebf_c(s, \del^IL^J\phi)^{1/2}
		& \leq C_1\vep s^{\delta},
		\qquad &&&&| I | + |J| \leq N-4.
		\endaligned
	\end{equation}
\end{theorem}

In the proof of the quoted theorem, the exponent $\delta$ is chosen, for instance, in the range $1/50\leq\delta\leq1/20$; in particular, the restriction $\delta<1/9$ imposed below is compatible with that construction.

For linear waves on fixed long-range asymptotically Minkowski
backgrounds, a logarithmic correction of the time-delay variable is used
in \cite{BaskinVasyWunsch2018}.
Inspired by the coordinates introduced in
\cite{Lindblad-3,KauffmanLindblad}, we set
\[
\hat t:=t-2m_{\mathrm S}\ln t,\qquad
\hat x^a:=x^a,\qquad \hat r:=|\hat x|=r,\qquad
\hat q:=\hat t-\hat r .
\]
We approach future null infinity along $(\hat t,\hat x)=(\hat r+\hat q,\hat r\omega)$
as $\hat r\to\infty$ with $\omega\in\mathbb S^2$ and $\hat q\geq1/2$ fixed.
For a scalar function $u$, its radiation field, when the limit exists, is
\begin{equation}\label{eq1-13-oct-2026}
\hat{\Rcal}[u](\omega,\hat{q}) := \lim_{\hat{r}\rightarrow+\infty}\hat r\,u(\hat{r}+\hat{q},\hat{r}\omega).
\end{equation}
In \eqref{eq1-13-oct-2026}, the radiation field is defined by the limit of
$\hat r u$, without differentiation in $\hat q$.
Now we are ready to state the main result.
\begin{theorem}\label{thm1-15-aug-2026}
	Let $(g_{\alpha\beta},\phi)$ be the global solution constructed in
	Theorem~\ref{thm main-PDE} with $0<\delta<1/9$, and take $\eps_0$
	sufficiently small. In the mass-adapted variables we write
	$\hbar(\hat t,\hat x):=\eta^{\alpha\beta}h_{\alpha\beta}(t(\hat t),\hat x)$,
	where $t(\hat t)$ is the inverse time change and the components
	$h_{\alpha\beta}$ are those in the original wave coordinates. Thus the
	trace is evaluated in the new variables, not recomputed from the
	transformed metric components. There exists a non-negative function
	\[
	U\in C\bigl(\overline{B_1(0)}\bigr)
	\]
	such that
	\begin{equation}\label{eq4-19-08-2026}
		\hbar(\hat t,\hat x) = \hat t^{-1}U(\hat x/\hat t)
		+ O_\delta\bigl(C_1\eps\hat t^{-1}\hat q^{-1+9\delta}\bigr),
		\qquad \hat q\geq 2\hat t_{\mathrm{init}}=4-4m_{\mathrm S}\ln2.
	\end{equation}
	The radiation field $\hat{\Rcal}[\hbar](\omega,\hat q)$ exists for
	$\omega\in\mathbb S^2$ and $\hat q\geq1/2$. Its late-time limit is the
	boundary value of $U$:
	\begin{equation}\label{eq:main-final-null-profile}
	U(\omega)=U_\infty(\omega)
	=\lim_{\rho\uparrow1}U(\rho\omega)
	=\lim_{\hat q\to+\infty}\hat{\Rcal}[\hbar](\omega,\hat q).
	\end{equation}
	Both limits are uniform in $\omega$.

	Set
	\[
	\kappa:=1-\frac{11}{2}\delta>0.
	\]
	Then, for every $z\in\overline{B_1(0)}$, the scalar initial energy satisfies the quantitative observability estimate
	\begin{equation}\label{eq:main-quantitative-observability}
	\Ebf_c^{\flat}(2,\phi)
	\lesssim_\delta
	(C_1\eps)^{\frac{2}{\kappa+1}}
	U(z)^{\frac{\kappa}{\kappa+1}},
	\qquad z\in\overline{B_1(0)},
	\end{equation}
	where
	\[
	\Ebf_c^{\flat}(2,\phi)
	:=\int_{|x|<1}\left(\sum_{\alpha=0}^3|\del_\alpha\phi|^2+c^2\phi^2\right)(2,x)\diff x.
	\]
	In particular, if $U(z_0)=0$ for some $z_0\in\overline{B_1(0)}$, then
	\[
	\phi\equiv0\qquad\text{and}\qquad U\equiv0.
	\]
	For $z\in B_1(0)$ the observation is along a timelike ray. For
	$z=\omega\in\mathbb S^2$, \eqref{eq:main-quantitative-observability}
	controls the initial energy from the late-time trace radiation
	profile in the single direction $\omega$.
\end{theorem}
\begin{remark}
To isolate the essential mechanism, we take the initial metric data to
agree exactly with Schwarzschild data outside a compact set. In particular,
the proof uses the resulting exact Schwarzschild exterior. Compact support
of the initial scalar-field data reflects the physical setting of interest:
the scalar field is initially confined to the source region, and we seek to
infer its presence through its effect on the gravitational field rather
than by measuring the scalar field directly. We expect an extension to
noncompact perturbations with sufficiently rapid fall-off to be possible,
but such an extension would require separate estimates for the exterior
tails and is not proved here.
\end{remark}

\paragraph{Main mechanism of the proof.}
The refined Klein--Gordon asymptotics identify a non-negative source
profile $P_0$ on $B_1(0)$. Propagating the leading source
$\hat t^{-3}P_0(\hat x/\hat t)$ through the wave equation gives
\[
U(z)=\frac{1}{4\pi}\int_{B_1(0)}K(z,\hat y)P_0(\hat y)\,\diff\hat y,
\qquad z\in\overline{B_1(0)},
\]
where $K$ is positive and the boundary value of $U$ is the final null
profile in \eqref{eq:main-final-null-profile}. The same source profile
determines the limiting flat Klein--Gordon energy:
\[
E_\infty[\phi]
=\frac{1}{8\pi}\int_{B_1(0)}\frac{P_0(\hat y)}{1-|\hat y|^2}\,\diff\hat y.
\]
A lower bound for $K$ controls the source integral by a single value
of $U$, while the decay of $P_0$ near the boundary controls the more
singular energy integral there. Comparing the two integrals and applying
the inverse energy estimate gives \eqref{eq:main-quantitative-observability}.

\subsection{Relation to the final Bondi mass}

For the Bondi--Sachs description of radiation and mass at null infinity,
we refer to \cite{BondiBurgMetzner1962,Sachs1962}.
For small-data solutions of the Einstein vacuum equations in harmonic
coordinates, the mass at null infinity was studied in
\cite{HeLindblad2024}, where Bondi--Sachs coordinates were also
constructed.

The Bondi-energy identities recalled in this subsection concern the solution
class and normalizations of \cite{IP3}. Whenever they are applied to our
solution, we additionally assume that, after the scalar-field, mass, and
initial wave-coordinate normalizations of that reference, the data satisfy
all hypotheses of \cite[Theorem~1.3]{IP3}, and that the resulting geometric
development is identified with the solution under consideration. This
additional hypothesis is used only for the Bondi interpretation and the
corollary below; it is not an assumption of
Theorem~\ref{thm1-15-aug-2026}.
The fields and coordinates in the following Bondi formulas are those of
the identified solution, in the normalization of \cite{IP3}.
For this solution class, we recall the Bondi energy introduced in
\cite[Eq.~(6.3.42)]{IP3}:
\begin{equation}
E_{\mathrm{Bondi}}(R) = \frac{1}{16\pi}\lim_{t\rightarrow \infty}\int_{S_{R,t}^u}W_j{\bf n}^j\diff \sigma,
\end{equation}
where $S_{R,t}^u := \{x\in \Sigma_t\mid u(t,x) = R\}$. Here $u$ is an approximate optical function constructed in \cite[Lemma~6.14]{IP3},  $W_j = \bar{g}^{ab}(\del_a h_{jb} - \del_jh_{ab})$ with $\bar{g}^{ab}$ the induced metric on $\Sigma_t:=\{(t,x)\in\Mcal\mid x\in\RR^3\}$. The vector ${\bf n}$ is defined as
\begin{equation}
{\bf n}_j:= \del_j u\big(\bar{g}^{ab}\del_au\del_b u\big)^{-1/2},
\end{equation}
which is the unit vector field in $\Sigma'_t = \{x\in\Sigma_t\mid |x|>2^{-8}t\}$, normal to the level sets of the function $u$. In this solution class, it is proved in \cite[Theorem~6.17, Eq.~(6.3.44)]{IP3} that 
\begin{equation}
\aligned
\lim_{R\rightarrow -\infty}E_{\mathrm{Bondi}} =& \frac{1}{16\pi}\|V_{\infty}^{\psi}\|_{L^2}^2,
\\
\lim_{R\rightarrow \infty}E_{\mathrm{Bondi}} =& E_{\mathrm{ADM}},
\endaligned
\end{equation}
Here $V_\infty^\psi$ denotes the nonlinear scattering data of the scalar
field in the normalization of \cite{IP3}. In that setting, the massive
scalar contributes to the energy at future timelike infinity, and the
final Bondi energy comes entirely from the scalar field.

\begin{corollary}[Conditional Bondi consequence]
Assume the hypotheses of Theorem~\ref{thm1-15-aug-2026} and the additional
compatibility and initial-data hypotheses of this subsection. If
$U_\infty(\omega_0)=0$ for some $\omega_0\in\mathbb S^2$, then the final
Bondi energy vanishes.
\end{corollary}
\begin{proof}
Theorem~\ref{thm1-15-aug-2026} gives $\phi\equiv0$. Under the additional
hypotheses, the corresponding scalar scattering profile $V_\infty^\psi$
vanishes, and the final Bondi energy is therefore zero by
\cite[Theorem~6.17]{IP3}.
\end{proof}

Although the late-time trace profile is defined in wave coordinates, Theorem~\ref{thm1-15-aug-2026} shows that its vanishing in even one direction is equivalent to a geometric condition: the scalar field vanishes identically. For a given geometric solution, the criterion therefore gives the same answer in every wave chart satisfying the hypotheses of the theorem, even though the theorem does not establish invariance of the profile itself. A natural next step is to seek a geometric formulation of this single-direction vanishing criterion at future null infinity.

\subsection{Relations to the asymptotic analysis of wave-Klein-Gordon type systems}
For a class of nonlinear Klein--Gordon equations in three space dimensions,
small-data global existence was proved in \cite{Klainerman85}.

We first recall some works on global solutions of wave--Klein--Gordon
systems. In \cite{Georgiev1990,Katayama2012}, global existence for classes
of such systems in three space dimensions is studied. The hyperboloidal
foliation method for wave--Klein--Gordon systems was developed in
\cite{PLF-YM-book1,PLF-YM-CMP}. For a semilinear wave--Klein--Gordon
system, global existence for small compactly supported initial data was
proved in \cite{DongWyatt2020}. Global existence and linear scattering
for a semilinear system with small non-compactly supported initial data
satisfying suitable polynomial decay assumptions were established in
\cite{Zhang2024Noncompact}.

An augmented conformal formulation of $f(R)$ gravity in wave gauge
yields a wave--Klein--Gordon system \cite{PLF-YM-SMF}.

The asymptotic profiles of small solutions of a coupled
wave--Klein--Gordon system are studied in \cite{IP2}. Modified wave
operators for such a system are constructed in \cite{Ouyang2023} for
sufficiently small scattering data in weighted Sobolev spaces.

In the present work, as we will see in Section~\ref{sec1-04-aug-2026},
the mechanism can be understood via the following simplified system:
\begin{equation}
\aligned
&\Box h = -(s/t)^2|\del_t\phi|^2,
\\
&\Box\phi - \phi = 0
\endaligned
\end{equation}
in the light-cone $\{r<t-1\}$ with $C_c^{\infty}$ initial data given in $\{(2,x)\mid |x|<1\}$. It can be shown that in 3-D, $\phi\sim t^{-3/2}$ and $|h|\lesssim t^{-1}$ (for example, by \cite[Proposition~3.1]{Alinhac-2006}). One naturally asks whether this $t^{-1}$ decay rate of $h$ is {\sl optimal}. In \cite{Chen-Lindblad-2023}, a similar quasilinear model system
\begin{equation}
\aligned
-\Box u = (\del_t\phi)^2+\phi^2
\\
-\Box \phi + \phi = H^{\alpha\beta}u\del_{\alpha}\del_{\beta}\phi
\endaligned
\end{equation} 
is discussed. We use our notation $s:=\sqrt{t^2-r^2}$ for the hyperboloidal radius (denoted by $\rho$ in \cite{Chen-Lindblad-2023}) and write $y=x/t$. In the region $\{t-r\geq 4\}$, the pointwise asymptotic profiles take the form
\begin{equation}
\aligned
&u = s^{-1}U(y) + O(t^{-1}(t-r)^{-1+\delta}),
\\
&\phi = s^{-3/2}\big(e^{\mathrm{i}s^*}a_+(y) 
+ e^{-\mathrm{i}s^*}a_-(y)\big) + O(t^{-5/2+\delta}),
\endaligned
\end{equation}
where, along each hyperboloidal ray $y=\mathrm{const.}$, the modified phase is determined by
\begin{equation}
 s^*=\int \left(1-H^{ss}(y)\frac{U(y)}{s}\right)^{-1/2}\,\diff s,
 \qquad
 H^{ss}(y):=H^{\alpha\beta}\del_\alpha s\,\del_\beta s,
\end{equation}
with the normalization corresponding to that in \cite{Chen-Lindblad-2023}. In particular, $s^*-s=O(\eps\ln s)$ along each fixed hyperboloidal ray. Here
\[
|U(y)|\lesssim \eps(1-|y|^2)^{1/2},
\]
and $a_{\pm}$ decay rapidly as $|y|\rightarrow 1$. Chen and Lindblad also identify the limiting value $A(\omega)$ of the wave radiation field with the radial boundary limit
\[
A(\omega)=\lim_{\rho\uparrow1}(1-\rho^2)^{-1/2}U(\rho\omega).
\]

\paragraph{Comparison with earlier asymptotic results.}
The asymptotic mechanism follows the model analysis of Chen and
Lindblad \cite[Theorem~1.1 and Sections~4.2 and~5]{Chen-Lindblad-2023}.
Their work already derives the modified Klein--Gordon phase, constructs
the interior wave profile from the non-oscillatory source, and matches
its boundary value to the late-time radiation field. We adapt this
mechanism to the metric trace in the full Einstein--Klein--Gordon
system: for the modified trace \eqref{eq3'-02-dec-2025}, the leading
scalar oscillations cancel in the weighted source
\eqref{eq:leading-source-cancellation}.

For the full system, Ionescu and Pausader
\cite[Proposition~6.16 and Theorem~6.17]{IP3} describe physical-space
asymptotics using modified scattering profiles and relate the limiting
scalar energy to the final Bondi energy in their solution class.
This total-energy identity does not directly give the late-time
radiation coefficient of the undifferentiated metric trace in a fixed
direction. Our contribution is to represent that coefficient and the
limiting scalar energy through the same non-negative source profile
$P_0$, as in \eqref{eq:positive-kernel-null} and
\eqref{eq:energy-profile-identity}, and thereby control the initial scalar
energy from the late-time trace profile in a single direction.

For a nontrivial scalar field, the estimate also implies that the
$t^{-1}$ decay rate of $\hbar$ is sharp along each timelike ray and
along all sufficiently late outgoing rays of the mass-adapted
coordinates.

Unique continuation from null infinity for linear waves on fixed
asymptotically flat backgrounds is studied in
\cite{AlexakisSchlueShao2016}. Uniqueness for semilinear wave systems
satisfying the null condition is studied in \cite{Yu2023Uniqueness}.
For a rigidity property of the wave component in a wave--Klein--Gordon
system, we refer to \cite{DCDS-Li-Ma}. That work uses a radiation field
defined from the time derivative of the wave component and recovers
vanishing of the wave data. Here we do not differentiate the trace
radiation field in retarded time, and it is the scalar initial energy
that is controlled.

\section{Recall of notation and basic facts}\label{sec1-13-05-2021}
\subsection{Frames and vector fields}
Let $(t,x)\in \RR^{1+3}$ with $x\in \RR^3$. Set $r:= \sqrt{|x^1|^2+|x^2|^2 +|x^3|^2}$. As we consider compact perturbations of the Schwarzschild spacetime, our metric coincides with \eqref{eq Sch-wave} outside the cone $\Kcal := \{r\leq t-1\}$. We thus introduce the restriction
$$
\Hcal_s^* := \Hcal_s\cap\Kcal
$$
of the hyperboloid to the light-cone. Furthermore, we introduce
$$
\Hcal_s := \Big\{(t,x)\in \RR^{1+3}| t = \sqrt{s^2+r^2} \Big\}, \quad \Hcal_{[s_0,s_1]} := \Kcal\cap \Big\{\sqrt{s_0^2+|x|^2}\leq t\leq \sqrt{s_1^2+|x|^2}\Big\},
$$
$$
\del\Kcal = \{t = |x|+1\},\quad \del\Kcal_{[s_0,s_1]} = \del\Kcal\cap \Big\{\sqrt{s_0^2+|x|^2}\leq t\leq \sqrt{s_1^2+|x|^2}\Big\}.
$$
and
$$
\Hinit_{s_0,s_1} := \Hcal_{s_0}^* \cup \del\Kcal_{[s_0,s_1]}.
$$

We recall the following notations introduced in \cite{PLF-YM-book1}:
$$
\delu_0:=\del_t,\quad \delu_a := \delb_a = (x^a/t)\del_t + \del_a.
$$
These vector fields form the semi-hyperboloidal frame which can be defined in $\{r\leq t-1/2\}$. The transition matrices between this frame and the natural frame $\{\del_{\alpha}\}$ are:
\begin{equation}\label{eq semi-frame}
\Phiu_{\alpha}^{\beta} := \left(
\begin{array}{cccc}
1 &0 &0 &0
\\
x^1/t &1 &0 &0
\\
x^2/t &0 &1 &0
\\
x^3/t &0 &0 &1
\end{array}
\right),
\quad
\Psiu_{\alpha}^{\beta} := \left(
\begin{array}{cccc}
1 &0 &0 &0
\\
-x^1/t &1 &0 &0
\\
-x^2/t &0 &1 &0
\\
-x^3/t &0 &0 &1
\end{array}
\right)
\end{equation}
with
$$
\delu_{\alpha} = \Phiu_{\alpha}^{\beta}\del_{\beta},\quad \del_{\alpha} = \Psiu_{\alpha}^{\beta}\delu_{\beta}.
$$

The vector fields $\delu_a$ are tangent to the hyperboloid $\Hcal_s$. We call them hyperbolic derivatives.

Let $T = T^{\alpha\beta}\del_{\alpha}\otimes\del_{\beta}$ be a two-tensor defined in $\Kcal$ or its subset. Then $T$ can be written with $\{\delu_{\alpha}\}$:
$$
T = \Tu^{\alpha\beta} \delu_{\alpha}\otimes\delu_{\beta}
\quad\text{with}\quad 
\Tu^{\alpha\beta} = T^{\alpha'\beta'}\Psiu_{\alpha'}^{\alpha}\Psiu_{\beta'}^{\beta}.
$$
With the above notation, in the semi-hyperboloidal frame we can express the Minkowski metric and its inverse as
\[
\etau_{\alpha\beta} =
\begin{pmatrix}
-1 & -x^1/t & -x^2/t & -x^3/t \\
-x^1/t & 1 - (x^1/t)^2 & -x^1x^2/t^2 & -x^1x^3/t^2 \\
-x^2/t & -x^2x^1/t^2 & 1 - (x^2/t)^2 & -x^2x^3/t^2 \\
-x^3/t & -x^3x^1/t^2 & -x^3x^2/t^2 & 1 - (x^3/t)^2
\end{pmatrix},
\]
\[
\etau^{\alpha\beta} =
\begin{pmatrix}
-s^2/t^2 & -x^1/t & -x^2/t & -x^3/t \\
-x^1/t & 1 & 0 & 0 \\
-x^2/t & 0 & 1 & 0 \\
-x^3/t & 0 & 0 & 1
\end{pmatrix}.
\]
Especially, 
\begin{equation}\label{eq3-20-june-2026}
\aligned
\underline{\gSch}^{00} - \etau^{00} =& -2\mSch\left( \frac{1}{r-\mSch} + \frac{r^2}{t^2(r+\mSch)} \right) 
\\
=& -4\mSch t^{-1} + 2\mSch t^{-1}\Big(\frac{r-t-\mSch}{r-\mSch} + \frac{r(t-r) + t\mSch}{t(r+\mSch)}\Big).
\endaligned
\end{equation}
\subsection{High-order derivatives and homogeneous coefficients}
In the region $\Kcal$, we introduce the following Lorentzian boosts:
$$
L_a = x^a\del_t + t\del_a,\quad a = 1, 2, 3
$$
and the following notation of high-order derivatives: let $I,J$ be multi-indices taking values in $\{0,1,2,3\}$ and $\{1,2,3\}$ respectively,
$$
I = (i_1,i_2,\cdots, i_m),\quad J = (j_1,j_2,\cdots, j_n).
$$
We define
$$
\del^IL^J = \del_{i_1}\del_{i_2}\cdots \del_{i_m}L_{j_1}L_{j_2}\cdots L_{j_n}
$$
to be an $(m+n)$-th-order derivative. 

Let $\mathscr{Z}$ be a family of vector fields. $\mathscr{Z} = \{Z_i|i=0,1,\cdots, 6\}$ with
$$
Z_0 = \del_t,\quad  Z_1=\del_1,\quad Z_2 = \del_2,\quad Z_3 = \del_3,\quad Z_4= L_1,\quad Z_5=L_2,\quad Z_6 = L_3,
$$
A high-order derivative of order $N$ on $\mathscr{Z}$ with multi-index $I = (i_1,i_2,\cdots, i_N)$, $i_j\in\{0,1,\cdots,6\}$ is defined as
$$
\mathscr{Z}^I := Z_{i_1}Z_{i_2}\cdots Z_{i_N}.
$$
A high-order derivative $\mathscr{Z}^I$ is said to be of type $(a,b)$, if it contains at most $a$ partial derivatives and $b$ Lorentzian boosts. We then recall  the following notation introduced in \cite{M-2021-strong}:
$$
\mathcal{I}_{p,k} = \{I| I \text{ is of type }(a,b)\text{ with }a+b \leq p, b\leq k \}
$$
so $\mathscr{Z}^I$ with $I\in \mathcal{I}_{p,k}$ stands for a high-order derivative composed of boosts and partial derivatives. Its order is less than or equal to $p$ and it contains at most $k$ boosts.  We define
\begin{equation}\label{eq1 notation}
\aligned
|u|_{p,k} &:= \max_{K\in \mathcal{I}_{p,k}}|\mathscr{Z}^K u|,\quad &&|u|_p := \max_{0\leq k\leq p}|u|_{p,k},
\\
|\del u|_{p,k} &:= \max_{\alpha}|\del_{\alpha} u|_{p,k}, &&|\del u|_p := \max_{0\leq k\leq p}|\del u|_{p,k},
\\
|\del^m u|_{p,k} &:= \max_{|I|=m}|\del^I u|_{p,k}, &&|\del^m u|_p := \max_{0\leq k\leq p}|\del^m u|_{p,k},
\\
|\delu u|_{p,k} &:= \max\{|\delu_a u|_{p,k}\}, &&|\delu u|_p := \max_{0\leq k\leq p}|\delu u|_{p,k},
\\
|\del\delu  u|_{p,k} &:=\max_{a,\alpha} \{|\delu_a\del_{\alpha} u|_{p,k},|\del_{\alpha}\delu_a u|_{p,k}\},
&&| \del\delu u|_p :=\max_{0\leq k\leq p}| \del\delu u|_{p,k},
\\
|\delu\delu  u|_{p,k} &:=\max_{a,b} \{|\delu_a\delu_b u|_{p,k}\},
&&| \delu\delu u|_p :=\max_{0\leq k\leq p}| \delu\delu u|_{p,k}.
\endaligned
\end{equation}
In the above expressions the $\max$ is taken for $\alpha\in\{0,1,2,3\}$ and $a,b\in\{1,2,3\}$.
These quantities will be used below to control various high-order derivatives. Our first task is to bound them by energy densities.

We also recall the notation of {\bf homogeneous coefficients}.
\begin{definition}[See for example \cite{PLF-YM-book2}]\label{def1-06-dec-2025}
	A {\rm smooth and homogeneous function} of degree $\alpha$ is, by definition, 
	a smooth function $\Phi$ defined in $\{r < t\}$ at least and satisfying 
	\begin{itemize}
		\item $\Phi(\lambda t, \lambda x) = \lambda^{\alpha} \Phi(t, x)$, for a fixed $\alpha \in \RR$ and for all $\lambda>0$,
		
		\item $\sup_{|x| <1}|\del^I\Phi(1, x)|< C_I<+\infty$. 
	\end{itemize}
\end{definition}
We remark that the elements of the transition matrix $\Phiu_{\alpha}^{\beta}$ are smooth and homogeneous of degree $0$.
The following properties of homogeneous functions are direct.
\begin{proposition}\label{prop 1 homo}
	Let $u,v$ be homogeneous of degree $k,l$ respectively. Then
	\begin{itemize}
		\item When $k=l$, $\alpha u + \beta v$ is homogeneous of degree $k$ where $\alpha$ and $\beta$ are constants.
		
		\item $uv$ is homogeneous of degree $k+l$.
		
		\item $\del^IL^J u$ is homogeneous of degree $k-|I|$.
		
		\item There is a positive constant determined by $I,J$ and $u$ such that the following inequality holds in $\Kcal$:
		\begin{equation}\label{eq 1 homo}
		|\del^IL^Ju|\leq Ct^{k-|I|}.
		\end{equation}
	\end{itemize}
\end{proposition}

\subsection{Energies}
Recalling that the nontrivial part of the initial data \eqref{eq main initial-data} is supported in $\{(2,x)\mid |x|<1\}$, Proposition~2.3 of \cite{PLF-YM-book2} implies that the nontrivial part of the metric $g_{\alpha\beta} - \gSch_{\alpha\beta}$ is supported in $\{r<t-1\}$. We thus introduce, for $p \in [1,+\infty)$, the $L^p$ norms on the hyperboloids endowed with the (flat) measure $\diff x$:
$$
\|u\|_{L^p(\Hcal_s^*)}^p := \int_{\Hcal^*_s}|u|^p dx =  \int_{\left\{r\leq \frac{s^2-1}{2}\right\}} \big|u\big(\sqrt{s^2 +r^2}, x\big) \big|^pdx.
$$
Applying the multiplier $-2\del_t \phi$ to $\Boxt_g\phi-c^2\phi=f$ and performing a standard calculation, we obtain
\begin{equation}\label{eq-div-14-dec-2025}
\begin{aligned}
\del_t(-2g^{0\beta}\del_t\phi \del_\beta \phi + g^{\alpha\beta}\del_\alpha\phi \del_\beta \phi+c^2\phi^2)
- 2\del_a({g^{a\beta}\del_t\phi \del_\beta \phi})
\\
=
(\del_tg^{\alpha\beta})\del_\alpha\phi \del_\beta \phi - 2(\del_{\alpha}g^{\alpha\beta})\del_t\phi\del_\beta \phi  
- 2(\del_t\phi) f.
\end{aligned}
\end{equation}
Integrating the above identity in suitable space-time regions leads to energy estimates. Assume that $\phi$ vanishes near $\del\Kcal$, integrating \eqref{eq-div-14-dec-2025} in $\Hcal_{[s_0,s_1]}$, we obtain
\begin{equation}
\Ebf_{g,c}(s_1,\phi) - \Ebf_{g,c}(s_0,\phi) = \int_{s_0}^{s_1}\int_{\Hcal_s^*} (s/t)R_1[h,\phi]\diff x\diff s 
+ \int_{s_0}^{s_1}\int_{\Hcal_s^*}-2(s/t)f\del_t\phi\diff x\diff s,
\end{equation}
with
$$
\aligned
\Ebf_{g,c}(s,\phi) &:=
 \int_{\Hcal^*_s} \Big(-g^{00}|\del_t \phi|^2 + g^{ab} \del_a\phi\del_b\phi + \sum_a \frac{2x^a}{t}g^{a \beta} \del_{\beta}\phi\del_t \phi + c^2\phi^2\Big) \, \diff x,
 \\
R_1[h,\phi] &:= (\del_tg^{\alpha\beta})\del_\alpha\phi \del_\beta \phi - 2(\del_{\alpha}g^{\alpha\beta})\del_t\phi\del_\beta \phi ,
\endaligned
$$
and, for the flat Minkowski background, we apply the notation
$$
\Ebf_{c}(s,\phi) := 
 \int_{\Hcal_s^*} \Big(|\del_t \phi|^2 + \sum_a|\del_a\phi|^2 + \sum_a \frac{2x^a}{t} \del_a\phi\del_t \phi + c^2\phi^2\Big) \, \diff x.
$$
We have the alternative form
$$
\aligned
\Ebf_{c}(s,\phi) & =  \int_{\Hcal_s^*} \Big((s/t)^2 |\del_t \phi|^2 + \sum_a|\delu_a \phi|^2 + c^2\phi^2\Big) \, dx
\\
& =  \int_{\Hcal_s^*} \Big(|\del_t \phi + (x^a/t) \del_a \phi|^2 +\frac{s^2}{t^2}\sum_a|\partial_a\phi|^2 + \sum_{a<b}|t^{-1} \Omega_{ab}\phi|^2 + c^2\phi^2\Big) \, dx, 
\endaligned
$$
where $\Omega_{ab} :=x^a \del_b-x^b\del_a$ denotes the spatial rotations.
When the parameter $c$ is taken to vanish, we also use the short-hand notation
$\Ebf_g(s,\phi) := \Ebf_{g,0}(s,\phi)$, and $\Ebf(s,\phi) := \Ebf_{0}(s,\phi)$.

On the other hand, integrating \eqref{eq-div-14-dec-2025} over the
strip-like domain
\begin{equation}
\{(t,x)\mid t_0\leq t\leq t_1\}
\end{equation}
gives
\begin{equation}
\Ebf^{\flat}_{g,c}(t_1,\phi) - \Ebf^{\flat}_{g,c}(t_0,\phi) = \int_{t_0}^{t_1}\int_{|x|<t-1}R_1[h,\phi]\diff x\diff t 
+ \int_{t_0}^{t_1}\int_{|x|<t-1}-2f\del_t\phi\diff x\diff t,
\end{equation}
with
\begin{equation}
\aligned
\Ebf^{\flat}_{g,c}(t,\phi) 
=&\int_{|x|<t-1}\Big(-2g^{0\beta}\del_t\phi\del_{\beta}\phi 
+ g^{\alpha\beta}\del_{\alpha}\phi\del_{\beta}\phi + c^2\phi^2\Big)(t,x)\diff x
\\
=&\int_{|x|<t-1}\big(-g^{00}|\del_t\phi|^2 + \sum_{a,b=1}^3g^{ab}\del_a\phi\del_b\phi + c^2\phi^2\big)(t,x)\diff x.
\endaligned
\end{equation}
In the flat case, 
\begin{equation}
	\Ebf^{\flat}_{c}(t,\phi)
	:=
	\int_{|x|<t-1}
	\Big(\sum_{\alpha}|\del_{\alpha} \phi|^2
	+c^2\phi^2
	\Big)(t,x)\,\diff x.
\end{equation}
\section{Analysis of the Global Solution of the Einstein-Klein-Gordon System}
\subsection{Global existence, energy and pointwise estimates}
We list the pointwise estimates which were obtained in \cite{PLF-YM-book2} and will be used directly in the following discussion.
\begin{proposition}
Let $(h_{\alpha\beta},\phi)$ be the global solution to \eqref{eq3-29-nov-2025} together with initial data \eqref{eq main initial-data} satisfying \eqref{eq1-07-dec-2025} and \eqref{eq main smallness-condition}. Then the following estimates hold:
\begin{subequations}
\begin{equation}\label{eq3-07-dec-2025}
(s/t)|\del h|_{N-2} + |\delu h|_{N-2}\lesssim_N C_1\eps t^{-3/2}s^{\delta},
\end{equation}
\begin{equation}\label{eq9-03-june-2026}
|\partial\partial h|_{N-4}\lesssim_N C_1\eps(t-r)^{-3/2+\delta}t^{-1+\delta},
\end{equation}
\begin{equation}\label{eq5-20-june-2026}
|\phi|_{N-6} + (s/t)|\del\phi|_{N-6} + |\delu \phi|_{N-6}\leq 
CC_1\eps t^{-3/2}s^{\delta},
\end{equation}
\begin{equation}\label{eq11-20-june-2026}
|\delu \phi|_{N-7}\lesssim_N C_1\eps t^{-5/2}s^{\delta},
\end{equation}
\end{subequations}
\begin{equation}\label{eq7-03-aug-2026}
|\phi| + |\del\phi|\lesssim_N C_1\eps t^{-3/2},
\end{equation}
\begin{equation}\label{eq1-14-dec-2025}
|h|_{N-7}\leq CC_1 \eps t^{-1}s^{C(C_1\eps)^{1/2}} \lesssim C_1\eps t^{-1}s^{\delta},
\end{equation}
\begin{equation}\label{eq4-01-may-2026}
|\delu h|_{N-8}\lesssim C_1\eps t^{-2}s^{\delta},
\end{equation}
\begin{equation}\label{eq7-20-june-2026}
|\del\delu h|_{N-9}\lesssim C_1\eps t^{-3/2}s^{-1+\delta},\quad t|\delu\delu h|_{N-9} \lesssim C_1\eps t^{-2}s^{\delta},
\end{equation}
\begin{equation}\label{eq2-07-dec-2025}
	\aligned
	|\del \hu^{00}|_{N-8}\lesssim& (s/t)^2|\del h|_{N-8} + |\delu h|_{N-8} + t^{-1}|h|_{N-8} 
	+ \sum_{p_1+p_2\leq N-8}|h|_{p_1}|\del h|_{p_2}
	\\
	\lesssim& C_1\eps t^{-2+\delta/2}(t-r)^{1/2+\delta/2},
	\endaligned
\end{equation}
\begin{equation}\label{eq2-20-jan-2026}
	|\hu^{00}|_{N-8}\leq Cm_{\mathrm{S}}t^{-1} + CC_1\eps (s/t)^{2}t^{-1/2}s^{\delta}.
\end{equation}
\end{proposition}
\begin{remark}
Because $h^{\alpha\beta} = g^{\alpha\beta} - \eta^{\alpha\beta}$, we also have, when $C_1\eps$ is sufficiently small,
\begin{equation}\label{eq1-20-jan-2026}
|L^Jh^{\alpha\beta}|\leq CC_1\eps t^{-1}s^{C(C_1\eps)^{1/2}},\quad |J|\leq N-7.
\end{equation}	
\end{remark}
\begin{proof}
These estimates are established in \cite{PLF-YM-book2}. More precisely,
\eqref{eq3-07-dec-2025} follows from \cite[Eqs.~(5.12a), (5.16a)]{PLF-YM-book2},
\eqref{eq9-03-june-2026} from \cite[Eqs.~(8.15), (8.17)]{PLF-YM-book2},
and \eqref{eq5-20-june-2026} from \cite[Eqs.~(5.14), (5.18a), (5.18b)]{PLF-YM-book2}.
Likewise, \eqref{eq11-20-june-2026} and \eqref{eq7-03-aug-2026}
follow respectively from \cite[Eqs.~(5.15), (5.19)]{PLF-YM-book2}
and \cite[Eq.~(11.3)]{PLF-YM-book2}.
For the full mixed norm in \eqref{eq1-14-dec-2025}, use
\cite[Eq.~(11.19)]{PLF-YM-book2} for terms containing only boosts.
For a term of total order at most $N-7$ containing a translation,
commuting that translation past the boosts reduces the estimate to
$|\del h|_{N-8}$. The refined derivative estimate
\cite[Eq.~(10.9)]{PLF-YM-book2}, valid up to order $N-6$, gives
\[
|\del h|_{N-8}\lesssim_N C_1\eps t^{-1/2}s^{-1+A C_1\eps},
\]
where $A$ depends only on the fixed regularity order.
Since $s^2\geq t$ in $\Kcal$, this is bounded by
\[
C C_1\eps t^{-1}s^{A C_1\eps}
\leq C C_1\eps t^{-1}s^{C(C_1\eps)^{1/2}}
\]
for sufficiently small $C_1\eps$. Together with the pure-boost estimate,
this proves \eqref{eq1-14-dec-2025} without weakening its exponent.
On the bounded initial region $s<2$, the same estimate follows from
local small-data bounds.
The final $s^\delta$ bound follows after decreasing the smallness threshold.
Furthermore, we recall the relations:
\begin{equation}
\aligned
\delu_a = t^{-1}L_a,\quad  \del_{\alpha}\delu_a = t^{-1}\del_{\alpha}L_a + t\del_{\alpha}(t^{-1})\delu_a,
\quad
|\delu_a\delu_b h|\lesssim t^{-2}|h|_2.
\endaligned
\end{equation}
Then \eqref{eq4-01-may-2026} follows from \eqref{eq1-14-dec-2025}. Moreover, using $\delu_a=t^{-1}L_a$ and the above product rule, we have
\[
|\del\delu h|_{N-9}
\lesssim t^{-1}|\del h|_{N-8}+t^{-1}|\delu h|_{N-9}
\lesssim C_1\eps t^{-3/2}s^{-1+\delta},
\]
which gives the first estimate in \eqref{eq7-20-june-2026}; the second one follows in the same way from \eqref{eq1-14-dec-2025}. The estimate \eqref{eq2-07-dec-2025} is an application of \cite[Lemma~4.8]{PLF-YM-book2} together with \eqref{eq3-07-dec-2025}; here the order-$N-8$ good-derivative term is controlled by \eqref{eq1-14-dec-2025} using $\delu_a=t^{-1}L_a$. The estimate \eqref{eq2-20-jan-2026} is obtained by integrating \eqref{eq2-07-dec-2025} in the radial direction; see also \cite[Lemma~7.1, Eq.~(7.2)]{PLF-YM-book2}.  
\end{proof}
\subsection{The nonlinear structure of the Einstein equations}
We recall the nonlinear terms in Einstein equations written within a wave gauge, i.e. the first term in the  right-hand side of \eqref{eq3-29-nov-2025}. It is decomposed (see for example \cite[Lemma~3.1]{Lindblad-Rodnianski-05}) into the null terms and quasi-null ones:
\begin{subequations}\label{eq1-02-dec-2025}
	\begin{equation}
		F_{\alpha \beta}(h;\del h, \del h) = P_{\alpha\beta} + Q_{\alpha\beta}
	\end{equation}
	with quasi-null terms
	\begin{equation}\label{eq1-2026-8-20}
		P_{\alpha \beta} := - \frac{1}{2}g^{\lambda \lambda'}g^{\delta \delta'} \del_{\alpha}g_{\delta \lambda'} \del_{\beta}g_{\lambda \delta'}
		+\frac{1}{4}g^{\delta \delta'}g^{\lambda \lambda'} \del_{\beta}g_{\delta \delta'} \del_{\alpha}g_{\lambda \lambda'}
	\end{equation}
	and null terms
	\begin{equation}
		\aligned
		Q_{\alpha \beta} &:= 
		g^{\lambda \lambda'}g^{\delta \delta'} \del_{\delta}g_{\alpha \lambda'} \del_{\delta'}g_{\beta \lambda}
		-g^{\lambda \lambda'}g^{\delta \delta'} \big
		(\del_{\delta}g_{\alpha \lambda'} \del_{\lambda}g_{\beta \delta'} - \del_{\delta}g_{\beta \delta'} \del_{\lambda}g_{\alpha \lambda'} \big)
		\\
		& \quad + g^{\lambda \lambda'}g^{\delta \delta'}
		\big(\del_{\alpha}g_{\lambda'\delta'} \del_{\delta}g_{\lambda \beta} - \del_{\alpha}g_{\lambda \beta} \del_{\delta}g_{\lambda'\delta'} \big)
		+\frac{1}{2}g^{\lambda \lambda'}g^{\delta \delta'}
		\big(\del_{\alpha}g_{\lambda \beta} \del_{\lambda'}g_{\delta \delta'} - \del_{\alpha}g_{\delta \delta'} \del_{\lambda'}g_{\lambda \beta} \big)
		\\
		& \quad + g^{\lambda \lambda'}g^{\delta \delta'}
		\big(\del_{\beta}g_{\lambda'\delta'} \del_{\delta}g_{\lambda \alpha} - \del_{\beta}g_{\lambda \alpha} \del_{\delta}g_{\lambda'\delta'} \big)+\frac{1}{2}g^{\lambda \lambda'}g^{\delta \delta'}
		\big(\del_{\beta}g_{\lambda \alpha} \del_{\lambda'}g_{\delta \delta'} - \del_{\beta}g_{\delta \delta'} \del_{\lambda'}g_{\lambda \alpha} \big).
		\endaligned
	\end{equation}
\end{subequations}
The null terms are bounded in the following manner
\begin{lemma}\label{lem1-20-june-2026}
	Let $(h_{\alpha\beta},\phi)$ be the global solution to \eqref{eq3-29-nov-2025} together with initial data \eqref{eq main initial-data} satisfying \eqref{eq1-07-dec-2025} and \eqref{eq main smallness-condition}. Then 
	\begin{equation}\label{eq1-20-june-2026}
		\aligned
		|Q_{\alpha\beta}|_{N-8}\lesssim_N& (s/t)^2|\del_th|_{N-8}|\del_th|_{[(N-8)/2]} + (C_1\eps)^2t^{-3+\delta}(t-r)^{-1/2+\delta}
		\\
		\lesssim_N& (C_1\eps)^2 t^{-3+\delta}(t-r)^{\delta},
		\endaligned
	\end{equation}
	\begin{equation}\label{eq2-20-june-2026}
		\aligned
		\big|\eta^{\alpha\beta}P_{\alpha\beta}\big|_{N-8}\lesssim_N& (s/t)^2|\del_th|_{N-8}|\del_th|_{[(N-8)/2]} + (C_1\eps)^2t^{-3+\delta}(t-r)^{-1/2+\delta}
		\\
		\lesssim_N& (C_1\eps)^2 t^{-3+\delta}(t-r)^{\delta}.
		\endaligned
	\end{equation}
\end{lemma}
\begin{proof}
	The estimate \eqref{eq1-20-june-2026} follows directly from the null structure. More precisely, $Q_{\alpha\beta}$ can be written as finite linear combinations of the following terms
	\begin{equation}
		(s/t)^2\del_th_{\mu\nu}\del_th_{\mu'\nu'},\quad \delu_{\alpha}h_{\mu\nu}\delu_{\beta}h_{\mu'\nu'},\quad (\alpha,\beta)\neq(0,0)
	\end{equation}
	up to cubic terms. Then we apply \eqref{eq3-07-dec-2025} and \eqref{eq4-01-may-2026}. For the estimate \eqref{eq2-20-june-2026}, we {\bf do not} need the wave gauge conditions. In fact we only need to remark that
	\begin{equation}
		\eta^{\alpha\beta}P_{\alpha\beta} = \etau^{00}\Pu_{00} + \sum_{(\alpha,\beta)\neq(0,0)}\etau^{\alpha\beta}\Pu_{\alpha\beta},
	\end{equation}
	and
	\begin{equation}
		\Pu_{\alpha\beta} = \underline{\Phi}_{\alpha}^{\alpha'}\underline{\Phi}_{\beta}^{\beta'}P_{\alpha'\beta'}
		= - \frac{1}{2}g^{\lambda \lambda'}g^{\delta \delta'} \delu_{\alpha}h_{\delta \lambda'} \delu_{\beta}h_{\lambda \delta'}
		+\frac{1}{4}g^{\delta \delta'}g^{\lambda \lambda'} \delu_{\alpha}h_{\delta \delta'} \delu_{\beta}h_{\lambda \lambda'}
	\end{equation}
	Thus
	\begin{equation}\label{eq1-01-aug-2026}
		\aligned
		&\big|\Pu_{00}\big|_{N-8}\lesssim |\del h|_{N-8}^2,
		\\
		&\big|\Pu_{\alpha\beta}\big|_{N-8}\lesssim (C_1\eps)^2 t^{-3+\delta}(t-r)^{-1/2+\delta},\quad (\alpha,\beta)\neq (0,0).
		\endaligned
	\end{equation}
\end{proof}
\begin{corollary}\label{coro1-2026-8-20}
Let $(h_{\alpha\beta},\phi)$ be the global solution to \eqref{eq3-29-nov-2025} together with initial data \eqref{eq main initial-data} satisfying \eqref{eq1-07-dec-2025} and \eqref{eq main smallness-condition}. Then 
\begin{equation}
|F_{\alpha\beta}|_{N-8}\lesssim (C_1\eps)^2t^{-2+\delta}(t-r)^{-1+\delta}.
\end{equation}
\end{corollary}
\begin{proof}
By \eqref{eq1-2026-8-20},
\bea
\big|P_{\alpha\beta}\big|_{N-8}\lesssim_N |\del_th|_{N-8}|\del_th|_{[(N-8)/2]} \lesssim (C_1\eps)^2t^{-2+\delta}(t-r)^{-1+\delta}.
\eea
Then by \eqref{eq1-20-june-2026}, we arrive at the conclusion.
\end{proof}
We state a direct result of the above estimates.
\begin{lemma}
	Let $(h_{\alpha\beta},\phi)$ be the global solution to \eqref{eq3-29-nov-2025} together with initial data \eqref{eq main initial-data} satisfying \eqref{eq1-07-dec-2025} and \eqref{eq main smallness-condition}. Then 
	\begin{equation}\label{eq12-20-june-2026}
		\big|-(1+4\mSch t^{-1})\del_t^2h_{\alpha\beta}+\Delta_xh_{\alpha\beta}\big|_{N-9}\lesssim_N (C_1\eps)^2t^{-2+2\delta}(t-r)^{-1+3\delta/2}.
	\end{equation}
\end{lemma}
\begin{proof}
We first establish the estimate:
\begin{equation}\label{eq9-20-june-2026}
\big|\Boxt h_{\alpha\beta}\big|_{N-9}\lesssim_N (C_1\epsilon)^2 t^{-2+\delta}(t-r)^{-1+\delta}.
\end{equation}
For this purpose we need
\begin{equation}
|\del_{\alpha}\phi\del_{\beta}\phi|_{N-8} + |V(\phi)g_{\alpha\beta}|_{N-8} + |F_{\alpha\beta}|_{N-8}\lesssim (C_1\eps)^2t^{-2+\delta}(t-r)^{-1+\delta}.
\end{equation}
The term $F_{\alpha\beta}$ is estimated in Corollary \ref{coro1-2026-8-20}.  For the terms containing $\phi$ we apply \eqref{eq5-20-june-2026}. Then we decompose the quasilinear terms as follows:
	\begin{equation}\label{eq6-20-june-2026}
		h^{\mu\nu}\del_{\mu}\del_{\nu}h_{\alpha\beta} 
		= \hu^{00}\del_t\del_t h_{\alpha\beta} 
		+ \sum_{(\mu,\nu)\neq(0,0)}\hu^{\mu\nu}\delu_{\mu}\delu_{\nu}h_{\alpha\beta} 
		+ h^{\mu\nu}\del_{\mu}\Psiu_{\nu}^{\nu'}\delu_{\nu'}h_{\alpha\beta}.
	\end{equation}
	Due to the homogeneity of $\Psiu$ and by applying \eqref{eq7-20-june-2026},
	\begin{equation}
		\Big|\sum_{(\mu,\nu)\neq(0,0)}\hu^{\mu\nu}\delu_{\mu}\delu_{\nu}h_{\alpha\beta} 
		+ h^{\mu\nu}\del_{\mu}\Psiu_{\nu}^{\nu'}\delu_{\nu'}h_{\alpha\beta}\Big|_{N-9}
		\lesssim_N (C_1\eps)^2t^{-3+\delta}(t-r)^{-1/2+\delta}.
	\end{equation}
	However, we still need to decompose $\hu^{00}$. Set
\[
H:=\hu^{00}+4\mSch t^{-1}.
\]
For a fixed $(t,x)\in \Kcal$, we integrate along the radial segment $\{(t,\lambda x)| 1\leq \lambda\leq (t-1)/r\}$. For $|I|+|J|\leq N-8$, this gives
\begin{equation}\label{eq2-02-aug-2026}
\del^IL^JH(t,x) = \del^IL^JH(t,(t-1)x/r) - \int_{r}^{t-1}\del_r\del^IL^JH\big(t,(\rho/r)x\big)\diff \rho.
\end{equation}
Here
\[
\del_r\del^IL^JH
=\del_r\del^IL^J\hu^{00}+4\mSch\del_r\del^IL^J(t^{-1}),
\qquad
|\del_r\del^IL^J(t^{-1})|\lesssim_N t^{-2}.
\]
Thus, by \eqref{eq2-07-dec-2025} and $\mSch\leq C_1\eps$, the radial integral is bounded by
\[
C_1\eps t^{-2+\delta/2}(t-r)^{3/2+\delta/2}
+C_1\eps t^{-2}(t-r),
\]
and the second term is absorbed by the first one in $\Kcal$.
Recall that
\begin{equation}
\aligned
\hu^{00}(t,(t-1)x/r)=\underline{\gSch}^{00}(t,t-1) - \etau^{00}(t,t-1)
= -2\mSch\left( \frac{1}{t-1-\mSch} + \frac{(t-1)^2}{t^2(t-1+\mSch)} \right).
\endaligned
\end{equation}
For the commuted boundary values, the commutations are performed on the exact Schwarzschild expression before restricting to $r=t-1$. Using
\[
L_a(t-r)=-\omega_a(t-r),
\]
together with \eqref{eq3-20-june-2026} and $\mSch\leq C_1\eps$, we obtain, for large $t$,
\begin{equation}\label{eq3-02-aug-2026}
\big|\hu^{00}\big(t,(t-1)x/r\big)  + 4\mSch t^{-1}\big|_{N-8}\lesssim_N C_1\epsilon t^{-2}.
\end{equation}
Combining this boundary estimate with \eqref{eq2-02-aug-2026}, we obtain
\begin{equation}\label{eq8-20-june-2026}
\big|\hu^{00}+4\mSch t^{-1}\big|_{N-8}\lesssim_N C_1\eps t^{-2+\delta/2}(t-r)^{3/2+\delta/2}.
\end{equation}
Then recalling \eqref{eq9-03-june-2026}, we obtain 
\begin{equation}
\big|\big(\hu^{00}+4\mSch t^{-1}\big)\del_t\del_th_{\alpha\beta}\big|_{N-9}\lesssim_N (C_1\eps)^2t^{-3+3\delta/2}(t-r)^{3\delta/2}.
\end{equation}
Then the desired estimate is obtained.
\end{proof}
\section{Further analysis on the metric components}
\subsection{Decomposition of the wave operator }
We rely on the decay estimate by integration along timelike hyperbolas (see \cite{M-2021-strong}). For convenience, we record a modified form suited to our application.

We first recall the following decomposition established in \cite{M-2021-strong} for wave operator in $\RR^{1+n}$:
\begin{equation}\label{decompo-wave-2}
	\aligned
	-\Box u =& (t-r)^{-\beta}t^{-\alpha}\big((s/t)^2\del_t 
	+ (2x^a/t)\delu_a\big)((t-r)^{\beta}t^{\alpha}\del_t u) + p_{n,\alpha,\beta}(t,r)\del_tu 
	\\
	&- \sum_a \delu_a\delu_a u
	\endaligned
\end{equation}
where
$$
p_{n,\alpha,\beta}(t,r) = \big((n-\alpha) - (\alpha+1)(r/t)^2 - \beta(t-r)t^{-1}\big)t^{-1}.
$$
In order to keep $p_{n,\alpha,\beta}$ positive, one needs:
$$
\alpha \leq \frac{n-1}{2},\quad \beta\leq n-\alpha.
$$
In our case, $n=3$, we set $\alpha=1$ and $1/2<\beta< 2.$ Thus
$$
p_{3,1,\beta}(t,r)=(2-\beta+2\frac{r}{t})\frac{t-r}{t^2} \simeq_\beta \frac{s^2}{t^2}t^{-1}.
$$
Then \eqref{decompo-wave-2} is written as
\begin{equation}\label{eq2-19-08-2022-M}
	\mathcal{J}\big((t-r)^{\beta}t \del_tu\big) + P\ \big((t-r)^{\beta}t \del_tu\big) = S^w[u] + \Delta^w[u]	
\end{equation}
with
$$
\aligned
&\mathcal{J} := \del_t + \frac{2tx^a}{t^2+r^2}\del_a,
\\
&P(t,r) := \frac{t^2}{t^2+r^2}p_{3,1,\beta}(t,r),
\\
&S^w[u] := \frac{t^3(t-r)^{\beta}(-\Box u)}{t^2+r^2},\quad 
\Delta^w[u] := \frac{t^3(t-r)^{\beta}\sum_a\delu_a\delu_a u}{t^2+r^2}.
\endaligned
$$
Let $(t,x)\in\Kcal$ and let $\zeta(\cdot;t,x)$ be the integral curve of $\Jcal$ through $(t,x)$. As observed in \cite{M-2021-strong}, for $r>0$ the quantity
\begin{equation}\label{eq13-19-08-2022-M}
	\frac{t^2-r^2}{r}=c_0
\end{equation}
is constant along $\zeta$, so these integral curves are timelike hyperbolas. Tracing such a curve backward in time, it meets $\Hcal_2^*\cup\del\Kcal$ at a unique point. More precisely, the intersection lies on $\del\Kcal$ when $c_0<8/3$, on $\Hcal_2^*$ when $c_0>8/3$, and at $\Hcal_2^*\cap\del\Kcal$ when $c_0=8/3$. When $r=0$, the integral curve is the vertical line $x=0$ and meets $\Hcal_2^*$.
\subsection{Improved decay bound on the metric components}
Let us denote by $U_{t,x}(\tau) := t(t-r)\del_t\del^IL^J h\big|_{\zeta(\tau;t,x)}$ (i.e., fix $\beta = 1$). Then \eqref{eq2-19-08-2022-M} leads to
\begin{equation}\label{eq12-19-08-2022-M}
	U_{t,x}'(\tau) + P_{t,x}(\tau) U_{t,x}(\tau) = \big(S^w[\del^IL^J h] + \Delta^w[\del^IL^J h]\big)\big|_{\zeta(\tau;t,x)}.
\end{equation}
Here $P_{t,x}(\tau) := P|_{\zeta(\tau;t,x)}$.

Then by integrating the above ODE, we have, for $(t,x)\in \Kcal$,
\begin{equation}\label{eq21-19-08-2022-M}
	U_{t,x}(t) = U_{t,x}(t_0)e^{-\int_{t_0}^tP_{t,x}(\eta) d\eta} 
	+ \int_{t_0}^t\big(S^w[\del^IL^J h] + \Delta^w[\del^IL^J h]\big)\big|_{\zeta(\tau;t,x)} e^{-\int_{\tau}^tP_{t,x}(\eta) d\eta}d\tau. 
\end{equation}

Then we establish the following result.
\begin{proposition}\label{prop1-18-dec-2025}
	Let $(h_{\alpha\beta},\phi)$ be the global solution to \eqref{eq3-29-nov-2025} together with initial data \eqref{eq main initial-data} satisfying \eqref{eq1-07-dec-2025} and \eqref{eq main smallness-condition}. Then the following estimates hold:
	\begin{equation}\label{eq2-2026-7-10}
		|\del h|_{N-9}\lesssim_{\delta,N} C_1\eps t^{-1+7\delta/2}(t-r)^{-1}.
	\end{equation}
\end{proposition}
\begin{proof}
	We apply \eqref{eq12-19-08-2022-M} -- \eqref{eq21-19-08-2022-M} on \eqref{eq12-20-june-2026}, where we remark that
	\begin{equation}
		|\mSch t^{-1}\del_t\del_th|_{N-9}\lesssim (C_1\eps)^2 t^{-2+\delta}(t-r)^{-3/2+\delta}.
	\end{equation}
	Thus combined with \eqref{eq12-20-june-2026}, we obtain
	\begin{equation}\label{eq1-02-aug-2026}
	|\Box h_{\alpha\beta}|_{N-9}\lesssim (C_1\eps)^2 t^{-2+2\delta}(t-r)^{-1+3\delta/2}.
	\end{equation}
	
	Then for $|I|+|J|\leq N-9$,
	\begin{equation}
		\big|S^w[\del^IL^J h] + \Delta^w[\del^IL^J h]\big|\lesssim_NC_1\eps t^{-1+7\delta/2}.
	\end{equation}
	We also remark that 
	\begin{equation}
		|U_{t,x}(t_0)|\lesssim \sup_{(t,x)\in\Hcal_2^*\cup\{r=t-1\}}|t(t-r)\del_t\del^IL^J h_{\alpha\beta}|\lesssim C_1\eps.
	\end{equation}
	Thus
	we obtain the desired estimate.
\end{proof}

\begin{corollary}
	Let $(h_{\alpha\beta},\phi)$ be the global solution to
	\eqref{eq3-29-nov-2025}, arising from the initial data
	\eqref{eq main initial-data} satisfying
	\eqref{eq1-07-dec-2025} and \eqref{eq main smallness-condition}.
	Then
	\begin{equation}\label{eq10-20-june-2026}
		|\eta^{\alpha\beta}F_{\alpha\beta}|_{N-9}
		\lesssim_{\delta,N}
		(C_1\epsilon)^2t^{-3+7\delta}(t-r)^{-1},
	\end{equation}
	\begin{equation}\label{eq1-2026-7-10}
		|\del_t\del_t h_{\alpha\beta}|_{N-10}
		\lesssim_{\delta,N} CC_1\eps t^{-1+\frac{7}{2}\delta}(t-r)^{-2}.
	\end{equation}
\end{corollary}

\begin{proof}
	To prove \eqref{eq10-20-june-2026}, we simply substitute the estimate
	\eqref{eq2-2026-7-10} into the expressions for
	$Q_{\alpha\beta}$ and $P_{\alpha\beta}$.

	The estimate \eqref{eq1-2026-7-10} is due to \cite[Lemmas~8.1 and~8.2]{PLF-YM-book2}. However, here we apply a clearer form from \cite[Proposition~2.4]{M-2021-strong}:
	\begin{equation}
	(s/t)^2|\del\del u|_{p,k}\lesssim |\Box u|_{p,k} + t^{-1}|\del u|_{p+1,k+1}.
	\end{equation}
	Then we apply \eqref{eq2-2026-7-10} and \eqref{eq1-02-aug-2026}.
\end{proof}

\begin{corollary}
Let $(h_{\alpha\beta},\phi)$ be the global solution to \eqref{eq3-29-nov-2025}, arising from the initial data
\eqref{eq main initial-data} satisfying \eqref{eq1-07-dec-2025} and \eqref{eq main smallness-condition}.
Then
\begin{equation}\label{eq4-02-aug-2026}
\big|\big(\hu^{00}+4\mSch t^{-1}\big)\del_t\del_th_{\alpha\beta}\big|_{N-10}
\lesssim (C_1\eps)^2t^{-3+7\delta}(t-r)^{-1}.
\end{equation}
\end{corollary}
\begin{proof}
Applying the same wave-coordinate estimate at order $N-9$ and combining it with \eqref{eq2-2026-7-10}, we obtain
\begin{equation}\label{eq:del-hu00-N9}
|\del \hu^{00}|_{N-9}\lesssim_{\delta,N} C_1\eps t^{-2+7\delta/2}.
\end{equation}
Applying the radial identity \eqref{eq2-02-aug-2026} at order $N-9$, together with
\[
|\del_r\del^IL^J(t^{-1})|\lesssim_N t^{-2},
\]
and using \eqref{eq3-02-aug-2026} for the boundary term, we obtain
\begin{equation}\label{eq-2026-8-4}
\big|\big(\hu^{00}+4\mSch t^{-1}\big)\big|_{N-9}\lesssim_{\delta,N}C_1\eps t^{-2+7\delta/2}(t-r).
\end{equation}
Then combined with \eqref{eq1-2026-7-10}, we obtain \eqref{eq4-02-aug-2026}.
\end{proof}

\section{Further analysis on the scalar field}
\subsection{Inverse energy estimates}
We introduce 
\begin{equation}
\mathcal{D}_{t_0} := \Big\{(t,x)|t_0\leq t\leq \sqrt{t_0^2+r^2}\Big\}\cap \Kcal.
\end{equation}
\begin{proposition}[Inverse energy estimate under hyperboloidal foliation]\label{prop1-20-jan-2026}
Let $(h_{\alpha\beta},\phi)$ be the global solution to \eqref{eq3-29-nov-2025} together with initial data \eqref{eq main initial-data} satisfying \eqref{eq1-07-dec-2025} and \eqref{eq main smallness-condition}. Then there exists a constant $A >1$ such that
\begin{equation}
A^{-1}\Ebf_c(2,\phi) \leq  \Ebf_c(s,\phi) \leq A\Ebf_c(2,\phi),\quad \forall s\geq 2. 
\end{equation}
\end{proposition}
\begin{proof}
We recall
\begin{equation}\label{eq2-12-june-2026}
R_1[h,\phi]:= (\del_t g^{\alpha\beta})\del_{\alpha}\phi \del_{\beta}\phi-2(\del_{\alpha}g^{\alpha\beta})\del_t\phi \del_\beta\phi .
\end{equation}
Integrating equation \eqref{eq-div-14-dec-2025} in $\Hcal_{[2,s_1]}$ and using Stokes' formula and the fact that $\phi$ vanishes near $\del\Kcal$, we obtain
\begin{equation}\label{eq1-28-june-2026}
\Ebf_{g,c}(s_1,\phi) - \Ebf_{g,c}(2,\phi) = \int_{2}^{s_1}\int_{\Hcal^*_s} (s/t)R_1[h,\phi] \diff x \diff s.
\end{equation}
In order to estimate the right-hand side integral, we rewrite $R_1$ in the semi-hyperboloidal frame. Recall that
$$\delu_{\alpha} = \Phiu_{\alpha}^{\beta}\del_{\beta},\quad \del_{\alpha} = \Psiu_{\alpha}^{\beta}\delu_{\beta},$$
and 
$$h^{\alpha \beta}=\Phiu_{\alpha^\prime}^{\alpha}\Phiu_{\beta^\prime}^{\beta}\hu^{\alpha^{\prime}\beta^{\prime}},$$
Recall $g^{\alpha\beta}= \eta^{\alpha\beta}+h^{\alpha \beta},$ for the first term in $R_1$ we have
\begin{equation}
		\aligned
		(\del_t g^{\alpha\beta})\del_{\alpha}\phi \del_{\beta}\phi
		&=(\del_t h^{\alpha\beta})\del_{\alpha}\phi \del_{\beta}\phi
		\\
		&=(\del_t \hu^{00})\del_t\phi \del_t\phi 
		+
		\sum_{(\alpha, \beta)\neq (0,0)}(\del_t \hu^{\alpha \beta})\delu_{\alpha}\phi \delu_{\beta}\phi 
		+
		\del_t (\Phiu_{\alpha^\prime}^{\alpha}\Phiu_{\beta^\prime}^{\beta})\hu^{\alpha^{\prime}\beta^{\prime}}\del_{\alpha}\phi \del_{\beta}\phi.
		\endaligned
	\end{equation}
	We begin to estimate the integral of each term on the right-hand side. By \eqref{eq2-07-dec-2025}, we obtain, on any time interval $[2,s_1]$,
	\begin{equation}
		\aligned
		\int_{2}^{s_1}\int_{\Hcal^*_s} |(\del_t \hu^{00})\del_t\phi \del_t\phi | (s/t) \diff x \diff s
		&\leq 
		\int_{2}^{s_1} CC_1\eps s^{-3/2+3\delta/2} \| (s/t)\del_t\phi \|^2_{L^2(\Hcal^*_s)} \diff s
		\\
		&\leq
		\int_{2}^{s_1} CC_1\eps s^{-3/2+3\delta/2} \Ebf_{c}(s,\phi) \diff s.
		\endaligned
	\end{equation}
	Since $\Phiu_\alpha^\beta$ and $\Psiu_\alpha^\beta$ are smooth and homogeneous of degree $0$, by Proposition \ref{prop 1 homo}, we have
	$$
	|\del_t (\Phiu_{\alpha^\prime}^{\alpha}\Phiu_{\beta^\prime}^{\beta})|
	+
	|\del_t (\Psiu_{\alpha^\prime}^{\alpha}\Psiu_{\beta^\prime}^{\beta})|
	\leq 
	Ct^{-1}.
	$$
	Therefore by \eqref{eq1-14-dec-2025} the third term
	\begin{equation}
		\aligned
		\int_{2}^{s_1}\int_{\Hcal^*_s}|\del_t (\Phiu_{\alpha^\prime}^{\alpha}\Phiu_{\beta^\prime}^{\beta})\hu^{\alpha^{\prime}\beta^{\prime}}\del_{\alpha}\phi \del_{\beta}\phi| (s/t)\diff x \diff s
		&\leq
		\int_{2}^{s_1} CC_1\eps s^{-3/2+3\delta/2} \| (s/t)\del \phi \|^2_{L^2(\Hcal^*_s)} \diff s
		\\
		&\leq
		 \int_{2}^{s_1} CC_1\eps s^{-3/2+3\delta/2} \Ebf_{c}(s,\phi) \diff s,
		\endaligned
	\end{equation}
	since $\hu^{\alpha^{\prime}\beta^{\prime}}=\Psiu_{\alpha}^{\alpha^\prime}\Psiu_{\beta}^{\beta^\prime}h^{\alpha\beta}.$ 
	In a similar way, by \eqref{eq3-07-dec-2025}, we have
	$$
	|\del_t \hu^{\alpha \beta}|\leq CC_1\eps (s/t)^{-1}t^{-3/2}s^{\delta},
	$$
	notice that the second term contains at least one ``good'' derivative, hence
	\begin{equation}
		\int_{2}^{s_1}\int_{\Hcal^*_s} (s/t)
		\big|\sum_{(\alpha, \beta)\neq (0,0)}(\del_t \hu^{\alpha \beta})\delu_{\alpha}\phi \delu_{\beta}\phi \big| \diff x \diff s
		\leq
		\int_{2}^{s_1} CC_1\eps s^{-3/2+3\delta/2} \Ebf_{c}(s,\phi) \diff s.
	\end{equation}
	Gathering the above results, we obtain the estimate for the integral of the first term in $R_1$:
	\begin{equation}
		\int_{2}^{s_1}\int_{\Hcal^*_s} (s/t)|(\del_t g^{\alpha\beta})\del_{\alpha}\phi \del_{\beta}\phi| \diff x \diff s 
		\leq
		\int_{2}^{s_1} CC_1\eps s^{-3/2+3\delta/2} \Ebf_{c}(s,\phi) \diff s.
	\end{equation}
	By a similar argument, we obtain
	\begin{equation}
		(\del_{\alpha}g^{\alpha\beta})\del_t\phi \del_\beta\phi
		=
		(\del_t \hu^{00})\del_t\phi \del_t\phi
		+
		\sum_{(\alpha,\beta)\neq (0,0)} \delu_{\alpha}\hu^{\alpha\beta}\del_t\phi \delu_{\beta}\phi
		+
		\del_{\alpha} (\Phiu_{\alpha^\prime}^{\alpha}\Phiu_{\beta^\prime}^{\beta})\hu^{\alpha^{\prime}\beta^{\prime}}\del_{t}\phi \del_{\beta}\phi
	\end{equation}
	and the estimate for the integral of the second term in $R_1$
	\begin{equation}\label{eq2-28-june-2026}
		\int_{2}^{s_1}\int_{\Hcal^*_s} (s/t)|(\del_{\alpha}g^{\alpha\beta})\del_t\phi \del_\beta\phi| \diff x \diff s 
		\leq
		\int_{2}^{s_1} CC_1\eps s^{-3/2+3\delta/2} \Ebf_{c}(s,\phi) \diff s.
	\end{equation}
	Then based on \eqref{eq1-28-june-2026} and \eqref{eq2-28-june-2026},
	\begin{equation}\label{eq2-14-dec-2025}
	\big|\Ebf_{g,c}(s_1,\phi) - \Ebf_{g,c}(2,\phi)\big| \lesssim C_1\eps\int_{2}^{s_1}  s^{-3/2+3\delta/2} \Ebf_{c}(s,\phi) \diff s.
	\end{equation}
	
	On the other hand, recall that
	\begin{equation}\label{eq3'-14-dec-2025}
		\aligned
		&\big|\Ebf_{g,c}(s,\phi) - \Ebf_c(s,\phi)\big|
		\\
		=&
		\Big| \int_{\Hcal^*_s}\big(h^{\alpha \beta}\del_\alpha \phi \del_\beta \phi +2\sum_{a}\frac{x^a}{t}h^{a\beta}\del_t \phi \del_{\beta} \phi-2h^{0\beta}\del_t \phi \del_{\beta}\phi \big)   \diff x \Big|
		\\
        =&\Big| \int_{\Hcal^*_s}\big(\hu^{\alpha \beta}\delu_\alpha \phi \delu_\beta \phi +2\sum_{a}\frac{x^a}{t}\hu^{\alpha^{\prime}\beta^{\prime}}\Phiu_{\alpha^{\prime}}^{a}\Phiu_{\beta^{\prime}}^{\beta}\del_t \phi \del_{\beta} \phi-2\hu^{\alpha^{\prime}\beta^{\prime}}\Phiu_{\alpha^{\prime}}^{0}\Phiu_{\beta^{\prime}}^{\beta}\del_t \phi \del_{\beta}\phi \big)   \diff x \Big|
        \\
        =&\left|\int_{\Hcal^*_s}-(t/s)^2\hu^{00}(s/t \del_t \phi)^2+\hu^{ab}\delu_a \phi \delu_b \phi \diff x\right|
		\leq
		CC_1 \eps \Ebf_c(s,\phi),
		\endaligned
	\end{equation}
	in which we applied \eqref{eq2-20-jan-2026}, \eqref{eq1-14-dec-2025}, and the fact $t\leq s^2$ in the region $\Kcal$.
When $C_1\eps$ is sufficiently small, we have
\begin{equation}\label{eq3-14-dec-2025}
		\aligned
		\big|\Ebf_{g,c}(s,\phi) - \Ebf_c(s,\phi)\big| \leq (1/2)\Ebf_c(s,\phi),
		\endaligned
	\end{equation}
Now we consider \eqref{eq2-14-dec-2025} and \eqref{eq3-14-dec-2025} together to obtain
\begin{equation}\label{eq21-june-2026}
|\Ebf_{g,c}(s_1,\phi) - \Ebf_{g,c}(2,\phi)|\lesssim C_1\eps\int_2^{s_1} s^{-3/2+3\delta/2} \Ebf_{g,c}(s,\phi) \diff s,
\end{equation}
hence
\begin{equation}\label{eq1-29-jan-2026}
\Ebf_{g,c}(s_1,\phi) \leq \Ebf_{g,c}(2,\phi) + CC_1\eps\int_2^{s_1}  s^{-3/2+3\delta/2} \Ebf_{g,c}(s,\phi) \diff s.
\end{equation}
Then by Grönwall's inequality we obtain
\begin{equation}
\Ebf_{g,c}(s_1,\phi)\leq e^{CC_1\eps\int_2^{s_1}\tau^{-3/2+3\delta/2}\diff \tau} \Ebf_{g,c}(2,\phi)
\lesssim \Ebf_{g,c}(2,\phi).
\end{equation}
Plug this back into \eqref{eq21-june-2026} and provided that $C_1\eps$ is sufficiently small,
\begin{equation}\label{eq3-21-june-2026}
(4/5)\Ebf_{g,c}(2,\phi)\leq \Ebf_{g,c}(s_1,\phi)\leq (5/4)\Ebf_{g,c}(2,\phi).
\end{equation}
Combining this with \eqref{eq3-14-dec-2025}, we obtain the desired estimate.
\end{proof}

\begin{corollary}[Inverse energy estimate under flat foliation]\label{coro1-28-5-2026}
 Let $(h_{\alpha\beta},\phi)$ be the global solution to \eqref{eq3-29-nov-2025} together with initial data \eqref{eq main initial-data} satisfying \eqref{eq1-07-dec-2025} and \eqref{eq main smallness-condition}. Then there exists a constant $A>1$ such that
\begin{equation}\label{eq6-21-june-2026}
A^{-1} \Ebf^{\flat}_{c}(2,\phi) \leq \Ebf^{\flat}_c(t,\phi) \leq A \Ebf^{\flat}_{c}(2,\phi),\qquad \forall\, t\geq2.
\end{equation}
\end{corollary}
\begin{proof}
Similar to \eqref{eq3'-14-dec-2025}, we have
\bea\label{eq1-30-6-2026}
\big|\Ebf^{\flat}_{g,c}(t,\phi) - \Ebf^{\flat}_c(t,\phi)\big| =\left|\int_{|x|<t-1}\big(-h^{00}|\del_t\phi|^2 + \sum_{a,b=1}^3h^{ab}\del_a\phi\del_b\phi  \big)\diff x\right|\leq CC_1 \eps \Ebf^{\flat}_c(t,\phi).
\eea
When $2\leq t_0\leq5/2$, integrating \eqref{eq-div-14-dec-2025} over
$\{(t,x):2\leq t\leq t_0\}$ and using the local estimates together with
\eqref{eq1-30-6-2026}, we obtain
\begin{equation}\label{eq4-28-june-2026}
\big|\Ebf^{\flat}_{g,c}(t_0,\phi) - \Ebf^{\flat}_{g,c}(2,\phi)\big| \lesssim \int_2^{t_0}\int_{\{r<t-1\}}|R_1[h,\phi]|\diff x\diff t
\lesssim C_1\eps\int_2^{t_0}\Ebf^{\flat}_{g,c}(t,\phi)\diff t,
\end{equation}
which, by Gr\"onwall's inequality, gives
\begin{equation}
\Ebf^{\flat}_{g,c}(t_0,\phi)\lesssim \Ebf^{\flat}_{g,c}(2,\phi).
\end{equation}
Then plugging back to \eqref{eq4-28-june-2026}, we obtain, for $2\leq t_0\leq 5/2$,
\begin{equation}\label{eq5-28-june-2026}
A_1^{-1} \Ebf^{\flat}_{g,c}(2,\phi)\leq\Ebf^{\flat}_{g,c}(t_0,\phi)\leq A_1\Ebf^{\flat}_{g,c}(2,\phi).
\end{equation}

On the other hand, we integrate \eqref{eq-div-14-dec-2025} in the region $\mathcal{D}_{t_0}$ and obtain
\begin{equation}\label{eq7-21-june-2026}
\Ebf_{g,c}(t_0,\phi) - \Ebf_{g,c}^{\flat}(t_0,\phi) = \int_{\mathcal{D}_{t_0}} R_1[h,\phi]\,\diff x\diff t
\end{equation}
with $R_1[h,\phi]$ defined by \eqref{eq2-12-june-2026}. For $t_0=2$,
\begin{equation}
\big|\Ebf_{g,c}(2,\phi) - \Ebf^{\flat}_{g,c}(2,\phi)\big|\lesssim C_1\eps\int_2^{5/2}\Ebf^{\flat}_{g,c}(t,\phi)\diff t
\lesssim C_1\eps \Ebf^{\flat}_{g,c}(2,\phi),
\end{equation}
where for the last inequality we have applied \eqref{eq5-28-june-2026}. This leads to
\begin{equation}\label{eq6-28-june-2026}
A_2^{-1}\Ebf^{\flat}_{g,c}(2,\phi)\leq \Ebf_{g,c}(2,\phi)\leq A_2\Ebf^{\flat}_{g,c}(2,\phi).
\end{equation}

Again, we recall \eqref{eq7-21-june-2026}. This leads to, when $t_0\geq 5/2$,
\begin{equation}
	\big|\Ebf_{g,c}(t_0,\phi) - \Ebf_{g,c}^{\flat}(t_0,\phi) \big|\leq \int_{\mathcal{D}_{t_0}}\big|R_1[h,\phi]\big|\,\diff x\diff t
	\leq \int_{\sqrt{2t_0-1}}^{t_0}\int_{\Hcal^*_s}(s/t)\big|R_1[h,\phi]\big|\diff x\diff s.
\end{equation}
Recalling \eqref{eq3-14-dec-2025}, \eqref{eq21-june-2026} and \eqref{eq3-21-june-2026},
\begin{equation}
	\aligned
	\big|\Ebf_{g,c}(t_0,\phi) - \Ebf_{g,c}^{\flat}(t_0,\phi) \big|\lesssim& C_1\eps \int_{\sqrt{2t_0-1}}^{t_0}s^{-3/2+2\delta}\Ebf_{g,c}(s,\phi)\diff s
	\\
	\lesssim& C_1\eps \Ebf_{g,c}(2,\phi)\lesssim C_1\eps\Ebf_{g,c}(t_0,\phi).
	\endaligned
\end{equation}
Provided that $C_1\eps$ is sufficiently small and $t\geq 5/2$,
\begin{equation}\label{eq4-21-june-2026}
	(4/5)\Ebf_{g,c}(t,\phi)\leq\Ebf^{\flat}_{g,c}(t,\phi)\leq (5/4)\Ebf_{g,c}(t,\phi)
\end{equation}
which, combined with \eqref{eq3-21-june-2026}, leads to
\begin{equation}
	(16/25)\Ebf_{g,c}(2,\phi)\leq\Ebf^{\flat}_{g,c}(t,\phi)\leq (25/16)\Ebf_{g,c}(2,\phi),\quad t\geq 5/2.
\end{equation}
Then by \eqref{eq6-28-june-2026}, there exists $A_3>1$ such that
\begin{equation}
A_3^{-1}\Ebf^{\flat}_{g,c}(2,\phi)\leq \Ebf^{\flat}_{g,c}(t,\phi)\leq A_3\Ebf^{\flat}_{g,c}(2,\phi),\quad t\geq 5/2.
\end{equation}
Combined with \eqref{eq5-28-june-2026}, there exists $A_4>1$ such that
\begin{equation}
A_4^{-1}\Ebf^{\flat}_{g,c}(2,\phi)\leq \Ebf^{\flat}_{g,c}(t,\phi)\leq A_4\Ebf^{\flat}_{g,c}(2,\phi),\quad \forall\, t\geq 2.
\end{equation}
Finally, we apply a calculation similar to \eqref{eq3'-14-dec-2025},
\begin{equation}
A_5^{-1}\Ebf^{\flat}_c(t,\phi)\leq\Ebf^{\flat}_{g,c}(t,\phi)\leq A_5\Ebf^{\flat}_c(t,\phi),\quad \forall\, t\geq 2.
\end{equation}
Then \eqref{eq6-21-june-2026} is established.

\end{proof}

\subsection{An \texorpdfstring{$L^2$}{L2} estimate for the semi-hyperboloidal derivatives}
We record the following direct consequence of the pointwise estimates for the scalar field.
\begin{lemma}\label{lem:hyperbolic-gradient-decay}
Let $(h_{\alpha\beta},\phi)$ be the global solution to \eqref{eq3-29-nov-2025} together with initial data \eqref{eq main initial-data} satisfying \eqref{eq1-07-dec-2025} and \eqref{eq main smallness-condition}. Then, for $a=1,2,3$,
\begin{equation}\label{eq:hyperbolic-gradient-L2}
\int_{|x|<t-1}|\delu_a\phi|^2(t,x)\diff x
\lesssim (C_1\eps)^2t^{-2+2\delta}.
\end{equation}
\end{lemma}
\begin{proof}
By \eqref{eq11-20-june-2026}, and since $s\leq t$ in $\Kcal$,
\[
|\delu_a\phi|^2\lesssim (C_1\eps)^2t^{-5+2\delta}.
\]
The volume of $\{|x|<t-1\}$ is bounded by $Ct^3$, and the conclusion follows.
\end{proof}
\section{The modified trace component}
\subsection{Equation for the modified trace component}
We introduce
\begin{equation}\label{eq3'-02-dec-2025}
	u := \hbar + 12\pi\phi^2 = \eta^{\alpha\beta}h_{\alpha\beta} + 12\pi\phi^2.
\end{equation}
This component satisfies the following wave equation:
\begin{equation}\label{eq1-25-may-2026}
	\big(1-\hu^{00}\big)\del_t\del_t u - \Delta u = \mathbb{S}[\phi] + \mathbb{F}[h] + \mathbb{T}[h,\phi],
\end{equation}
where
\begin{equation}
	\mathbb{S}[\phi] := 8\pi\Big((s/t)^2|\del_t\phi|^2+ c^2\phi^2\Big),
	\quad
	\mathbb{F}[h] := -\eta^{\alpha\beta}F_{\alpha\beta},
\end{equation}
and
\begin{equation}\label{eq4-20-june-2026}
	\aligned
	\mathbb{T}[h,\phi] :=& 16\pi\Big((x^a/t)\del_t\phi\delu_a\phi - \frac{1}{2}\sum_a|\delu_a\phi|^2 + V(\phi)\eta^{\alpha\beta}h_{\alpha\beta} - (3/2)h^{\alpha\beta}\del_{\alpha}\phi\del_{\beta}\phi \Big)
	\\
	&+ \sum_{(\alpha, \beta)\neq (0,0)}\hu^{\alpha\beta}\delu_{\alpha}\delu_{\beta}u 
	+ h^{\alpha\beta}\del_{\alpha}\Psiu_{\beta}^{\nu}\delu_{\nu}u.
	\endaligned
\end{equation}
This can be found by direct calculation. In fact we first remark that
\begin{subequations}\label{eq5-20-jan-2026}
	\begin{equation}\label{eq5a-20-jan-2026}
		-\Boxt_g u = 8\pi\Big((s/t)^2|\del_t\phi|^2+ c^2\phi^2\Big)
		-\eta^{\alpha\beta}F_{\alpha\beta}  + 16\pi \mathbb{C}[h,\phi]
	\end{equation}
	with
	\begin{equation}\label{eq5b-20-jan-2026}
		\mathbb{C}[h,\phi] :=(x^a/t)\del_t\phi\delu_a\phi - \frac{1}{2}\sum_a|\delu_a\phi|^2 + V(\phi)\eta^{\alpha\beta}h_{\alpha\beta} - (3/2)h^{\alpha\beta}\del_{\alpha}\phi\del_{\beta}\phi.
	\end{equation}
\end{subequations}
To see this, we recall
$$
R_{\alpha\beta} = 8\pi(\nabla_{\alpha}\phi\nabla_{\beta}\phi + V(\phi)g_{\alpha\beta})
$$
and
$$
R_{\alpha\beta} = -\frac{1}{2}\Boxt_g h_{\alpha\beta} + \frac{1}{2}F_{\alpha\beta},\quad \Boxt_gu := g^{\mu\nu}\del_{\mu}\del_{\nu}u,
$$
which leads to
$$
-\frac{1}{2}\Boxt_gh_{\alpha\beta} = -\frac{1}{2}F_{\alpha\beta} + 8\pi(\nabla_{\alpha}\phi\nabla_{\beta}\phi + V(\phi)g_{\alpha\beta})
$$
Thus, recalling $g^{\alpha\beta} = \eta^{\alpha\beta} + h^{\alpha\beta}$, $g_{\alpha\beta} = \eta_{\alpha\beta} + h_{\alpha\beta}$ with $\eta_{\alpha\beta} = \mathrm{diag}(-1,1,1,1)$,
\begin{equation}\label{eq1-25-nov-2025}
	-\frac{1}{2}\Boxt_g \big(\eta^{\alpha\beta}h_{\alpha\beta}\big) = -\frac{1}{2}\eta^{\alpha\beta}F_{\alpha\beta} 
	+ 8\pi\eta^{\alpha\beta}\del_{\alpha}\phi\del_{\beta}\phi + 32\pi V(\phi) 
	+ 8\pi V(\phi)\eta^{\alpha\beta}h_{\alpha\beta}.
\end{equation}
On the other hand, recalling
$$
\Boxt_g\phi - V'(\phi) = 0,
$$
one has
$$
\Boxt_g(\phi^2) = 2g^{\alpha\beta}\del_{\alpha}\phi\del_{\beta}\phi + 2\phi\Boxt_g\phi
=2g^{\alpha\beta}\del_{\alpha}\phi\del_{\beta}\phi + 2\phi V'(\phi)
$$
This (multiplied by $6\pi$) together with \eqref{eq1-25-nov-2025} leads to, recall that $V(\phi  )= \frac{c^2}{2}\phi^2$,
$$
\aligned
\Boxt_g\Big(-\frac{1}{2}\eta^{\alpha\beta}h_{\alpha\beta} - 6\pi\phi^2\Big) 
=& 32\pi\big(V(\phi) - \frac{3}{8}\phi V'(\phi)\big) 
- 4\pi \eta^{\alpha\beta}\del_{\alpha}\phi\del_{\beta}\phi
\\
&- 12\pi h^{\alpha\beta}\del_{\alpha}\phi\del_{\beta}\phi 
+  8\pi V(\phi)\eta^{\alpha\beta}h_{\alpha\beta}
-\frac{1}{2}\eta^{\alpha\beta}F_{\alpha\beta} 
\\
=&4\pi\Big((s/t)^2|\del_t\phi|^2 + c^2\phi^2\Big) 
+ 8\pi\Big((x^a/t)\del_t\phi\delu_a\phi - \frac{1}{2}\sum_a|\delu_a\phi|^2\Big)
\\
&+8\pi\big(V(\phi)\eta^{\alpha\beta}h_{\alpha\beta} - (3/2)h^{\alpha\beta}\del_{\alpha}\phi\del_{\beta}\phi\big)
- \frac{1}{2}\eta^{\alpha\beta}F_{\alpha\beta}.
\endaligned
$$
Applying the decomposition \eqref{eq6-20-june-2026} gives
\eqref{eq1-25-may-2026}. The favorable terms satisfy
\begin{equation}
|\mathbb F[h]|_{N-9}+|\mathbb T[h,\phi]|_{N-10}
\lesssim_{\delta,N}(C_1\eps)^2t^{-3+7\delta}(t-r)^{-1}.
\end{equation}
For $\mathbb F$ this is \eqref{eq10-20-june-2026}. To estimate the metric
terms in $\mathbb T$, the product rule and
$\delu_a=t^{-1}L_a$ give, for $p\leq N-10$,
\[
|\del\delu h|_p
\lesssim_N t^{-1}|\del h|_{p+1}+t^{-2}|h|_{p+1},
\qquad
|\delu\delu h|_p\lesssim_N t^{-2}|h|_{p+2}.
\]
The derivatives of the frame coefficients have size $O(t^{-1})$ and obey
the corresponding commuted bounds. Thus the metric contributions are
bounded by sums of terms of the form
\[
\sum_{p_1+p_2\leq p}|h|_{p_1}
\left(t^{-1}|\del h|_{p_2+1}+t^{-2}|h|_{p_2+2}\right).
\]
Here $p_2+1\leq N-9$ and $p_2+2\leq N-8$, so
\eqref{eq1-14-dec-2025} and \eqref{eq2-2026-7-10} apply.
Since $s\leq t$ and $t-r\leq t$, these terms have the stated bound.
The scalar terms, including the contribution of $12\pi\phi^2$ to $u$,
are controlled by the same product rule and
\eqref{eq5-20-june-2026}, \eqref{eq11-20-june-2026}, and
\eqref{eq1-14-dec-2025}. For example,
\[
|\del\phi|_p|\delu\phi|_p
\lesssim (C_1\eps)^2t^{-3}s^{-1+2\delta}
\lesssim (C_1\eps)^2t^{-3+2\delta}(t-r)^{-1},
\]
and
\[
|h|_p|\del\phi|_p^2
\lesssim (C_1\eps)^3t^{-2}s^{-2+3\delta}
\lesssim (C_1\eps)^2t^{-3+3\delta}(t-r)^{-1}.
\]
The remaining scalar products have at least the same decay. Only the
order-$N-10$ bound for $\mathbb T$ is used in the definition and estimates
of $\mathbb R$ below.

Finally, substituting \eqref{eq8-20-june-2026} into \eqref{eq1-25-may-2026} and applying \eqref{eq1-2026-7-10} together with \eqref{eq5-20-june-2026}, we obtain
\begin{equation}\label{eq1-5-27-2026}
	(1+4\mSch t^{-1})\del_t^2 u -\Delta_xu = \mathbb{S}[\phi]+ \mathbb{R}[h,\phi]
\end{equation}
with
\bea
\mathbb{R}[h,\phi]= \mathbb{F}[h] + \mathbb{T}[h,\phi]+\left(\hu^{00}+4\mSch t^{-1}\right)\del_t\del_tu,
\eea
Then recalling \eqref{eq4-02-aug-2026},
\begin{equation}\label{eq1-22-june-2026}
\big|\mathbb{R}[h,\phi]\big|_{N-10}\lesssim (C_1\eps)^2t^{-3+7\delta}(t-r)^{-1}.
\end{equation}

%
%
%
%
%

\subsection{The non-trivial support of the global solution}
In \cite[Proposition~2.3]{PLF-YM-book2} we have clarified that in the exterior region $\{r\geq t-1\}$ the global solution is trivial in the sense that the spacetime agrees with the Schwarzschild one while the scalar field vanishes. In the present article we need a more detailed analysis on the non-trivial region of the global solution. More precisely, we establish the following result.
\begin{lemma}\label{lem1-31-july-2026}
The global solution $(g_{\alpha\beta},\phi)$ of \eqref{eq3-29-nov-2025} with initial data satisfying \eqref{eq1-07-dec-2025} and \eqref{eq main smallness-condition} coincides with the Schwarzschild metric \eqref{eq Sch-wave} {\bf out of} the following region:
\begin{equation}
\Kcal^m:= \big\{(t,x)\big| |x|+2m_{\mathrm{S}}\ln|x| <t-1\big\}.
\end{equation}
Here and below, we adopt the convention
\[
r+2m_{\mathrm S}\ln r=-\infty\qquad\text{at }r=0.
\]
\end{lemma}
\begin{proof}
	We only need to remark that the outgoing characteristic surface of the ball $\big\{(2,x)\big||x|<1\big\}$ is determined by the following ODE:
	\begin{equation}
		R'(t) = \frac{R(t) - m_{\mathrm{S}}}{R(t) + m_{\mathrm{S}}},\quad R(2) = 1.
	\end{equation} 
	More precisely, this characteristic surface is $\big\{(t,x) \big||x| = R(t) \big\}$. This leads to the implicit solution
	\begin{equation}\label{eq2-03-june-2026}
		t - R(t) -1 = 2m_{\mathrm{S}}\ln \Big(\frac{R(t)-m_{\mathrm{S}}}{1-m_\mathrm{S}}\Big).
	\end{equation}
	It is clear that for $t\geq 2$, $R(t)\geq 1\gg m_{\mathrm{S}}$. Thus 
	\begin{equation}\label{eq1-03-june-2026}
		2m_{\mathrm{S}}\ln\Big(\frac{R(t)-m_{\mathrm{S}}}{1-m_{\mathrm{S}}}\Big)\geq 2m_{\mathrm{S}}\ln R(t).
	\end{equation}
	On the other hand, if $(t,x)$ belongs to the region of influence of $\big\{(2,x)\big||x|<1\big\}$, then $|x|\leq R(t)$. Thus the desired result follows.
\end{proof}

Then we have the following result:
\begin{corollary}\label{coro1-22-june-2026}
	One has
	\begin{equation}
		t-r-2m_{\mathrm{S}}\ln r \geq 1,\quad \forall(t,x)\in \mathcal{K}^m.
	\end{equation}
	Furthermore, the function
	\begin{equation}\label{eq1-15-aug-2026}
		\hat{q}(t,r) := t-r-2m_{\mathrm{S}}\ln t
	\end{equation}
	satisfies the following estimates:
	\begin{equation}
		1/2<\hat{q}(t,r)\leq t-r-2m_{\mathrm{S}}\ln r,
		\quad \forall(t,x)\in\mathcal{K}^m.
	\end{equation}
\end{corollary}
\begin{proof}
	For $r>0$, the first estimate follows directly from the definition of
	$\mathcal K^m$, while it is automatic at $r=0$ by the convention above.
	We now prove the lower bound for $\hat q$. We first remark that
	\begin{equation}
		\del_t\hat q(t,r)=1-2m_{\mathrm S}t^{-1}>0,
		\qquad t\geq2.
	\end{equation}
By decreasing the smallness threshold if necessary, we may assume that \(0\leq m_{\mathrm S}\leq\frac1{4\ln2}.\)
	When $0\leq r\leq1$ , we have
	\begin{equation}
		\hat q(t,r)\geq\hat q(2,r)
		=2-r-2m_{\mathrm S}\ln2
		\geq1-2m_{\mathrm S}\ln2\geq\frac12.
	\end{equation}
	The inequality is strict in $\mathcal K^m$, since $(t,r)=(2,1)$ does
	not belong to $\mathcal K^m$. When $r\geq1$, the definition of
	$\mathcal K^m$ gives
	\[
		t>r+1+2m_{\mathrm S}\ln r.
	\]
Remark that
	\[
		1+2m_{\mathrm S}\ln r\leq r.
	\]
	Consequently,
	\begin{equation}
		\aligned
		\hat q(t,r)
		&>1+2m_{\mathrm S}\ln r
		-2m_{\mathrm S}\ln\big(r+1+2m_{\mathrm S}\ln r\big)
		\\
		&=1-2m_{\mathrm S}\ln\left(
		\frac{r+1+2m_{\mathrm S}\ln r}{r}\right)
		\geq1-2m_{\mathrm S}\ln2
		\geq\frac12.
		\endaligned
	\end{equation}
	Finally, for $r>0$, one has $t\geq r$, and thus
	\[
		\hat q(t,r)=t-r-2m_{\mathrm S}\ln t
		\leq t-r-2m_{\mathrm S}\ln r.
	\]
	The upper bound is automatic at $r=0$ by the convention above.
\end{proof}

\subsection{The mass-adapted coordinate transform}
Inspired by the coordinates introduced in \cite{Lindblad-3,KauffmanLindblad}, we make the following construction.
\begin{equation}
\hat{t} = \hat{x}^0 := t-2m_{\mathrm{S}}\ln t,\quad \hat{x}^a = x^a.
\end{equation}
We denote by 
$$
\hat{\partial}_{\alpha} = \frac{\partial}{\partial \hat{x}^{\alpha}}.
$$
Then for any regular scalar field $\psi$,
\begin{equation}\label{eq3-04-june-2026}
	\aligned
	&\partial_t\psi = \partial_t\hat{t}\,\hat{\partial}_0\psi = (1-2m_{\mathrm{S}}t^{-1})\hat{\partial}_0\psi,
	\quad 
	&&\partial_a\psi = \hat{\partial}_a\psi,
	\\
	&\partial_t^2\psi= (1-2m_{\mathrm{S}}t^{-1})^2\hat{\partial}_0^2\psi
	+ \frac{2m_{\mathrm{S}}}{t^2-2m_{\mathrm{S}}t}\partial_t\psi,
	\quad
	&&\partial_a^2\psi = \hat{\partial}_a^2\psi,
	\endaligned
\end{equation}
also set
\bea
\hat{L}_a
:=\hat{x}^a\hat{\partial}_0+\hat{t}\hat{\partial}_a,
\qquad
\hat{\Omega}_{bc}
:=\hat{x}^b\hat{\partial}_c-\hat{x}^c\hat{\partial}_b.
\eea
Then
\begin{equation}
\aligned
(1+4m_{\mathrm{S}}t^{-1})\partial_t\partial_t \psi = \hat{\partial}_0^2\psi
&+ 4\mSch^2 t^{-2}\frac{4\mSch t^{-1}-3}{(1-2\mSch t^{-1})^2}\del_t^2\psi
+ 2\mSch t^{-2}\frac{1+4\mSch t^{-1}}{1-2\mSch t^{-1}}\del_t\psi
\\
&- 8\mSch^3t^{-4}\frac{4\mSch t^{-1}-3}{(1-2\mSch t^{-1})^3}\del_t\psi,
\endaligned
\end{equation}
and
\bea
\hat{L}_a\psi
=
L_a\psi
+\frac{2m_{\mathrm{S}}x^a}{t-2m_{\mathrm{S}}}\del_t\psi
-2m_{\mathrm{S}}\ln t\,\del_a\psi.
\qquad
\hat{\Omega}_{bc}\psi
=
\Omega_{bc}\psi.
\eea
By induction, we obtain, for $|I|+|J|\leq N,\quad |J|=k,$
\bea\label{eq1-2026-8-5}
|\hat{\del}^{I}\hat{L}^{J}\psi| \lesssim_{\delta,N} t^{\delta}|\psi|_{N,k}.
\eea
We denote by
\begin{equation}
\hat{\Box} = -\hat{\del}_0^2 + \sum_{a=1}^3\hat{\del}_a^2. 	
\end{equation}
Direct calculation shows that
\begin{equation}\label{eq1-2026-8-6}
\hat{\Box}\psi = -(1+4\mSch t^{-1})\del_t^2\psi + \Delta_x\psi + \mSch t^{-2} \mathbb{A}[\psi]
\end{equation}
with
\begin{equation}\label{eq2-2026-8-6}
\mathbb{A}[\psi] = 4\mSch\frac{4\mSch t^{-1}-3}{(1-2\mSch t^{-1})^2}\del_t^2\psi
+ \frac{2}{(1-2\mSch t^{-1})^3}\del_t\psi.
\end{equation}
Now we fix $\psi = u$. Then by estimates \eqref{eq2-2026-7-10}, \eqref{eq1-2026-7-10} and \eqref{eq5-20-june-2026},
\begin{equation}\label{eq2-22-june-2026}
\aligned
|\mathbb{A}[u]|_{N-10}\lesssim& |\del_t^2 h_{\alpha\beta}|_{N-10}+|\del_t h_{\alpha\beta}|_{N-10} 
+ \big|\del\del(\phi^2)\big|_{N-10} + \big|\del(\phi^2)\big|_{N-10}
\\
\lesssim& C_1\eps t^{-1+7\delta/2}(t-r)^{-1}.
\endaligned
\end{equation}
We therefore obtain the following proposition.
\begin{proposition}
Let $(h_{\alpha\beta},\phi)$ be the global solution to \eqref{eq3-29-nov-2025} together with initial data \eqref{eq main initial-data} satisfying \eqref{eq1-07-dec-2025} and \eqref{eq main smallness-condition}. Let $u$ be defined by \eqref{eq3'-02-dec-2025}. Then $u$ satisfies the following wave equation
\begin{equation}\label{eq3-22-june-2026}
-\hat{\Box} u = \mathbb{S}[\phi] + \mathbb{R}[h,\phi] - \mSch t^{-2}\mathbb{A}[u]
\end{equation}
with
\begin{equation}
\mathbb{S}[\phi] := 8\pi\Big((s/t)^2|\del_t\phi|^2+ c^2\phi^2\Big),
\end{equation}
\begin{equation}\label{eq1-01-august-2026}
\left|\mathbb{R}[h,\phi]\right|_{N-10} + \left|\mSch t^{-2}\mathbb{A}[u]\right|_{N-10}\lesssim 
\begin{cases}
(C_1\eps)^2 \hat{t}^{-3+7\delta}\hat{q}^{-1},\quad &\text{in}\quad \Kcal^m;
\\
(C_1\eps)^2 \hat{r}^{-4},\quad &\text{out of}\quad \Kcal^m.
\end{cases}
\end{equation}
Furthermore, by setting 
\begin{equation}
\hat{t}_{\mathrm{init}}=2-2\mSch\ln2,
\end{equation}
the initial data for $u$ under the coordinates $(\hat{t},\hat{x})$ are
\begin{equation}
\aligned
&u(\hat{t}_{\mathrm{init}},\hat{x}) = \eta^{\alpha\beta}h_{\alpha\beta}(\hat{t}_{\mathrm{init}},\hat{x})+12\pi\phi^2(\hat{t}_{\mathrm{init}},\hat{x}),
\\
&\hat{\del}_0u(\hat{t}_{\mathrm{init}},\hat{x}) = (1-\mSch)^{-1}\big(\eta^{\alpha\beta}\del_th_{\alpha\beta} + 24\pi\phi\del_t\phi\big)(\hat{t}_{\mathrm{init}},\hat{x}).
\endaligned
\end{equation}
\end{proposition}
\begin{proof}
Combining \eqref{eq1-5-27-2026} with \eqref{eq1-2026-8-6} gives \eqref{eq3-22-june-2026}. We then recall \eqref{eq1-22-june-2026} and \eqref{eq2-22-june-2026} and only need to translate the decay rate in $(t,r)$ into $(\hat{t},\hat{r})$. It is clear that $t \simeq \hat{t}$ for $t\geq 2$, while for $(t-r)$, we apply Corollary~\ref{coro1-22-june-2026}, which guarantees that in $\Kcal^m$,
\begin{equation}
1/2<\hat{t}-\hat{r}\leq t-r.
\end{equation}
Outside \(\mathcal K^m\), the metric is exactly Schwarzschild and \(\phi=0\); hence a direct computation from the explicit Schwarzschild expressions yields
\[
\left|\mathbb{R}[h,\phi]\right|_{N-10}
+
\left|m_{\mathrm S}t^{-2}\mathbb A[u]\right|_{N-10}
\lesssim
m_{\mathrm S}^2\hat r^{-4}
\lesssim
(C_1\varepsilon)^2\hat r^{-4}.
\]
\end{proof}
In the following discussion, we concentrate on \eqref{eq3-22-june-2026}, and the analysis will be carried out mainly in the coordinates $(\hat{t},\hat{x})$. For simplicity we often write $f(\hat{t},\hat{x}):=f(t,x)|_{t=t(\hat{t}),x=\hat{x}}.$
\subsection{Decomposition of the modified trace component}
We decompose
\begin{equation}\label{eq1-23-7-2026}
u=u_1+W,
\end{equation}
where
\begin{equation}\label{eq13-20-june-2026}
\begin{aligned}
&-\hat{\Box}u_1=\mathbb{G}[h,\phi],
\qquad
\mathbb{G}:=\mathbb{R}[h,\phi]-\mSch t^{-2}\mathbb{A}[u],
\\
&u_1(\hat{t}_{\mathrm{init}},\hat{x})
=\eta^{\alpha\beta}h_{\alpha\beta}(\hat{t}_{\mathrm{init}},\hat{x})
+12\pi\phi^2(\hat{t}_{\mathrm{init}},\hat{x}),
\\
&\hat{\del}_0u_1(\hat{t}_{\mathrm{init}},\hat{x})
=(1-\mSch)^{-1}
\bigl(
\eta^{\alpha\beta}\del_th_{\alpha\beta}
+24\pi\phi\del_t\phi
\bigr)(\hat{t}_{\mathrm{init}},\hat{x}),
\end{aligned}
\end{equation}
and
\begin{equation}\label{eq:def-W}
-\hat{\Box}W=\mathbb{S}[\phi]
=8\pi\Big((s/t)^2|\del_t\phi|^2+c^2\phi^2\Big),
\qquad
W(\hat{t}_{\mathrm{init}},\hat{x})=\hat{\del}_0W(\hat{t}_{\mathrm{init}},\hat{x})=0.
\end{equation}
We establish the estimates on the first quantity in the decomposition \eqref{eq1-23-7-2026}.
\begin{proposition}\label{prop1-19-07-2026}
	Let $(h_{\alpha\beta},\phi)$ be the global solution to
	\eqref{eq3-29-nov-2025} with initial data
	\eqref{eq main initial-data} satisfying
	\eqref{eq1-07-dec-2025} and
	\eqref{eq main smallness-condition}. Let $u_1$ be defined by
	\eqref{eq13-20-june-2026}. Then
	\begin{equation}\label{eq2-2026-8-3}
		|u_1|
		\lesssim_{\delta}
		C_1\eps\hat{t}^{-1}\langle\hat{t}-\hat{r}\rangle^{-1+8\delta},
	\qquad \hat t-\hat r\geq1/2.
	\end{equation}
	Furthermore, the radiation field\footnote{The definition of this radiation field depends on the choice of coordinates. Thus we use the notation $\hat{\Rcal}$.}
	\begin{equation}\label{eq5-03-aug-2026}
	\hat{\Rcal}[u_1](\omega,\hat{q}) := \lim_{\hat r\rightarrow+\infty}\hat{r}u_1(\hat{q}+\hat{r},\hat{r}\omega)
	\end{equation}
	exists for $\hat q\geq 1/2$ and $\omega\in\mathbb S^2$, and
	\begin{equation}\label{eq6-03-aug-2026}
	\big|\hat{\Rcal}[u_1](\omega,\hat{q})\big|\lesssim_{\delta} C_1\eps \hat{q}^{-1+9\delta}.
	\end{equation}
\end{proposition}
The proof is postponed to Appendix~\ref{sec1-31-july-2026}.


\section{The scalar contribution at timelike infinity}
\label{sec:kg-induced-wave-profile}
We now consider the scalar contribution $W$ in the decomposition
\eqref{eq1-23-7-2026}. Recall that
\[
u=u_1+W,
\]
where $u_1$ satisfies Proposition~\ref{prop1-19-07-2026} and
\[
-\hat\Box W=\mathbb S[\phi]
=8\pi\Big((s/t)^2|\del_t\phi|^2+c^2\phi^2\Big),
\qquad
W|_{\hat t=\hat t_{\mathrm{init}}}
=\hat\del_0W|_{\hat t=\hat t_{\mathrm{init}}}=0.
\]
We shall determine the limit of $\hat t W(\hat t,\hat t z)$ for
$|z|<1$ and estimate the remainder.

\subsection{Existence of the radiation field}
We first apply the estimates of Appendix~\ref{sec1-31-july-2026} to the
scalar source. The resulting radiation field will be identified in
Section~\ref{sec1-04-aug-2026}.
By the support property established above and finite propagation speed,
$\phi$ and $W$ vanish in $\{\hat r\geq\hat t-1/2\}$.

\begin{proposition}[Existence of the radiation fields]\label{prop:radiation-field-existence}
For $\hat q\in[1/2,+\infty)$ and $\omega\in\mathbb S^2$, the radiation
fields $\hat{\Rcal}[u](\omega,\hat q)$ and
$\hat{\Rcal}[W](\omega,\hat q)$ both exist. Furthermore,
\begin{equation}\label{eq1-26-8-17}
\big|\hat{\Rcal}[u] - \hat\Rcal[W]\big|\lesssim_{\delta} C_1\eps\hat{q}^{-1+9\delta}.
\end{equation}
\begin{equation}\label{eq10-03-aug-2026}
\big|\hat\Rcal[W](\omega,\hat{q})\big|\lesssim_{\delta}(C_1\eps)^2.
\end{equation}
\end{proposition}
\begin{proof}
By \eqref{eq5-03-aug-2026}, \eqref{eq6-03-aug-2026}, and
\eqref{eq1-23-7-2026}, it suffices to prove the existence of
$\hat\Rcal[W]$. Indeed,
\[
\hat\Rcal[u]-\hat\Rcal[W]=\hat\Rcal[u_1],
\]
and \eqref{eq1-26-8-17} then follows from \eqref{eq6-03-aug-2026}.

Write $f:=\mathbb S[\phi]$. In $\{\hat r<\hat t-1/2\}$,
\eqref{eq7-03-aug-2026} gives
\begin{equation}\label{eq8-03-aug-2026}
|f|\lesssim(C_1\eps)^2 \hat{t}^{-3}
\lesssim (C_1\eps)^2\hat{t}^{-5/2}(\hat{t}-\hat{r})^{-1/2}.
\end{equation}
Since both Cauchy data of $W$ vanish, Corollary~\ref{cor1-02-aug-2026}
with $a=b=1/2$ gives
\begin{equation}\label{eq11-03-aug-2026}
|W|\lesssim (C_1\eps)^2\hat{t}^{-1}.
\end{equation}

Furthermore, the product estimates
\eqref{eq5-20-june-2026} and \eqref{eq7-03-aug-2026} give
\begin{equation}\label{eq9-03-aug-2026}
\begin{aligned}
|f|_{N-6}
&\lesssim (C_1\eps)^2t^{-3}s^{2\delta}
\lesssim (C_1\eps)^2\hat t^{-3+2\delta}
\\
&\lesssim (C_1\eps)^2
\hat t^{-5/2+\delta}
(\hat t-\hat r)^{-1/2+\delta},
\end{aligned}
\end{equation}
where in the last step we used $\hat t-\hat r\leq\hat t$ and $\delta<1/2$.
Since $\hat\Omega_{bc}=\Omega_{bc}$, \eqref{eq9-03-aug-2026} also gives
\begin{equation}
|\Delta_\omega f|
\lesssim |f|_{N-6}
\lesssim (C_1\eps)^2
\hat t^{-5/2+\delta}
(\hat t-\hat r)^{-1/2+\delta}.
\end{equation}
Both Cauchy data of $\Delta_\omega W$ also vanish. Thus
Corollary~\ref{cor1-03-aug-2026}, with $a=b=1/2-\delta$, gives the
existence of $\hat\Rcal[W]$. The bound \eqref{eq10-03-aug-2026} follows
from \eqref{eq11-03-aug-2026}.
\end{proof}

\subsection{Asymptotic profiles of the Klein--Gordon field}
\label{subsec:refined-kg-asymptotics-main}
We use asymptotic profiles for both $\phi$ and its time derivative.
The construction follows the integration along hyperboloidal rays in
Chen and Lindblad \cite[Section~4.1]{Chen-Lindblad-2023}, with a modified
phase as in their quasilinear analysis
\cite[Section~5]{Chen-Lindblad-2023}. The coefficient bounds and the
remainder estimates for the present system are proved in
Appendix~\ref{sec:appendix-kg-profile-technical}.

\begin{theorem}[Klein--Gordon asymptotic profiles]
\label{thm:refined-kg-asymptotics}
Let $(h_{\alpha\beta},\phi)$ be the global solution to
\eqref{eq3-29-nov-2025} with initial data \eqref{eq main initial-data}
satisfying \eqref{eq1-07-dec-2025} and \eqref{eq main smallness-condition},
and let $0<\delta<1/9$. Write
\[
\hat r=|\hat x|,\qquad
\hat s=\sqrt{\hat t^2-|\hat x|^2},\qquad
\hat y=\frac{\hat x}{\hat t},
\]
and set
\[
B_1(0):=\{\hat y\in\mathbb R^3:\ |\hat y|<1\}.
\]
For each fixed $\hat y\in B_1(0)$, define the hyperboloidal ray in the
hatted coordinates by
\[
\gamma_{\hat y}(\hat s)
:=\left(
\frac{\hat s}{\sqrt{1-|\hat y|^2}},
\frac{\hat s\hat y}{\sqrt{1-|\hat y|^2}}
\right).
\]
Set
\[
\mathfrak q(\hat s,\hat y)
:=\frac1{1-|\hat y|^2}
\left.
\left[
\bigl(\hu^{00}+4\mSch t^{-1}\bigr)
\bigl(1-2\mSch t^{-1}\bigr)^2
\right]
\right|_{\gamma_{\hat y}(\hat s)},
\qquad
a(\hat s,\hat y):=\sqrt{1+\mathfrak q(\hat s,\hat y)},
\]
where $t$ denotes the original time coordinate, evaluated along the ray
through the inverse coordinate transform. Set $s_0:=\hat t_{\mathrm{init}}$.
For $C_1\eps$ sufficiently small, one has $|\mathfrak q|\leq1/3$
for $\hat s\geq s_0$ along these rays, and hence $a$ is real-valued
and continuous on $[s_0,\infty)\times B_1(0)$.
For phase integrals only, we extend $a$ to $0\leq\lambda<s_0$ by
\[
a(\lambda,\hat y):=a(s_0,\hat y),\qquad 0\leq\lambda<s_0.
\]
This real continuous extension is a convention; the metric square-root
formula is used only for $\lambda\geq s_0$.
Then there exist continuous profiles
$b^\pm=b^\pm(\hat y)$ on $B_1(0)$ such that
\begin{equation}\label{eq:kg-profile-bound-main}
|b^\pm(\hat y)|
\lesssim C_1\eps(1-|\hat y|^2)^{5/4-11\delta/4},
\qquad |\hat y|<1,
\end{equation}
and, since $\phi$ is real-valued,
\[
b^-(\hat y)=\overline{b^+(\hat y)},\qquad |\hat y|<1.
\]
Define the associated phase by
\[
\Phi(\hat s,\hat y)
:=c\int_{\hat t_{\mathrm{init}}}^{\hat s}a(\sigma,\hat y)\diff\sigma,
\]
and set
\[
\mathcal A(\hat s,\hat y)
:=e^{-i\Phi(\hat s,\hat y)}b^+(\hat y)
+e^{i\Phi(\hat s,\hat y)}b^-(\hat y),
\]
\[
\mathcal D(\hat s,\hat y)
:=-ci e^{-i\Phi(\hat s,\hat y)}b^+(\hat y)
+ci e^{i\Phi(\hat s,\hat y)}b^-(\hat y).
\]
In the region $\hat t-\hat r\geq1/2$ one has
\begin{equation}\label{eq:simultaneous-phi-error}
\left|\hat s^{3/2}\phi-\mathcal A(\hat s,\hat y)\right|
\lesssim
C_1\eps
\left(\frac{\hat s}{\hat t}\right)^{3/2+7\delta/2}
\hat s^{-1+9\delta},
\end{equation}
and
\begin{equation}\label{eq:simultaneous-dt-error}
\begin{aligned}
&\left|
\hat s^{3/2}\frac{\hat s}{\hat t}\hat\del_0\phi
-\mathcal D(\hat s,\hat y)
\right|
\\
&\qquad\lesssim C_1\eps\left[
\left(\frac{\hat s}{\hat t}\right)^{3/2+7\delta/2}
\hat s^{-1+9\delta}
+
\left(\frac{\hat s}{\hat t}\right)^{3/2-\delta}
\hat s^{-1+3\delta}
\right].
\end{aligned}
\end{equation}
In particular,
\[
\phi(\hat t,\hat x)
=\hat s^{-3/2}\left(
 e^{-i\Phi(\hat s,\hat y)}b^+(\hat y)
 +e^{i\Phi(\hat s,\hat y)}b^-(\hat y)
\right)
+O\left(C_1\eps
\left(\frac{\hat s}{\hat t}\right)^{3/2+7\delta/2}
\hat s^{-5/2+9\delta}\right).
\]
\end{theorem}

\subsection{The leading scalar source}
We substitute the Klein--Gordon profiles into $\mathbb S[\phi]$ and
separate the leading source from the remainder, following the source
decomposition in \cite[Section~4.2]{Chen-Lindblad-2023}. In the weighted
combination defining $\mathbb S[\phi]$, the leading oscillatory terms
cancel. This gives the following non-negative source profile.

\begin{proposition}[The scalar-induced source profile]
\label{prop:scalar-induced-source-profile}
Let $b^\pm$ be the Klein--Gordon profiles given by
Theorem~\ref{thm:refined-kg-asymptotics}, and let $\hat y=\hat x/\hat t$.
In the region $\hat t-\hat r\geq1/2$,
\begin{equation}\label{eq:modified-source-decomposition}
\mathbb S[\phi](\hat t,\hat x)
=\hat t^{-3}P_0(\hat y)+R_0(\hat t,\hat x),
\end{equation}
where
\begin{equation}\label{eq:def-new-P0}
P_0(\hat y)
:=32\pi c^2(1-|\hat y|^2)^{-3/2}|b^+(\hat y)|^2\geq0,
\qquad |\hat y|<1.
\end{equation}
Moreover,
\begin{equation}\label{eq:modified-source-profile-bounds}
0\leq P_0(\hat y)
\lesssim (C_1\eps)^2(1-|\hat y|^2)^{1-11\delta/2},
\qquad |\hat y|<1,
\end{equation}
and
\begin{equation}\label{eq:modified-source-remainder-bound}
|R_0(\hat t,\hat x)|
\lesssim_\delta (C_1\eps)^2\hat t^{-4+9\delta},
\qquad \hat t-\hat r\geq\frac12.
\end{equation}
\end{proposition}

\begin{proof}
Set
\[
\vartheta:=\frac{\hat s}{\hat t}=\sqrt{1-|\hat y|^2},
\qquad
\lambda:=\hat s,
\qquad
B:=1-2\mSch t^{-1}.
\]
By Theorem~\ref{thm:refined-kg-asymptotics}, the leading Klein--Gordon profiles are
\[
\mathcal A=e^{-i\Phi}b^++e^{i\Phi}b^- ,
\qquad
\mathcal D=-ci e^{-i\Phi}b^++ci e^{i\Phi}b^- .
\]
Writing
\[
\lambda^{3/2}\phi=\mathcal A+\rho_\phi,
\qquad
\lambda^{3/2}\vartheta\hat\del_0\phi
=\mathcal D+\rho_t,
\]
the remainders $\rho_\phi$ and $\rho_t$ satisfy the estimates of
Theorem~\ref{thm:refined-kg-asymptotics}. Since $\phi$ is real-valued,
$b^-=\overline{b^+}$. Recall that
\begin{equation}\label{eq:s-hat-s-comparison}
\left(\frac{s}{t}\right)^2
=
\left(\frac{\hat s}{\hat t}\right)^2
+\frac{4\mSch r^2\ln t\,(t-\mSch\ln t)}
{t^2(t-2\mSch\ln t)^2}.
\end{equation}
Writing $(s/t)^2=\vartheta^2+\mathcal E_m$, we have
$|\mathcal E_m|\lesssim\mSch t^{-1}\ln t$. Since
$\del_t=B\hat\del_0$ by \eqref{eq3-04-june-2026},
\begin{equation}\label{eq:source-before-cancellation}
\mathbb S[\phi]=8\pi\lambda^{-3}\left[
B^2(\mathcal D+\rho_t)^2+c^2(\mathcal A+\rho_\phi)^2\right]
+8\pi\mathcal E_m(\del_t\phi)^2.
\end{equation}
For the leading terms,
\[
\begin{aligned}
\mathcal D^2
&=-c^2e^{-2i\Phi}(b^+)^2-c^2e^{2i\Phi}(b^-)^2+2c^2b^+b^-,\\
c^2\mathcal A^2
&=c^2e^{-2i\Phi}(b^+)^2+c^2e^{2i\Phi}(b^-)^2+2c^2b^+b^-.
\end{aligned}
\]
Thus the oscillatory leading terms cancel exactly:
\begin{equation}\label{eq:leading-source-cancellation}
\mathcal D^2+c^2\mathcal A^2
=4c^2b^+b^-=4c^2|b^+|^2.
\end{equation}
Since $\lambda=\hat t\vartheta$, the leading term is
\[
\hat t^{-3}32\pi c^2\vartheta^{-3}|b^+(\hat y)|^2
=\hat t^{-3}P_0(\hat y).
\]
By \eqref{eq:kg-profile-bound-main},
\eqref{eq:modified-source-profile-bounds} follows.
Lemma~\ref{lem:scalar-source-remainder} in
Appendix~\ref{sec:appendix-kg-profile-technical} then gives
\eqref{eq:modified-source-remainder-bound}.
\end{proof}

Since $1-11\delta/2>0$, we extend $P_0$ continuously by zero to
$\overline{B_1(0)}$ and set it equal to zero on $\RR^3\setminus B_1(0)$.

\begin{remark}
\label{rem:direct-limit-cancellation-role}
The identity \eqref{eq:leading-source-cancellation} concerns the leading
terms of $\mathbb S[\phi]$. The factors $B^2-1$ and $\mathcal E_m$ in
\eqref{eq:source-before-cancellation} contribute to $R_0$ and are estimated
in Lemma~\ref{lem:scalar-source-remainder}.
\end{remark}

\subsection{The induced wave profile}
We derive an integral representation of the wave profile over $B_1(0)$
with an explicit positive kernel. The representation follows from the
change of variables and Jacobian calculation in the proof below.
For $z\in \overline{B_1(0)}$ and $\hat y\in B_1(0)$ set
\begin{equation}\label{eq:def-positive-kernel-discriminant}
D(z,\hat y)
:=(1-z\cdot\hat y)^2-(1-|z|^2)(1-|\hat y|^2)
\end{equation}
and
\begin{equation}\label{eq:def-positive-kernel}
K(z,\hat y):=D(z,\hat y)^{-1/2}.
\end{equation}
For $z=\omega\in\mathbb S^2$, this reduces to
\begin{equation}\label{eq:positive-kernel-boundary-value}
K(\omega,\hat y)=\frac1{1-\omega\cdot\hat y},
\qquad |\hat y|<1.
\end{equation}
We define the interior profile by
\begin{equation}\label{eq:positive-kernel-interior}
U(z):=\frac1{4\pi}\int_{|\hat y|<1}
K(z,\hat y)P_0(\hat y)\diff\hat y,
\qquad |z|<1.
\end{equation}
The local integrability of the kernel is shown in the proof below.

\begin{proposition}
\label{prop:direct-full-source-limit} Let $W$ be defined by \eqref{eq:def-W}. Then
for every fixed $z\in B_1(0)$,
\begin{equation}\label{eq:direct-full-source-limit}
\lim_{\hat t\rightarrow+\infty}4\pi \hat t\, W(\hat t,\hat t z)
=
\int_{|\hat y|<1}K(z,\hat y)P_0(\hat y)\diff\hat y
=4\pi U(z).
\end{equation}
\end{proposition}

\begin{proof}
Kirchhoff's formula for \eqref{eq:def-W} gives
\begin{equation}\label{eq:direct-full-kirchhoff-start}
W(\hat t,\hat t z)
=
\frac1{4\pi}
\int_{\hat t_{\mathrm{init}}}^{\hat t}
\frac1{\hat t-\hat\tau}
\int_{|\hat\xi-\hat t z|=\hat t-\hat\tau}
\mathbb S[\phi](\hat\tau,\hat\xi)
\diff S_{\hat\xi}\diff\hat\tau.
\end{equation}
Write
\[
\lambda:=\frac{\hat\tau}{\hat t},
\qquad
\hat\xi=\hat t z+\hat t(1-\lambda)\eta,
\qquad
\eta\in\mathbb S^2.
\]
Since
\[
\diff S_{\hat\xi}=\hat t^2(1-\lambda)^2\diff\eta,
\qquad
\diff\hat\tau=\hat t\diff\lambda,
\]
we obtain
\begin{equation}\label{eq:direct-full-rescaled-kirchhoff}
4\pi \hat t\, W(\hat t,\hat t z)
=
\int_{\hat t_{\mathrm{init}}/\hat t}^{1}
\int_{\mathbb S^2}
(1-\lambda)\lambda^{-3}
F_{\hat t}(\lambda,\eta)
\diff\eta\diff\lambda,
\end{equation}
where
\begin{equation}\label{eq:def-direct-rescaled-source}
F_{\hat t}(\lambda,\eta)
:=(\lambda\hat t)^3
\mathbb S[\phi]\left(
\lambda\hat t,
\hat t\bigl(z+(1-\lambda)\eta\bigr)
\right).
\end{equation}
Chen and Lindblad \cite[Section~4.2.1]{Chen-Lindblad-2023} use this
rescaling of Kirchhoff's formula to construct the wave profile generated
by a non-oscillatory source.
For $\lambda>0$, set
\begin{equation}\label{eq:def-direct-y-lambda-eta}
\hat y=\hat y(\lambda,\eta)
:=\frac{z+(1-\lambda)\eta}{\lambda}.
\end{equation}
Then
\[
\left(
\lambda\hat t,\,
\hat t\bigl(z+(1-\lambda)\eta\bigr)
\right)
=
(\lambda\hat t,\lambda\hat t\hat y).
\]

If $|\hat y|<1$, then
\[
\lambda\hat t-|\lambda\hat t\hat y|
=\lambda\hat t(1-|\hat y|)\longrightarrow+\infty,\quad \hat t\to \infty.
\]
Hence, for all sufficiently large $\hat t$,
\eqref{eq:modified-source-decomposition} and
\eqref{eq:modified-source-remainder-bound} apply at this point, and
\begin{equation}\label{eq:direct-rescaled-source-convergence-interior}
F_{\hat t}(\lambda,\eta)
=P_0(\hat y)
+O_\delta\bigl((C_1\eps)^2(\lambda\hat t)^{-1+9\delta}\bigr)
\longrightarrow P_0(\hat y).
\end{equation}
If $|\hat y|\geq1$, the scalar field vanishes at the corresponding point,
so $F_{\hat t}(\lambda,\eta)=0=P_0(\hat y)$. Thus, for every fixed $\lambda>0$
and $\eta\in\mathbb S^2$,
\begin{equation}\label{eq:direct-rescaled-source-convergence}
F_{\hat t}(\lambda,\eta)
\longrightarrow
P_0\left(\frac{z+(1-\lambda)\eta}{\lambda}\right).
\end{equation}

If $F_{\hat t}(\lambda,\eta)\neq0$, the support property gives $|\hat y|<1$.
From \eqref{eq:def-direct-y-lambda-eta},
\[
1-\lambda=|\lambda\hat y-z|
\leq\lambda|\hat y|+|z|,
\]
and consequently
\begin{equation}\label{eq:direct-lambda-lower-bound}
\lambda
\geq\frac{1-|z|}{1+|\hat y|}
>\frac{1-|z|}{2}
=: \lambda_*(z)>0.
\end{equation}
For fixed $z\in B_1(0)$, the support is therefore bounded away from
$\lambda=0$. On this support,
\eqref{eq:modified-source-decomposition}--\eqref{eq:modified-source-remainder-bound}
and \eqref{eq:modified-source-profile-bounds} imply
\[
|F_{\hat t}(\lambda,\eta)|
\lesssim_{\delta,z}(C_1\eps)^2,
\]
because $-1+9\delta<0$. Hence
\begin{equation}\label{eq:direct-dominated-majorant}
(1-\lambda)\lambda^{-3}|F_{\hat t}(\lambda,\eta)|
\lesssim_{\delta,z}
(C_1\eps)^2\mathds 1_{[\lambda_*(z),1]}(\lambda),
\end{equation}
which is integrable independently of $\hat t$. Dominated convergence gives
\begin{equation}\label{eq:direct-limit-lambda-eta}
\lim_{\hat t\rightarrow+\infty}4\pi \hat t\, W(\hat t,\hat t z)
=
\int_0^1\int_{\mathbb S^2}
(1-\lambda)\lambda^{-3}
P_0\left(\frac{z+(1-\lambda)\eta}{\lambda}\right)
\diff\eta\diff\lambda.
\end{equation}

We now use \eqref{eq:def-direct-y-lambda-eta} to change variables from
$(\lambda,\eta)$ to $\hat y$. The relation
$|\lambda\hat y-z|=1-\lambda$ is equivalent to
\begin{equation}\label{eq:direct-quadratic-lambda}
(1-|\hat y|^2)\lambda^2
-2(1-z\cdot\hat y)\lambda
+(1-|z|^2)=0.
\end{equation}
For $\hat y\neq z$, the left-hand side is positive at $\lambda=0$ and
negative at $\lambda=1$, and hence exactly one root lies in $(0,1)$; the
other root is larger than one. Thus
$(\lambda,\eta)\mapsto\hat y$ is one-to-one onto
$B_1(0)\setminus\{z\}$.  The omitted point corresponds to the degenerate
endpoint $\lambda=1$ and is negligible for integration.

For $z\neq0$, put $e=z/|z|$ and write
$\hat y=\alpha e+v$ with $v\perp e$. Then
\begin{equation}\label{eq:positive-kernel-discriminant}
D(z,\hat y)
=(\alpha-|z|)^2+(1-|z|^2)|v|^2,
\end{equation}
while $D(0,\hat y)=|\hat y|^2$. In particular,
\begin{equation}\label{eq:D-distance-comparison}
(1-|z|^2)|\hat y-z|^2
\leq D(z,\hat y)
\leq |\hat y-z|^2.
\end{equation}
Hence $D(z,\hat y)>0$ for $\hat y\neq z$, while near the diagonal
\[
D(z,\hat y)^{-1/2}
\simeq_{z}
|\hat y-z|^{-1}.
\]
The resulting diagonal singularity is locally integrable in three dimensions.

At fixed $\lambda$, the image of $\eta\mapsto\hat y$ is a sphere of radius
$(1-\lambda)/\lambda$, so the two tangential directions contribute the area
factor $((1-\lambda)/\lambda)^2$. Moreover,
\[
\partial_\lambda\hat y=-\lambda^{-2}(z+\eta),
\]
and its component in the spherical normal direction $\eta$ has absolute
value $\lambda^{-2}(1+z\cdot\eta)$. If
$A=1-|\hat y|^2$, $B=1-z\cdot\hat y$, and $C=1-|z|^2$, then the root
in $(0,1)$ satisfies $B-A\lambda=\sqrt{D(z,\hat y)}$, and
\[
1+z\cdot\eta
=\frac{C-\lambda B}{1-\lambda}
=\frac{\lambda\sqrt{D(z,\hat y)}}{1-\lambda}.
\]
Therefore
\begin{equation}\label{eq:direct-jacobian}
\diff\hat y
=\sqrt{D(z,\hat y)}
(1-\lambda)\lambda^{-3}\diff\eta\diff\lambda.
\end{equation}
Using this change of variables in \eqref{eq:direct-limit-lambda-eta} gives
\begin{equation}\label{eq:direct-limit-positive-kernel}
\lim_{\hat t\rightarrow+\infty}4\pi \hat t\, W(\hat t,\hat t z)
=
\int_{|\hat y|<1}
\frac{P_0(\hat y)}{\sqrt{D(z,\hat y)}}\diff\hat y
=4\pi U(z),
\end{equation}
which proves the proposition.
\end{proof}

The source remainder gives the following estimate away from the light cone.

\begin{theorem}[Timelike asymptotics]
\label{thm:scalar-induced-wave-profile}
Let $W$ be defined by \eqref{eq:def-W} and let $U$ be given by
\eqref{eq:positive-kernel-interior}. Then, for every fixed $z\in B_1(0)$,
\begin{equation}\label{eq:theorem-direct-timelike-limit}
\lim_{\hat t\to+\infty}\hat t\,W(\hat t,\hat t z)=U(z).
\end{equation}
Moreover,
\begin{equation}\label{eq3-2026-8-18}
W(\hat t,\hat x)
=
\frac1{\hat t}U\left(\frac{\hat x}{\hat t}\right)
+O_\delta\left((C_1\eps)^2\hat t^{-1}
\langle\hat t-\hat r\rangle^{-1+9\delta}\right),
\qquad
\hat t-\hat r\geq2\hat t_{\mathrm{init}}.
\end{equation}
\end{theorem}

\begin{proof}[Proof of Theorem~\ref{thm:scalar-induced-wave-profile}]
The fixed-ray limit \eqref{eq:theorem-direct-timelike-limit} follows from
Proposition~\ref{prop:direct-full-source-limit}. To prove
\eqref{eq3-2026-8-18}, denote by $\mathcal G[F]$ the solution of
$-\hat\Box\mathcal G[F]=F$ with both Cauchy data zero at
$\hat t=\hat t_{\mathrm{init}}$. Kirchhoff's formula gives
\begin{equation}\label{eq:def-retarded-operator}
\mathcal G[F](\hat t,\hat x)
:=\frac1{4\pi}
\int_{\hat t_{\mathrm{init}}}^{\hat t}
\frac1{\hat t-\hat\tau}
\int_{|\hat\xi-\hat x|=\hat t-\hat\tau}
F(\hat\tau,\hat\xi)\diff S_{\hat\xi}\diff\hat\tau.
\end{equation}
Thus $W=\mathcal G[\mathbb S[\phi]]$. Extend $R_0$ by
$R_0=\mathbb S[\phi]-\hat t^{-3}P_0(\hat x/\hat t)$ outside
$\{\hat t-\hat r\geq1/2\}$. Then \eqref{eq:modified-source-decomposition}
and linearity give
\begin{equation}\label{eq:full-source-kirchhoff-split}
W
=\mathcal G\!\left[\hat t^{-3}P_0(\hat x/\hat t)\right]
+\mathcal G[R_0].
\end{equation}

Let $z=\hat x/\hat t$ and assume
$\hat t-\hat r=\hat t(1-|z|)\geq2\hat t_{\mathrm{init}}$. The same
rescaling and change of variables used in the proof of
Proposition~\ref{prop:direct-full-source-limit} give the exact identity
\begin{equation}\label{eq:leading-source-exact-profile}
\mathcal G\!\left[\hat t^{-3}P_0(\hat x/\hat t)\right](\hat t,\hat x)
=\hat t^{-1}U(\hat x/\hat t).
\end{equation}
For the remainder, in the region
$\hat q:=\hat t-\hat r\geq1/2$,
\eqref{eq:modified-source-remainder-bound} gives
\[
|R_0|
\lesssim_\delta(C_1\eps)^2\hat t^{-4+9\delta}
\lesssim_\delta(C_1\eps)^2
\hat t^{-3+8\delta}\hat q^{-1+\delta}.
\]
In the exterior region $\hat r\geq\hat t-1/2$, the scalar source vanishes.
If $\hat r\geq\hat t$, the extended $P_0$ also vanishes. If
$\hat t-1/2\leq\hat r<\hat t$, then
\[
1-|\hat y|^2
=\frac{(\hat t-\hat r)(\hat t+\hat r)}{\hat t^2}
\lesssim\hat r^{-1},
\]
and hence, by \eqref{eq:modified-source-profile-bounds},
\[
|R_0|
=\hat t^{-3}|P_0(\hat y)|
\lesssim_\delta(C_1\eps)^2\hat r^{-4+11\delta/2}.
\]
Since $\mathcal G[R_0]$ has vanishing initial data,
Corollary~\ref{cor1-02-aug-2026} applies with
\[
a=1-8\delta,\qquad
b=1-\delta,\qquad
d=3-\frac{11}{2}\delta,
\]
for which $d\geq2+a$. Thus
\begin{equation}\label{eq:retarded-R0-bound}
|\mathcal G[R_0](\hat t,\hat x)|
\lesssim_\delta(C_1\eps)^2\hat t^{-1}
\langle\hat t-\hat r\rangle^{-1+9\delta}.
\end{equation}
Combining \eqref{eq:full-source-kirchhoff-split},
\eqref{eq:leading-source-exact-profile}, and
\eqref{eq:retarded-R0-bound} proves \eqref{eq3-2026-8-18}.
\end{proof}

\section{Asymptotics at null infinity and single-direction observability}
\label{sec1-04-aug-2026}

We now extend the timelike profile $U$ to $\partial B_1(0)$ and identify
its boundary value with the final radiation profile. We then express the
limiting Klein--Gordon energy in terms of $P_0$. Comparing this expression
with the integral representation of $U$, we prove quantitative single-direction
observability.

\subsection{Continuous extension and null asymptotics}
\begin{proposition}[Continuous extension to the closed ball]
\label{prop:closed-ball-extension}
Let $P_0$ and $U$ be defined by
\eqref{eq:def-new-P0} and \eqref{eq:positive-kernel-interior}. For every
$\omega\in\mathbb S^2$, set
\begin{equation}\label{eq:positive-kernel-null}
U_\infty(\omega)
:=\frac1{4\pi}\int_{|\hat y|<1}
\frac{P_0(\hat y)}{1-\omega\cdot\hat y}\diff\hat y.
\end{equation}
Then $U\in C(B_1(0))$, $U_\infty\in C(\mathbb S^2)$, and
\begin{equation}\label{eq:uniform-radial-profile-limit}
\sup_{\omega\in\mathbb S^2}
|U(\rho\omega)-U_\infty(\omega)|\longrightarrow0,
\qquad \rho\uparrow1.
\end{equation}
Consequently $U$ extends uniquely to a function in
$C(\overline{B_1(0)})$, still denoted by $U$, with
$U|_{\mathbb S^2}=U_\infty$. Equivalently,
\begin{equation}\label{eq:positive-kernel-closed-ball}
U(z)=\frac1{4\pi}\int_{|\hat y|<1}
K(z,\hat y)P_0(\hat y)\diff\hat y,
\qquad z\in\overline{B_1(0)},
\end{equation}
where on the boundary we understand
\[
K(\omega,\hat y)=\frac1{1-\omega\cdot\hat y}.
\]
\end{proposition}

\begin{proof}
We first consider the interior. Fix $z_0\in B_1(0)$ and choose
$\Omega\Subset B_1(0)$ containing a neighborhood of $z_0$. Then
\eqref{eq:D-distance-comparison} gives, uniformly for $z\in\Omega$,
\[
K(z,\hat y)\lesssim_\Omega |\hat y-z|^{-1}.
\]
For $0<\eta<\frac13\operatorname{dist}(z_0,\partial\Omega)$ and
$|z-z_0|\leq\eta$, the near contributions over $|\hat y-z_0|<2\eta$
are bounded uniformly by
\[
C_\Omega\|P_0\|_{L^\infty}
\int_{|\hat y-z|<3\eta}|\hat y-z|^{-1}\diff\hat y
\lesssim_\Omega\|P_0\|_{L^\infty}\eta^2.
\]
On the complement $|\hat y-z_0|\geq2\eta$, the same condition
$|z-z_0|\leq\eta$ gives
\[
|\hat y-z|\geq|\hat y-z_0|-|z-z_0|\geq\eta.
\]
Thus the kernel is uniformly nonsingular there and dominated convergence applies. Letting first
$z\to z_0$ and then $\eta\downarrow0$ gives
\[
U\in C(B_1(0)).
\]

We now study the boundary. By
\eqref{eq:positive-kernel-boundary-value}, for fixed
$\omega\in\mathbb S^2$ and $|\hat y|<1$,
\begin{equation}\label{eq:positive-kernel-boundary-limit}
K(\rho\omega,\hat y)
\longrightarrow K(\omega,\hat y)
=\frac1{1-\omega\cdot\hat y},
\qquad \rho\uparrow1.
\end{equation}
To pass to the integral uniformly in $\omega$, we estimate the contribution
near the singularity of the boundary kernel and the kernel difference on
its complement separately.
Set
\[
\kappa:=1-\frac{11}{2}\delta>0.
\]
By \eqref{eq:modified-source-profile-bounds},
\[
|P_0(\hat y)|\lesssim(C_1\eps)^2(1-|\hat y|^2)^\kappa.
\]
For $0<\eta<1/4$ and $\omega\in\mathbb S^2$, define
\[
\Gamma_\eta(\omega)
:=\{\hat y\in B_1(0):1-\omega\cdot\hat y<\eta\}.
\]
We first prove
\begin{equation}\label{eq:boundary-cap-uniform-integrability}
\sup_{\substack{\omega\in\mathbb S^2\\1/2\leq\rho\leq1}}
\int_{\Gamma_\eta(\omega)}
K(\rho\omega,\hat y)|P_0(\hat y)|\diff\hat y
\lesssim
(C_1\eps)^2\eta^{\kappa+1}.
\end{equation}
Write
\[
\hat y=\alpha\omega+v,
\qquad v\perp\omega.
\]
On $\Gamma_\eta(\omega)$ one has $1-\eta<\alpha<1$, and
\[
1-|\hat y|^2=1-\alpha^2-|v|^2,
\qquad
D(\rho\omega,\hat y)
=(\alpha-\rho)^2+(1-\rho^2)|v|^2.
\]
For fixed $\alpha$, polar integration in the $v$-plane gives
\[
\begin{aligned}
&\int_{|v|<\sqrt{1-\alpha^2}}
\frac{(1-\alpha^2-|v|^2)^\kappa}
{\sqrt{(\alpha-\rho)^2+(1-\rho^2)|v|^2}}\diff v
\\
&\qquad\leq2\pi(1-\alpha^2)^\kappa
\int_0^{\sqrt{1-\alpha^2}}
\frac{r\diff r}{\sqrt{(\alpha-\rho)^2+(1-\rho^2)r^2}}
\\
&\qquad=
\frac{2\pi(1-\alpha^2)^{\kappa+1}}
{1-\alpha\rho+|\alpha-\rho|}.
\end{aligned}
\]
Here we used
\[
(\alpha-\rho)^2+(1-\rho^2)(1-\alpha^2)=(1-\alpha\rho)^2.
\]
Since $1-\eta<\alpha<1$, one has $1-\alpha^2\leq2(1-\alpha)$ and
$1-\alpha\rho\geq1-\alpha$ for $1/2\leq\rho\leq1$.
Thus the transverse $v$-integral is bounded by $C_\kappa(1-\alpha)^\kappa$.
Integrating in $\alpha$ yields
\[
\int_{\Gamma_\eta(\omega)}
K(\rho\omega,\hat y)|P_0(\hat y)|\diff\hat y
\lesssim
(C_1\eps)^2\int_{1-\eta}^1(1-\alpha)^\kappa\diff\alpha
\lesssim
(C_1\eps)^2\eta^{\kappa+1},
\]
uniformly in $\omega$ and $1/2\leq\rho\leq1$. This proves
\eqref{eq:boundary-cap-uniform-integrability}. Taking $\rho=1$ controls the
only possible singularity of the boundary kernel; away from the cap,
$(1-\omega\cdot\hat y)^{-1}\leq\eta^{-1}$. Hence the integral in
\eqref{eq:positive-kernel-null} is absolutely convergent.

On $B_1\setminus\Gamma_\eta(\omega)$ we have
$1-\omega\cdot\hat y\geq\eta$. Recall that
\[
K(z,\hat y)=D(z,\hat y)^{-1/2}.
\]
By \eqref{eq:def-positive-kernel-discriminant}, with
$\alpha=\omega\cdot\hat y$,
\[
\begin{aligned}
\big|D(\rho\omega,\hat y)-(1-\omega\cdot\hat y)^2\big|
&=(1-\rho)\left|2\alpha-(1+\rho)(1-|\hat y|^2+\alpha^2)\right|
\\
&\leq4(1-\rho),
\end{aligned}
\]
where $|\alpha|\leq1$ and $0\leq1-|\hat y|^2+\alpha^2\leq1$.
Hence, for fixed $\eta>0$ and $0\leq1-\rho\leq\eta^2/8$,
\[
D(\rho\omega,\hat y)
\geq(1-\omega\cdot\hat y)^2-4(1-\rho)
\geq\frac12\eta^2
\]
on $B_1\setminus\Gamma_\eta(\omega)$. To compare the kernels, write
$A:=1-\omega\cdot\hat y$ and $D_\rho:=D(\rho\omega,\hat y)$. Then
\[
\begin{aligned}
\left|K(\rho\omega,\hat y)-\frac1{1-\omega\cdot\hat y}\right|
&=\frac{|A^2-D_\rho|}
{A\sqrt{D_\rho}\bigl(A+\sqrt{D_\rho}\bigr)}
\\
&\leq4\sqrt2(1-\rho)\eta^{-3},
\end{aligned}
\]
since $A\geq\eta$ and $\sqrt{D_\rho}\geq\eta/\sqrt2$. It follows that
\[
K(\rho\omega,\hat y)
\longrightarrow K(\omega,\hat y)
=\frac1{1-\omega\cdot\hat y}
\]
uniformly in $(\omega,\hat y)$ on the complement of the cap.

For $1/2\leq\rho<1$, we therefore have
\[
\begin{aligned}
4\pi|U(\rho\omega)-U_\infty(\omega)|
\leq{}&
\int_{\Gamma_\eta(\omega)}
\left(K(\rho\omega,\hat y)+\frac1{1-\omega\cdot\hat y}\right)
|P_0(\hat y)|\diff\hat y
\\
&+\int_{B_1\setminus\Gamma_\eta(\omega)}
\left|K(\rho\omega,\hat y)-\frac1{1-\omega\cdot\hat y}\right|
|P_0(\hat y)|\diff\hat y.
\end{aligned}
\]
The first term is bounded uniformly by
$C(C_1\eps)^2\eta^{\kappa+1}$ by
\eqref{eq:boundary-cap-uniform-integrability}. For fixed $\eta$, the second
term tends to zero uniformly in $\omega$ by the preceding uniform convergence
and $P_0\in L^1(B_1)$. Therefore
\[
\limsup_{\rho\uparrow1}
\sup_{\omega\in\mathbb S^2}
|U(\rho\omega)-U_\infty(\omega)|
\lesssim
(C_1\eps)^2\eta^{\kappa+1}.
\]
Letting $\eta\downarrow0$ proves
\eqref{eq:uniform-radial-profile-limit}.

For every fixed $0<\rho<1$, the map
$\omega\mapsto U(\rho\omega)$ is continuous. Since
\eqref{eq:uniform-radial-profile-limit} is uniform in $\omega$, it follows
that $U_\infty\in C(\mathbb S^2)$. Finally, if
$z_n=\rho_n\omega_n\to\omega\in\mathbb S^2$, then
\[
|U(z_n)-U_\infty(\omega)|
\leq
|U(\rho_n\omega_n)-U_\infty(\omega_n)|
+|U_\infty(\omega_n)-U_\infty(\omega)|
\longrightarrow0.
\]
This proves the continuous extension and
\eqref{eq:positive-kernel-closed-ball}.
\end{proof}

\begin{corollary}[Identification of the final null profile]\label{cor:full-source-null-profile}
Under the assumptions of Theorem~\ref{thm:scalar-induced-wave-profile}, for every
$\omega\in\mathbb S^2$ and $\hat q\geq2\hat t_{\mathrm{init}}$,
\begin{equation}\label{eq1-2026-8-18}
\hat{\Rcal}[W](\omega,\hat q)
=U_\infty(\omega)
+O_\delta\bigl((C_1\eps)^2\hat q^{-1+9\delta}\bigr).
\end{equation}
Furthermore, for $u$ defined by \eqref{eq3'-02-dec-2025}, we have
\begin{equation}\label{eq2-26-8-17}
\hat{\Rcal}[u](\omega,\hat q)
=U_\infty(\omega)
+O_\delta\bigl(C_1\eps\hat q^{-1+9\delta}\bigr).
\end{equation}
\end{corollary}

\begin{proof}
Set $\hat t=\hat r+\hat q$ and $\hat x=\hat r\omega$. From
\eqref{eq:full-source-kirchhoff-split} and
\eqref{eq:leading-source-exact-profile},
\[
W(\hat r+\hat q,\hat r\omega)
=(\hat r+\hat q)^{-1}
U\left(\frac{\hat r}{\hat r+\hat q}\omega\right)
+\mathcal G[R_0](\hat r+\hat q,\hat r\omega).
\]
By Proposition~\ref{prop:closed-ball-extension}, the radiation limit of
the first term is $U_\infty(\omega)$.  The existence of
$\hat\Rcal[W]$ was proved in Proposition~\ref{prop:radiation-field-existence}; hence the radiation
limit of the second term exists as the difference.
Moreover, \eqref{eq:retarded-R0-bound} gives
\[
\limsup_{\hat r\to\infty}
\hat r\,|\mathcal G[R_0](\hat r+\hat q,\hat r\omega)|
\lesssim_\delta(C_1\eps)^2\hat q^{-1+9\delta}.
\]
This proves \eqref{eq1-2026-8-18}. Finally, combining this estimate with
\eqref{eq1-26-8-17} gives \eqref{eq2-26-8-17}.
\end{proof}

\subsection{The limiting Klein--Gordon energy}
\label{subsec:kg-energy-profile-main}

We next express the limiting flat Klein--Gordon energy in terms of
the source profile $P_0$.

\begin{proposition}[Limiting Klein--Gordon energy]
\label{prop:kg-energy-profile}
The limit
\[
E_\infty[\phi]
:=\lim_{t\to\infty}\Ebf_c^{\flat}(t,\phi)
\]
exists and satisfies
\begin{equation}\label{eq:energy-profile-identity}
E_\infty[\phi]
=\frac1{8\pi}\int_{B_1}
\frac{P_0(\hat y)}{1-|\hat y|^2}\diff\hat y.
\end{equation}
\end{proposition}

\begin{proof}
We work in the hatted variables used above. Since $\hat t$ depends only on
$t$ and $\hat x=x$, each hypersurface $\{t=t_0\}$ coincides with the
corresponding hypersurface $\{\hat t=\hat t_0\}$, where
$\hat t_0=t_0-2m_{\mathrm S}\ln t_0$, and $d\hat x=dx$.
Moreover,
\begin{equation}\label{eq:app-energy-coordinate-comparison}
\del_t=(1-2m_{\mathrm S}t^{-1})\hat\del_0=:B\hat\del_0,
\qquad
\del_a=\hat\del_a,
\end{equation}
so $B^2=1+O(m_{\mathrm S}t^{-1})$.

\paragraph{The energy decomposition and the error estimates.}

Recall the flat energy
\begin{equation}\label{eq:app-flat-energy-definition}
\Ebf_c^{\flat}(t,\phi)
=\int_{|x|<t-1}
\left(\sum_{\alpha=0}^3|\del_\alpha\phi|^2+c^2\phi^2\right)(t,x)\diff x.
\end{equation}
On a constant-$\hat t$ slice, set
\[
\hat y:=\frac{\hat x}{\hat t},
\qquad
\vartheta:=\sqrt{1-|\hat y|^2},
\qquad
D_t:=\{|x|<t-1,\ \hat t-\hat r\geq1/2\}.
\]
By Corollary~\ref{coro1-22-june-2026}, the scalar field is supported in
$D_t$, so the flat energy may be restricted to this region.
Since $\phi$ is real-valued, $b^-=\overline{b^+}$. Recall from
Theorem~\ref{thm:refined-kg-asymptotics} the profiles
\[
\mathcal A=e^{-i\Phi}b^++e^{i\Phi}\overline{b^+},
\qquad
\mathcal D=-ci e^{-i\Phi}b^++ci e^{i\Phi}\overline{b^+},
\]
evaluated here at $(\hat s,\hat y)=(\hat t\vartheta,\hat y)$.
Define the corresponding leading terms by
\[
\phi_{\mathrm{as}}:=\hat t^{-3/2}\vartheta^{-3/2}\mathcal A,
\qquad
\psi_{\mathrm{as}}:=\hat t^{-3/2}\vartheta^{-5/2}\mathcal D,
\]
and write
\begin{equation}\label{eq:app-leading-phi-time}
\phi=\phi_{\mathrm{as}}+e_\phi,
\qquad
\hat\del_0\phi=\psi_{\mathrm{as}}+e_0.
\end{equation}
Set
\[
\widehat{\delu}_a:=\hat t^{-1}\hat L_a
=\hat\del_a+\hat y_a\hat\del_0.
\]
The spatial derivatives satisfy
\begin{equation}\label{eq:app-spatial-leading}
\hat\del_a\phi=\widehat{\delu}_a\phi-\hat y_a\hat\del_0\phi.
\end{equation}
Substituting these identities and
\eqref{eq:app-energy-coordinate-comparison} into the energy gives the
exact decomposition
\begin{equation}\label{eq:app-energy-exact-decomposition}
\begin{aligned}
\Ebf_c^{\flat}(t,\phi)
={}&\int_{D_t}\left[(1+|\hat y|^2)\psi_{\mathrm{as}}^2
+c^2\phi_{\mathrm{as}}^2\right]\diff x
\\
&+\mathcal E_1(t)+\mathcal E_2(t)+\mathcal E_3(t),
\end{aligned}
\end{equation}
where
\[
\mathcal E_1(t)
:=\int_{D_t}(B^2-1)|\hat\del_0\phi|^2\diff x,
\]
\[
\mathcal E_2(t)
:=\int_{D_t}\left[
\sum_{a=1}^3|\widehat{\delu}_a\phi|^2
-2\hat\del_0\phi\sum_{a=1}^3\hat y_a\widehat{\delu}_a\phi\right]\diff x,
\]
and
\[
\begin{aligned}
\mathcal E_3(t):=\int_{D_t}\Big[
&(1+|\hat y|^2)\bigl(2\psi_{\mathrm{as}}e_0+e_0^2\bigr)
\\
&+c^2\bigl(2\phi_{\mathrm{as}}e_\phi+e_\phi^2\bigr)
\Big]\diff x.
\end{aligned}
\]
These terms account for the coordinate correction, the tangential
derivatives, and the profile remainders, respectively. We prove that
each tends to zero. All $L^2$ norms below are taken over $D_t$, unless
another domain is specified.

For $\mathcal E_1$, the bound $|B^2-1|\lesssim m_{\mathrm S}t^{-1}$ and
$B\simeq1$ give
\[
|\mathcal E_1(t)|
\leq\int_{D_t}|B^2-1|\,|\hat\del_0\phi|^2\diff x
\lesssim m_{\mathrm S}t^{-1}\Ebf_c^{\flat}(t,\phi)=o(1),
\]
by the uniform flat-energy bound \eqref{eq6-21-june-2026}.

For $\mathcal E_2$, \eqref{eq:hyperbolic-gradient-L2}, the coordinate
relations \eqref{eq3-04-june-2026}, and \eqref{eq6-21-june-2026} give
\begin{equation}\label{eq:app-tangential-energy-small}
\int_{|x|<t-1}|\widehat{\delu}_a\phi|^2\diff x
\lesssim (C_1\eps)^2t^{-2+2\delta}=o(1).
\end{equation}
Since $|\hat y|\leq1$ on $D_t$, Cauchy--Schwarz yields
\[
|\mathcal E_2(t)|
\leq\sum_{a=1}^3\|\widehat{\delu}_a\phi\|_{L^2}^2
+2\|\hat\del_0\phi\|_{L^2}
\left(\sum_{a=1}^3\|\widehat{\delu}_a\phi\|_{L^2}^2\right)^{1/2}
=o(1).
\]

To estimate $\mathcal E_3$, we first bound $e_\phi$ and $e_0$ in $L^2$.
Substituting $\hat s=\hat t\vartheta$ into
\eqref{eq:simultaneous-phi-error} and \eqref{eq:simultaneous-dt-error} gives
\begin{equation}\label{eq:app-profile-remainders-plane}
\begin{aligned}
|e_\phi|
&\lesssim C_1\eps\hat t^{-5/2+9\delta}
\vartheta^{-1+25\delta/2},
\\
|e_0|
&\lesssim C_1\eps\left(
\hat t^{-5/2+9\delta}\vartheta^{-2+25\delta/2}
+\hat t^{-5/2+3\delta}\vartheta^{-2+2\delta}
\right).
\end{aligned}
\end{equation}
On $D_t$,
\begin{equation}\label{eq:app-support-vartheta-lower}
\vartheta^2
=\frac{(\hat t-\hat r)(\hat t+\hat r)}{\hat t^2}
\geq\frac{\hat t-\hat r}{\hat t}.
\end{equation}
Since $\hat x=x$, we have $d\hat x=dx=r^2\diff r\diff\omega$ and
$\hat r=r$. Moreover, $0\leq r\leq\hat t-1/2$ on $D_t$, and
\eqref{eq:app-support-vartheta-lower} implies, for $p>0$,
\[
\vartheta^{-p}\leq\hat t^{p/2}(\hat t-r)^{-p/2}.
\]
For the first estimate in \eqref{eq:app-profile-remainders-plane}, if
$0<\delta<2/25$ we use this bound with $p=2-25\delta$, while if
$2/25\leq\delta<1/9$ we use $\vartheta^{-2+25\delta}\leq1$.
In both cases,
\begin{equation}\label{eq:app-profile-remainders-L2}
\|e_\phi\|_{L^2}^2
\lesssim (C_1\eps)^2\hat t^{-2+18\delta}=o(1).
\end{equation}
Similarly, after squaring the two terms in the second line of
\eqref{eq:app-profile-remainders-plane}, for $0<\delta<1/9$ we have
\[
\hat t^{-1+11\delta/2}
\int_0^{\hat t-1/2}(\hat t-r)^{-2+25\delta/2}\diff r=o(1),
\qquad
\hat t^{-1+4\delta}
\int_0^{\hat t-1/2}(\hat t-r)^{-2+2\delta}\diff r=o(1).
\]
Hence
\begin{equation}\label{eq:app-time-remainder-L2}
\|e_0\|_{L^2}^2=o(1).
\end{equation}
The leading terms have bounded $L^2$ norms. Indeed, by
\eqref{eq:kg-profile-bound-main} and $d\hat x=\hat t^3\diff\hat y$,
\[
\|\phi_{\mathrm{as}}\|_{L^2}^2
+\|\psi_{\mathrm{as}}\|_{L^2}^2
\lesssim\int_{B_1}\vartheta^{-5}|b^+(\hat y)|^2\diff\hat y
\lesssim_\delta(C_1\eps)^2.
\]
Consequently, Cauchy--Schwarz gives
\[
\begin{aligned}
|\mathcal E_3(t)|\lesssim{}&
\|\psi_{\mathrm{as}}\|_{L^2}\|e_0\|_{L^2}
+\|e_0\|_{L^2}^2
\\
&+\|\phi_{\mathrm{as}}\|_{L^2}\|e_\phi\|_{L^2}
+\|e_\phi\|_{L^2}^2=o(1).
\end{aligned}
\]
Thus all three error terms in \eqref{eq:app-energy-exact-decomposition}
are $o(1)$.

We next extend the leading integral to $B_1$. Under
$\hat y=\hat x/\hat t$, the region $D_t$ is the ball
$|\hat y|<R_t$, where
\[
R_t=\min\left\{\frac{t-1}{\hat t},\ 1-\frac1{2\hat t}\right\}
\geq1-\frac1{\hat t}.
\]
Hence the omitted region is a shrinking boundary layer with
$\vartheta\lesssim\hat t^{-1/2}$. The leading density satisfies
\[
(1+|\hat y|^2)\vartheta^{-5}\mathcal D^2
+c^2\vartheta^{-3}\mathcal A^2
\lesssim\vartheta^{-5}|b^+(\hat y)|^2
\lesssim(C_1\eps)^2\vartheta^{-11\delta}.
\]
Since $d\hat y\simeq\vartheta\diff\vartheta\diff\omega$ near
$\partial B_1$, the error from extending the integral is bounded by
\[
\int_{\vartheta\lesssim\hat t^{-1/2}}
\vartheta^{-5}|b^+(\hat y)|^2\diff\hat y
\lesssim_\delta(C_1\eps)^2\hat t^{-1+11\delta/2}=o(1).
\]

In this extension to $B_1$, the phase uses the real continuous extension
specified in Theorem~\ref{thm:refined-kg-asymptotics}. Its exponential
has modulus one, so the density bound is independent of the convention
on the omitted layer. Every fixed $K_{\rho,\eta}$ satisfies
$\hat t\vartheta\geq s_0$ for all sufficiently large $\hat t$.
Thus the later differentiations and integrations by parts on that compact
set use exactly the original phase.
Combining this with \eqref{eq:app-energy-exact-decomposition}, we obtain
\begin{equation}\label{eq:app-leading-energy-expanded}
\Ebf_c^{\flat}(t,\phi)
=\int_{B_1}\left[
(1+|\hat y|^2)\vartheta^{-5}\mathcal D^2
+c^2\vartheta^{-3}\mathcal A^2
\right]\diff\hat y+o(1).
\end{equation}

We now expand the two profiles:
\[
\begin{aligned}
\mathcal D^2
&=-c^2e^{-2i\Phi}(b^+)^2
-c^2e^{2i\Phi}\overline{b^+}^{\,2}
+2c^2|b^+|^2,\\
\mathcal A^2
&=e^{-2i\Phi}(b^+)^2
+e^{2i\Phi}\overline{b^+}^{\,2}
+2|b^+|^2.
\end{aligned}
\]
Since $\vartheta^2=1-|\hat y|^2$, the integrand in
\eqref{eq:app-leading-energy-expanded} reduces to
\begin{equation}\label{eq:app-energy-nonosc-plus-osc}
4c^2\vartheta^{-5}|b^+(\hat y)|^2
-2c^2|\hat y|^2\vartheta^{-5}
\left(
 e^{-2i\Phi}(b^+(\hat y))^2
 +e^{2i\Phi}\overline{b^+(\hat y)}^{\,2}
\right).
\end{equation}
The first term is integrable by \eqref{eq:kg-profile-bound-main}.
It remains to prove that the two conjugate oscillatory integrals tend to zero.

\paragraph{Vanishing of the oscillatory terms.}

The weights in the plane energy leave the oscillatory terms in
\eqref{eq:app-energy-nonosc-plus-osc}, whereas those in the scalar source
give the pointwise cancellation \eqref{eq:leading-source-cancellation}.
We estimate the energy integrals by radial integration by parts.
For the related treatment of oscillatory source terms in
Kirchhoff's formula, see \cite[Section~4.2.2]{Chen-Lindblad-2023}.
On the plane $\hat t=\mathrm{const}$, recall that
$\vartheta(\hat y)=\sqrt{1-|\hat y|^2}$ and write
\begin{equation}\label{eq:app-phase-on-plane}
\Phi_{\hat t}(\hat y)
:=c\int_{\hat t_{\mathrm{init}}}^{\hat t\vartheta(\hat y)}a(\sigma,\hat y)\diff\sigma
=c\hat t\vartheta(\hat y)-c\hat t_{\mathrm{init}}+\chi_{\hat t}(\hat y),
\end{equation}
where
\begin{equation}\label{eq:app-phase-correction}
\chi_{\hat t}(\hat y)
:=c\int_{\hat t_{\mathrm{init}}}^{\hat t\vartheta(\hat y)}(a(\sigma,\hat y)-1)\diff\sigma.
\end{equation}
Define
\[
I_{\mathrm{osc}}(\hat t)
:=\int_{B_1}
|\hat y|^2\vartheta^{-5}(b^+(\hat y))^2
 e^{-2i\Phi_{\hat t}(\hat y)}\diff\hat y.
\]
By \eqref{eq:app-leading-energy-expanded} and
\eqref{eq:app-energy-nonosc-plus-osc},
\begin{equation}\label{eq:app-energy-oscillatory-decomposition}
\Ebf_c^{\flat}(t,\phi)
=4c^2\int_{B_1}\vartheta^{-5}|b^+(\hat y)|^2\diff\hat y
-4c^2\operatorname{Re}I_{\mathrm{osc}}(\hat t)+o(1).
\end{equation}
It therefore suffices to prove
\begin{equation}\label{eq:app-oscillatory-goal}
I_{\mathrm{osc}}(\hat t)\longrightarrow0
\qquad\text{as }\hat t\to+\infty.
\end{equation}

By \eqref{eq:kg-profile-bound-main},
\begin{equation}\label{eq:app-osc-amplitude-bound}
|\hat y|^2\vartheta^{-5}|b^+(\hat y)|^2
\lesssim (C_1\eps)^2|\hat y|^2\vartheta^{-11\delta},
\end{equation}
so the amplitude of $I_{\mathrm{osc}}$ belongs to $L^1(B_1)$.
Fix $0<\rho,\eta<1/4$ and decompose
\[
B_1=\{|\hat y|<\rho\}\cup\{\vartheta<\eta\}\cup K_{\rho,\eta},
\qquad
K_{\rho,\eta}:=\{|\hat y|\geq\rho,\ \vartheta\geq\eta\}.
\]
On the central region,
\begin{equation}\label{eq:app-center-small}
\int_{|\hat y|<\rho}|\hat y|^2\vartheta^{-5}|b^+|^2\diff\hat y
\lesssim (C_1\eps)^2\rho^5,
\end{equation}
while near the boundary, writing $\vartheta=\sqrt{1-r^2}$ gives
$d\hat y\simeq\vartheta\diff\vartheta\diff\omega$ and hence
\begin{equation}\label{eq:app-boundary-small}
\int_{\vartheta<\eta}|\hat y|^2\vartheta^{-5}|b^+|^2\diff\hat y
\lesssim (C_1\eps)^2\eta^{2-11\delta}.
\end{equation}
On $K_{\rho,\eta}$, for
$\hat y=r\omega$,
\[
\del_r(c\hat t\vartheta)=-c\hat t\frac r\vartheta.
\]
To estimate the correction $\chi_{\hat t}$, we use
Lemma~\ref{lem:app-q-y-derivatives} in
Appendix~\ref{sec:appendix-kg-profile-technical}.
Since $a=\sqrt{1+\mathfrak q}$ and $|\mathfrak q|\leq1/3$, the chain rule and
\eqref{eq:app-q-y-derivatives} give
\begin{equation}\label{eq:app-a-y-derivatives}
|\del_{\hat y}^k a(\lambda,\hat y)|
\lesssim_{\rho,\eta}C_1\eps\,\lambda^{-1+7\delta/2},
\qquad k=1,2.
\end{equation}
By the ray-wise coefficient bounds established in
Appendix~\ref{sec:appendix-kg-profile-technical}, in particular
\eqref{eq:ode-quantitative-bounds} and
\eqref{eq:exterior-phase-coefficient} on the exterior segment, we also have
\begin{equation}\label{eq:app-a-lambda-derivative}
|a-1|\lesssim_{\rho,\eta} C_1\eps\,\lambda^{-1+9\delta/2},
\qquad
|\del_\lambda a|\lesssim_{\rho,\eta}C_1\eps\,\lambda^{-2+9\delta/2}.
\end{equation}
Write $\hat y=r\omega$. Since
\[
\hat t\vartheta(r)\simeq_{\rho,\eta}\hat t,
\qquad
|\del_r(\hat t\vartheta)|+|\del_r^2(\hat t\vartheta)|
\lesssim_{\rho,\eta}\hat t,
\]
differentiating \eqref{eq:app-phase-correction} gives
\begin{equation}\label{eq:app-phase-correction-first}
\del_r\chi_{\hat t}
=c\bigl(a(\hat t\vartheta,r\omega)-1\bigr)\del_r(\hat t\vartheta)
+c\int_{\hat t_{\mathrm{init}}}^{\hat t\vartheta}\del_r a(\sigma,r\omega)\diff\sigma,
\end{equation}
and therefore, by \eqref{eq:app-a-y-derivatives} and
\eqref{eq:app-a-lambda-derivative},
\begin{equation}\label{eq:app-phase-correction-first-bound}
|\del_r\chi_{\hat t}|
\lesssim_{\rho,\eta}C_1\eps\hat t^{9\delta/2}.
\end{equation}
A second derivative gives
\begin{equation}\label{eq:app-phase-correction-second}
\begin{aligned}
\del_r^2\chi_{\hat t}=c\Big(&
(\del_\lambda a)(\hat t\vartheta,r\omega)
\bigl(\del_r(\hat t\vartheta)\bigr)^2
+2(\del_r a)(\hat t\vartheta,r\omega)\del_r(\hat t\vartheta)
\\
&+\bigl(a(\hat t\vartheta,r\omega)-1\bigr)\del_r^2(\hat t\vartheta)
+\int_{\hat t_{\mathrm{init}}}^{\hat t\vartheta}\del_r^2a(\sigma,r\omega)\diff\sigma\Big),
\end{aligned}
\end{equation}
so the same estimates yield
\begin{equation}\label{eq:app-phase-correction-second-bound}
|\del_r^2\chi_{\hat t}|
\lesssim_{\rho,\eta}C_1\eps\hat t^{9\delta/2}.
\end{equation}
Since also
\[
\del_r^2(c\hat t\vartheta)=O_{\rho,\eta}(\hat t),
\]
these bounds and $9\delta/2<1$ imply, for sufficiently large $\hat t$,
\begin{equation}\label{eq:app-phase-nonstationary}
|\del_r\Phi_{\hat t}|\gtrsim_{\rho,\eta}\hat t,
\qquad
|\del_r^2\Phi_{\hat t}|\lesssim_{\rho,\eta}\hat t.
\end{equation}
Thus the phase is non-stationary on $K_{\rho,\eta}$.

Since only continuity of $b^+$ has been established, we approximate the
amplitude in $L^1$. For fixed
$\rho,\eta$ and any $\varepsilon_0>0$, the $L^1$ integrability in
\eqref{eq:app-osc-amplitude-bound} allows us to choose
$G\in C_c^1(\operatorname{int}K_{\rho,\eta})$ such that
\[
\left\|
|\hat y|^2\vartheta^{-5}(b^+)^2-G
\right\|_{L^1(K_{\rho,\eta})}
<\varepsilon_0.
\]
For the smooth amplitude $G$, polar coordinates give
\[
\int_{K_{\rho,\eta}}G(\hat y)e^{-2i\Phi_{\hat t}(\hat y)}\diff\hat y
=
\int_{\mathbb S^2}\int_{\rho}^{\sqrt{1-\eta^2}}
 r^2G(r\omega)e^{-2i\Phi_{\hat t}(r,\omega)}\diff r\diff\omega.
\]
Since
\[
e^{-2i\Phi_{\hat t}}
=-\frac1{2i\del_r\Phi_{\hat t}}\del_r e^{-2i\Phi_{\hat t}},
\]
one integration by parts in $r$ gives no boundary term and yields
\[
\begin{aligned}
&\int_{\rho}^{\sqrt{1-\eta^2}}
 r^2G(r\omega)e^{-2i\Phi_{\hat t}(r,\omega)}\diff r
\\
&\qquad
=\int_{\rho}^{\sqrt{1-\eta^2}}
\del_r\left(
\frac{r^2G(r\omega)}{2i\del_r\Phi_{\hat t}(r,\omega)}
\right)e^{-2i\Phi_{\hat t}(r,\omega)}\diff r.
\end{aligned}
\]
By \eqref{eq:app-phase-nonstationary},
\[
\left|
\del_r\left(
\frac{r^2G}{\del_r\Phi_{\hat t}}
\right)
\right|
\lesssim_{\rho,\eta}
\hat t^{-1}\bigl(|G|+|\del_rG|\bigr),
\]
and therefore
\begin{equation}\label{eq:app-osc-annulus-ibp}
\left|
\int_{K_{\rho,\eta}}G(\hat y)e^{-2i\Phi_{\hat t}(\hat y)}\diff\hat y
\right|
\lesssim_{\rho,\eta}
\hat t^{-1}
\left(\|G\|_{L^1}+\|\del_rG\|_{L^1}\right)
\longrightarrow0.
\end{equation}
The approximation error is uniform in $\hat t$, since
$|e^{-2i\Phi_{\hat t}}|=1$:
\[
\left|
\int_{K_{\rho,\eta}}
\left(|\hat y|^2\vartheta^{-5}(b^+)^2-G\right)
 e^{-2i\Phi_{\hat t}}\diff\hat y
\right|
<\varepsilon_0.
\]
Hence, for fixed $\rho$ and $\eta$,
\[
\limsup_{\hat t\to\infty}
\left|
\int_{K_{\rho,\eta}}
|\hat y|^2\vartheta^{-5}(b^+)^2
 e^{-2i\Phi_{\hat t}}\diff\hat y
\right|
\leq\varepsilon_0.
\]
Since $\varepsilon_0$ is arbitrary, the oscillatory integral over
$K_{\rho,\eta}$ tends to zero. Combining this with
\eqref{eq:app-center-small} and \eqref{eq:app-boundary-small} gives
\[
\limsup_{\hat t\to\infty}|I_{\mathrm{osc}}(\hat t)|
\lesssim
(C_1\eps)^2\left(\rho^5+\eta^{2-11\delta}\right).
\]
Letting $\rho,\eta\downarrow0$ after the limit $\hat t\to\infty$ proves
\eqref{eq:app-oscillatory-goal}.

By \eqref{eq:app-energy-oscillatory-decomposition} and
\eqref{eq:app-oscillatory-goal}, we obtain
\[
\Ebf_c^{\flat}(t,\phi)
=4c^2\int_{B_1}
(1-|\hat y|^2)^{-5/2}|b^+(\hat y)|^2\diff\hat y+o(1).
\]
Since $\hat t\to\infty$ together with $t\to\infty$, letting $t\to\infty$ and
using \eqref{eq:def-new-P0}, namely
\[
\frac1{8\pi}\frac{P_0(\hat y)}{1-|\hat y|^2}
=4c^2(1-|\hat y|^2)^{-5/2}|b^+(\hat y)|^2,
\]
we obtain
\[
E_\infty[\phi]
=4c^2\int_{B_1}
(1-|\hat y|^2)^{-5/2}|b^+(\hat y)|^2\diff\hat y
=\frac1{8\pi}\int_{B_1}
\frac{P_0(\hat y)}{1-|\hat y|^2}\diff\hat y,
\]
which proves \eqref{eq:energy-profile-identity}.

\end{proof}

\subsection{Quantitative single-direction observability}

By \eqref{eq:positive-kernel-closed-ball} and
\eqref{eq:energy-profile-identity},
\begin{equation}\label{eq:two-positive-representations}
U(z)=\frac1{4\pi}\int_{B_1}K(z,\hat y)P_0(\hat y)\diff\hat y,
\qquad
E_\infty[\phi]=\frac1{8\pi}\int_{B_1}
\frac{P_0(\hat y)}{1-|\hat y|^2}\diff\hat y.
\end{equation}
The lower bound for $K$ and the decay of $P_0$ near $\partial B_1$
give the following estimate.

\begin{proposition}[Quantitative single-direction observability]
\label{prop:closed-ball-observability}
Let
\[
\kappa:=1-\frac{11}{2}\delta>0.
\]
Then, for every $z\in\overline{B_1(0)}$,
\begin{equation}\label{eq:holder-observability-initial}
\Ebf_c^{\flat}(2,\phi)
\lesssim_\delta
(C_1\eps)^{\frac{2}{\kappa+1}}
U(z)^{\frac{\kappa}{\kappa+1}}.
\end{equation}
\end{proposition}

\begin{proof}
By \eqref{eq:D-distance-comparison},
\begin{equation}\label{eq:kernel-uniform-lower}
K(z,\hat y)=D(z,\hat y)^{-1/2}
\geq |\hat y-z|^{-1}\geq\frac12,
\qquad
z\in\overline{B_1(0)},\quad \hat y\in B_1(0).
\end{equation}
Thus
\begin{equation}\label{eq:P0-L1-by-profile}
\int_{B_1}P_0(\hat y)\diff\hat y\leq8\pi U(z).
\end{equation}

By \eqref{eq:modified-source-profile-bounds},
\[
0\leq P_0(\hat y)\lesssim
(C_1\eps)^2(1-|\hat y|^2)^\kappa.
\]
For $0<\eta<1/2$, split the integral in
\eqref{eq:energy-profile-identity} into
$\{1-|\hat y|^2\geq\eta\}$ and $\{1-|\hat y|^2<\eta\}$. On the first
region, \eqref{eq:P0-L1-by-profile} gives
\[
\int_{1-|\hat y|^2\geq\eta}
\frac{P_0(\hat y)}{1-|\hat y|^2}\diff\hat y
\lesssim \eta^{-1}U(z).
\]
On the boundary layer, write $\hat y=\rho\theta$. Since
$1-\rho^2<\eta<1/2$, one has $\rho\simeq1$ and
$\diff(1-\rho^2)=-2\rho\diff\rho$, so
$\diff\hat y=\rho^2\diff\rho\diff\theta\simeq
|\diff(1-\rho^2)|\diff\theta$. Using the profile bound and $\kappa>0$, we
obtain
\[
\int_{1-|\hat y|^2<\eta}
\frac{P_0(\hat y)}{1-|\hat y|^2}\diff\hat y
\lesssim
(C_1\eps)^2\int_0^\eta \sigma^{\kappa-1}\diff\sigma
\lesssim_\delta (C_1\eps)^2\eta^\kappa.
\]
Therefore
\begin{equation}\label{eq:energy-profile-eta}
E_\infty[\phi]
\lesssim_\delta
\eta^{-1}U(z)+(C_1\eps)^2\eta^\kappa.
\end{equation}

If $0<U(z)\leq(C_1\eps)^2$, choose
\[
\eta\simeq
\left(\frac{U(z)}{(C_1\eps)^2}\right)^{\frac1{\kappa+1}},
\]
with the implicit constant small enough that $\eta<1/2$. This gives
\[
E_\infty[\phi]
\lesssim_\delta
(C_1\eps)^{\frac{2}{\kappa+1}}
U(z)^{\frac{\kappa}{\kappa+1}}.
\]
If $U(z)=0$, the same conclusion follows by letting $\eta\downarrow0$ in
\eqref{eq:energy-profile-eta}. If $U(z)>(C_1\eps)^2$, the direct bound
$E_\infty[\phi]\lesssim_\delta(C_1\eps)^2$, which follows from
\eqref{eq:energy-profile-identity} and the profile bound, gives the same
estimate. Finally, the inverse flat-energy comparison
\eqref{eq6-21-june-2026} yields
\[
\Ebf_c^{\flat}(2,\phi)\lesssim E_\infty[\phi],
\]
and hence \eqref{eq:holder-observability-initial}.
\end{proof}

For $z=\omega\in\partial B_1(0)$, this estimate bounds the initial
scalar energy in terms of $U_\infty(\omega)$ at a single null direction.

\subsection{Proof of Theorem~\ref{thm1-15-aug-2026}}
We first translate the asymptotic profile of $W$ into that of $\hbar$. By \eqref{eq3'-02-dec-2025} and \eqref{eq1-23-7-2026},
\begin{equation}\label{eq:trace-W-comparison}
|\hbar(\hat t,\hat x)-W(\hat t,\hat x)|
\lesssim |\phi(\hat t,\hat x)|^2+|u_1(\hat t,\hat x)|
\lesssim_\delta C_1\eps\hat t^{-1}\langle\hat t-\hat r\rangle^{-1+8\delta},
\end{equation}
where \eqref{eq7-03-aug-2026} and \eqref{eq2-2026-8-3} are applied. Combining this with Theorem~\ref{thm:scalar-induced-wave-profile} gives \eqref{eq4-19-08-2026}. Proposition~\ref{prop:closed-ball-extension} gives the continuous extension $U\in C(\overline{B_1(0)})$ and the uniform radial convergence to $U_\infty$, proving the closed-ball profile assertions in the theorem.

To identify the radiation limit of $\hbar$, recall that
$u=\hbar+12\pi\phi^2$ by \eqref{eq3'-02-dec-2025}. By
\eqref{eq7-03-aug-2026} in $\mathcal K$ and vanishing of $\phi$ in the
exterior, $\hat r\,|\phi(\hat r+\hat q,\hat r\omega)|^2\to0$ for each
fixed $\hat q\geq1/2$. Proposition~\ref{prop:radiation-field-existence}
therefore gives existence of $\hat\Rcal[\hbar]=\hat\Rcal[u]$ in that
range. By Corollary~\ref{cor:full-source-null-profile},
\[
\hat\Rcal[\hbar](\omega,\hat q)
=U_\infty(\omega)+O_\delta\bigl(C_1\eps\hat q^{-1+9\delta}\bigr),
\qquad \hat q\geq2\hat t_{\mathrm{init}},
\]
with the error uniform in $\omega$. This proves \eqref{eq:main-final-null-profile}.

Proposition~\ref{prop:closed-ball-observability} gives
\eqref{eq:main-quantitative-observability}. If $U(z_0)=0$, the estimate
forces the scalar initial data to vanish, and uniqueness gives
$\phi\equiv0$. In that case the scalar source and hence $W$ vanish, so
$U\equiv0$.

\appendix
\section{Integration along Characteristics}\label{sec1-31-july-2026}
\subsection{Technical preparation}
\begin{lemma}\label{lem1-19-jan-2026}
	Let $u\in C^2\bigl(\{(t,x):t\geq t_{\mathrm{init}}\}\bigr)$, with
	$t_{\mathrm{init}}\geq1$, and write
	\[
		u_0(x):=u(t_{\mathrm{init}},x),
		\qquad
		u_1(x):=\del_tu(t_{\mathrm{init}},x).
	\]
	Assume that there exist constants
	$C_{\mathrm{int}},C_{\mathrm{ext}},C_0>0$,
	$0<a\leq1$, $0<b<1$ and $d\geq2+a$ such that
	\begin{equation}\label{eq3-08-jan-2026}
		\left|
		(\del_t-\del_r)
		(\del_t+\del_r)
		(ru)
		\right|
		\leq
		\begin{cases}
			C_{\mathrm{int}}(t+r)^{-1-a}\langle t-r\rangle^{-b},
			\quad &r\leq t-1/2;
			\\
			C_{\mathrm{ext}}(t+r)^{-d},
			\quad &r\geq t-1/2.
		\end{cases}
	\end{equation}
	Assume furthermore that
	\begin{equation}\label{eq1-21-aug-2026-B}
		|u_0(r\omega)|
		\lesssim C_0\langle r\rangle^{-1},
		\qquad
		\left|
		\del_r(ru_0(r\omega))+ru_1(r\omega)
		\right|
		\lesssim C_0\langle r\rangle^{-1-a}.
	\end{equation}
	Then, for $(t,x)\in\{r\leq t-1/2\}$,
	\begin{equation}\label{eq5-02-aug-2026}
		|u(t,x)|
		\lesssim
		a^{-1}
		\bigl(
		C_0+C_{\mathrm{ext}}+(1-b)^{-1}C_{\mathrm{int}}
		\bigr)
		t^{-1}\langle t-r\rangle^{1-a-b}.
	\end{equation}
	Furthermore, the radiation field
	\begin{equation}\label{eq2-19-jan-2026}
		\Rcal[u](\omega,q)
		:=
		\lim_{r\rightarrow+\infty}ru(r+q,r\omega)
	\end{equation}
	exists for all
	$(\omega,q)\in\mathbb S^2\times[1/2,+\infty)$, and
	\begin{equation}\label{eq6-02-aug-2026}
		\big|\Rcal[u](\omega,q)\big|
		\lesssim
		a^{-1}
		\bigl(
		C_0+C_{\mathrm{ext}}+(1-b)^{-1}C_{\mathrm{int}}
		\bigr)
		\langle q\rangle^{1-a-b}.
	\end{equation}
	Moreover, for $r\geq q$ and $r+q\geq t_{\mathrm{init}}$,
	\begin{equation}\label{eq1-12-aug-2026}
		\begin{aligned}
		\big|
		\Rcal[u](\omega,q)-ru(r+q,r\omega)
		\big|\lesssim
		a^{-1}
		\bigl(
		C_0+C_{\mathrm{ext}}+(1-b)^{-1}C_{\mathrm{int}}
		\bigr)
		\langle q\rangle^{1-b}r^{-a}.
		\end{aligned}
	\end{equation}
\end{lemma}
\begin{proof}
	Set
	\[
		U(t,r,\omega):=ru(t,r\omega).
	\]
	We fix $\omega\in\mathbb S^2$. For
	$(t_1,r_1)\in\{r\leq t-1/2\}$, we first integrate
	\eqref{eq3-08-jan-2026} along the line
	$\{t+r=t_1+r_1\}$. By \eqref{eq1-21-aug-2026-B},
	\[
		\begin{aligned}
		\big| (\del_t+\del_r)
		U(t_{\mathrm{init}},t_1+r_1-t_{\mathrm{init}},\omega)\big|
		&=\left|\del_r(ru_0)+ru_1\right|
		(t_1+r_1-t_{\mathrm{init}},\omega)
		\\
		&\lesssim C_0(t_1+r_1)^{-1-a}.
		\end{aligned}
	\]
	Consequently, if
	\[
		t_1+r_1-t_{\mathrm{init}}
		<\frac{t_1+r_1}{2}-\frac14,
	\]
	the characteristic segment does not meet the exterior region, and the
	$C_{\mathrm{ext}}$-term in the estimate below is omitted. Thus,
	\begin{equation}\label{eq7-02-aug-2026}
		\begin{aligned}
		\big|(
		\del_t+\del_r)U(t_1,r_1,\omega)
		\big|
	     &\lesssim
		C_0(t_1+r_1)^{-1-a}
		+C_{\mathrm{ext}}(t_1+r_1)^{-d}
		\int_{\frac{t_1+r_1}{2}-\frac14}^{t_1+r_1-t_{\mathrm{init}}}
		\diff r
		\\&\quad+C_{\mathrm{int}}(t_1+r_1)^{-1-a}
		\int_{r_1}^{\frac{t_1+r_1}{2}-\frac14}
		\frac{\diff r}
		{\left\langle\frac{t_1+r_1}{2}-r\right\rangle^b}
		\\
		&\lesssim
		\bigl(
		C_0+C_{\mathrm{ext}}+(1-b)^{-1}C_{\mathrm{int}}
		\bigr)
		(t_1+r_1)^{-1-a}
		\langle t_1-r_1\rangle^{1-b}.
		\end{aligned}
	\end{equation}
	Here and below the implicit constants may depend on
	$t_{\mathrm{init}}$.

	Define the starting point on the line
	$\{t-r=t_1-r_1\}$ by
	\[
		(t_*,r_*):=
		\begin{cases}
			(t_1-r_1,0),
			& t_1-r_1\geq t_{\mathrm{init}},
			\\
			(t_{\mathrm{init}},t_{\mathrm{init}}-(t_1-r_1)),
			& t_1-r_1<t_{\mathrm{init}}.
		\end{cases}
	\]
	In the first case, $U(t_*,r_*,\omega)=0$. In the second case,
	\[
		|U(t_*,r_*,\omega)|
		=r_*|u_0(r_*\omega)|
		\lesssim C_0r_*\langle r_*\rangle^{-1}.
	\]
	Since $1/2\leq t_1-r_1<t_{\mathrm{init}}$ and
	$r_*\leq r_1$, this term is bounded by the right-hand sides of
	the estimates below. In the second case the actual integration starts
	at $t=t_{\mathrm{init}}$ and is therefore shorter.

	For
	$(t_1,r_1)\in\{r<t-1/2\}\cap\{r_1\geq t_1/2\}$,
	we integrate \eqref{eq7-02-aug-2026} along
	$\{t-r=t_1-r_1\}$ and obtain
	\begin{equation}\label{eq8-02-aug-2026}
		\begin{aligned}
		|U(t_1,r_1,\omega)|
		\lesssim&
		|U(t_*,r_*,\omega)|+
		\bigl(
		C_0+C_{\mathrm{ext}}+(1-b)^{-1}C_{\mathrm{int}}
		\bigr)
		\langle t_1-r_1\rangle^{1-b}
		\int_{t_1-r_1}^{t_1}
		\frac{\diff t}
		{\bigl(t-\frac{t_1-r_1}{2}\bigr)^{1+a}}
		\\
		\lesssim&
		a^{-1}
		\bigl(
		C_0+C_{\mathrm{ext}}+(1-b)^{-1}C_{\mathrm{int}}
		\bigr)
		\langle t_1-r_1\rangle^{1-a-b}.
		\end{aligned}
	\end{equation}
	This leads to \eqref{eq5-02-aug-2026} in this region. For
	$(t_1,r_1)\in\{r<t-1/2\}\cap\{r_1<t_1/2\}$,
	we similarly obtain
	\[
		\begin{aligned}
		|U(t_1,r_1,\omega)|
		\lesssim&
		a^{-1}
		\bigl(
		C_0+C_{\mathrm{ext}}+(1-b)^{-1}C_{\mathrm{int}}
		\bigr)
		\\
		&\qquad\times
		r_1t_1^{-1-a}
		\langle t_1-r_1\rangle^{1-b}.
		\end{aligned}
	\]
	Since $t_1\simeq t_1-r_1$ in this region, this also gives
	\eqref{eq5-02-aug-2026}.

	Finally, along $t-r=q$, estimate \eqref{eq7-02-aug-2026} gives
	\[
		\big|
		(\del_t+\del_r)U(q+r,r,\omega)
		\big|
		\lesssim
		\bigl(
		C_0+C_{\mathrm{ext}}+(1-b)^{-1}C_{\mathrm{int}}
		\bigr)
		(q+2r)^{-1-a}\langle q\rangle^{1-b}.
	\]
	The right-hand side is integrable in $r$. Hence
	$U(q+r,r,\omega)$ has a limit as $r\rightarrow+\infty$,
	which proves the existence of \eqref{eq2-19-jan-2026}.
	Estimate \eqref{eq6-02-aug-2026} follows from
	\eqref{eq8-02-aug-2026}, while integration on
	$[r,+\infty)$ gives \eqref{eq1-12-aug-2026}.
\end{proof}

\begin{corollary}\label{cor1-02-aug-2026}
	Let $u$ be a $C^2$ solution to
	\begin{equation}\label{eq1-03-aug-2026}
		-\Box u=f,
		\qquad
		u(t_{\mathrm{init}},x)=u_0(x),
		\qquad
		\del_tu(t_{\mathrm{init}},x)=u_1(x),
		\qquad
		t_{\mathrm{init}}\geq1.
	\end{equation}
	Assume that the initial data satisfy \eqref{eq1-21-aug-2026-B}.
	Assume furthermore that there exists $R_0\geq1$ such that
	$u_0$ and $u_1$ are spherically symmetric in $\{r\geq R_0\}$ and
	\begin{equation}\label{eq2-21-aug-2026-B}
		|\del u_0(x)|+|u_1(x)|
		\lesssim C_0,
		\qquad r\leq2R_0.
	\end{equation}
	Suppose that
	\begin{equation}\label{eq3-03-aug-2026}
		|f(t,x)|
		\lesssim
		\begin{cases}
			C_{\mathrm{int}}t^{-2-a}(t-r)^{-b},
			\quad &r<t-1/2;
			\\
			C_{\mathrm{ext}}r^{-d-1},
			\quad &r\geq t-1/2,
		\end{cases}
	\end{equation}
	where $C_{\mathrm{int}},C_{\mathrm{ext}},C_0>0$,
	$0<a\leq1$, $0<b<1$ and $d\geq2+a$.
	Then, for $(t,x)\in\{r<t-1/2\}$,
	\begin{equation}\label{eq4-03-aug-2026}
		|u(t,x)|
		\lesssim
		a^{-1}
		\bigl(
		C_0+C_{\mathrm{ext}}+(1-b)^{-1}C_{\mathrm{int}}
		\bigr)
		t^{-1}\langle t-r\rangle^{1-a-b}.
	\end{equation}
	The implicit constant may depend on $t_{\mathrm{init}}$, $R_0$, and $d$.
\end{corollary}
\begin{proof}
	Choose a radial function $\chi$ such that
	\[
		0\leq\chi\leq1,
		\qquad
		\chi=0\ \text{on }[0,R_0],
		\qquad
		\chi=1\ \text{on }[2R_0,+\infty).
	\]
	The initial data
	$(\chi u_0,\chi u_1)$ are spherically symmetric. Moreover,
	\eqref{eq1-21-aug-2026-B} and \eqref{eq2-21-aug-2026-B} imply
	\[
		|\chi u_0|\lesssim C_0\langle r\rangle^{-1},
		\qquad
		\left|
		\del_r(r\chi u_0)+r\chi u_1
		\right|
		\lesssim C_0\langle r\rangle^{-1-a}.
	\]
	Thus the spherically symmetric homogeneous solution generated by these
	data satisfies \eqref{eq4-03-aug-2026}, by
	Lemma~\ref{lem1-19-jan-2026} with vanishing right-hand side.

	The remaining initial data
	$((1-\chi)u_0,(1-\chi)u_1)$ are supported in
	$\{r\leq2R_0\}$. Kirchhoff's formula, the strong Huygens principle
	and \eqref{eq2-21-aug-2026-B} show that the corresponding homogeneous
	solution also satisfies \eqref{eq4-03-aug-2026} with constant
	$\lesssim C_0$.

	It remains to estimate the solution with zero initial data and source
	$f$. Let $v$ be the spherically symmetric solution to
	\[
		-\Box v=g,
		\qquad
		v(t_{\mathrm{init}},x)
		=\del_tv(t_{\mathrm{init}},x)=0,
	\]
	where
	\[
		g(t,r):=
		\begin{cases}
			C_{\mathrm{int}}t^{-2-a}(t-r)^{-b},
			& r<t-1/2,
			\\
			C_{\mathrm{ext}}r^{-d-1},
			& r\geq t-1/2.
		\end{cases}
	\]
We interpret $v$ as the retarded integral solution. Since $g$ is
non-negative and locally bounded, for $r>0$ the radial d'Alembert
formula gives
\[
v(t,r)=\frac1{2r}\int_{t_{\mathrm{init}}}^{t}
\int_{|r-(t-\tau)|}^{r+t-\tau}
\rho g(\tau,\rho)\diff\rho\diff\tau.
\]
Its value at $r=0$ is obtained by continuity. Positivity of the forward
fundamental solution bounds the source part of $u$ by $v$.
The function $w=rv$ has zero initial data and satisfies
\[
(\del_t-\del_r)(\del_t+\del_r)w=rg
\]
in the sense of distributions. The right-hand side is locally integrable
and obeys \eqref{eq3-08-jan-2026}, up to fixed multiplicative constants.
The characteristic integrations in the proof of
Lemma~\ref{lem1-19-jan-2026} therefore apply to $w$ in their integral
form: this follows by Fubini's theorem on characteristic rectangles,
or by local smooth approximation in the displayed d'Alembert formula.
Consequently the decay estimate \eqref{eq5-02-aug-2026} holds for $v$.
This is the only conclusion of that lemma needed here and completes
the proof without assuming that this piecewise-source solution is $C^2$.
\end{proof}

\begin{corollary}\label{cor1-03-aug-2026}
	Let $u$ be a $C^4$ solution to \eqref{eq1-03-aug-2026}
	satisfying the assumptions of Corollary~\ref{cor1-02-aug-2026}.
	Assume in addition that
	\begin{equation}\label{eq2-03-aug-2026}
		|\Delta_\omega f|
		\lesssim
		\begin{cases}
			C_{\mathrm{int}}t^{-2-a}(t-r)^{-b},
			\quad &r<t-1/2;
			\\
			0,
			\quad &r\geq t-1/2,
		\end{cases}
	\end{equation}
	where
	\[
		\Delta_\omega
		:=
		\sum_{1\leq b<c\leq3}\Omega_{bc}^2.
	\]
	Set
	\[
		z:=\Delta_\omega u,
		\qquad
		z_0:=\Delta_\omega u_0,
		\qquad
		z_1:=\Delta_\omega u_1.
	\]
	The two initial values of $z$ are
	\begin{equation}\label{eq3-21-aug-2026-B}
		z(t_{\mathrm{init}},x)=z_0(x),
		\qquad
		\del_tz(t_{\mathrm{init}},x)=z_1(x).
	\end{equation}
	Assume that
	\begin{equation}\label{eq4-21-aug-2026-B}
		\begin{aligned}
		&\mathrm{supp}\{z_0\}
		\cup
		\mathrm{supp}\{z_1\}
		\subseteq
		\{|x|\leq t_{\mathrm{init}}-1/2\},
		\\
		&\sum_{|I|\leq1}|\del^Iz_0(x)|
		+|z_1(x)|
		\lesssim C_0.
		\end{aligned}
	\end{equation}
	Then the radiation field
	\begin{equation}\label{eq5-21-aug-2026-B}
		\Rcal[u](\omega,q)
		=
		\lim_{r\rightarrow+\infty}ru(q+r,r\omega),
		\qquad
		q\geq1/2,
		\quad
		\omega\in\mathbb S^2,
	\end{equation}
	exists. Furthermore,
	\begin{equation}\label{eq6-21-aug-2026-B}
		\big|\Rcal[u](\omega,q)\big|
		\lesssim_{a,b,d}
		\bigl(C_0+C_{\mathrm{ext}}+C_{\mathrm{int}}\bigr)
		\langle q\rangle^{1-a-b},
		\qquad q\geq1/2.
	\end{equation}
	The implicit constant may also depend on $t_{\mathrm{init}}$ and $R_0$,
	as in Corollary~\ref{cor1-02-aug-2026}.
\end{corollary}
\begin{proof}
	Since the spherical Laplacian commutes with $\Box$, we have
	\begin{equation}\label{eq7-21-aug-2026-B}
		-\Box z=\Delta_\omega f.
	\end{equation}
	The assumptions \eqref{eq4-21-aug-2026-B} imply all the initial-data
	bounds in Corollary~\ref{cor1-02-aug-2026} for $z$. Applying that
	corollary to \eqref{eq7-21-aug-2026-B}, and using finite propagation
	speed, we obtain
	\begin{equation}\label{eq8-21-aug-2026-B}
		|z(t,x)|
		\lesssim_{a,b}
		\begin{cases}
			(C_0+C_{\mathrm{int}})
			t^{-1}(t-r)^{1-a-b},
			& r<t-1/2,
			\\
			0,
			& r\geq t-1/2.
		\end{cases}
	\end{equation}
	Consequently, in the region
	$\{r<t-1/2\}\cap\{r\geq t/2\}$,
	\begin{equation}\label{eq9-21-aug-2026-B}
		\begin{aligned}
			r^{-1}|\Delta_\omega u|
			&\lesssim_{a,b}
			(C_0+C_{\mathrm{int}})
			t^{-2}(t-r)^{1-a-b}
			\\
			&=
			(C_0+C_{\mathrm{int}})
			t^{-1-a}(t-r)^{-b}
			\left(\frac{t-r}{t}\right)^{1-a}
			\\
			&\lesssim_{a,b}
			(C_0+C_{\mathrm{int}})
			t^{-1-a}(t-r)^{-b}.
		\end{aligned}
	\end{equation}

	We now fix $q\geq1/2$. Along the ray $t=r+q$, the condition
	$r\geq q$ implies $r\geq t/2$. Moreover, every segment
	$\{t+r=\mathrm{constant}\}$ used in the first characteristic
	integration and ending at such a point remains in $\{r\geq t/2\}$.
	On the other hand, the wave equation gives
	\[
		(\del_t-\del_r)(\del_t+\del_r)(ru)
		-r^{-1}\Delta_\omega u
		=rf.
	\]
	Thus \eqref{eq3-03-aug-2026} and
	\eqref{eq9-21-aug-2026-B} imply, on all the characteristic segments
	relevant to the tail $r\geq q$,
	\[
		\big|
		(\del_t-\del_r)(\del_t+\del_r)(ru)
		\big|
		\lesssim_{a,b}
		\begin{cases}
			(C_0+C_{\mathrm{int}})
			t^{-1-a}(t-r)^{-b},
			& r<t-1/2,
			\\
			C_{\mathrm{ext}}r^{-d},
			& r\geq t-1/2.
		\end{cases}
	\]
	The initial data of $u$ satisfy \eqref{eq1-21-aug-2026-B}.
	Repeating the two characteristic integrations in the proof of
	Lemma~\ref{lem1-19-jan-2026}, now only for $r\geq q$, proves the
	existence of \eqref{eq5-21-aug-2026-B}.

	Finally, Corollary~\ref{cor1-02-aug-2026} applied directly to $u$
	gives \eqref{eq4-03-aug-2026}. Taking the limit along
	$t=r+q$ then gives \eqref{eq6-21-aug-2026-B}.
\end{proof}
\subsection{Proof of Proposition~\ref{prop1-19-07-2026}}
\begin{proof}
	Set
	\begin{equation}
		\begin{aligned}
			u_{1,0}(\hat{x})
			&:=u_1(\hat{t}_{\mathrm{init}},\hat{x})
			=\eta^{\alpha\beta}h_{\alpha\beta}
			(\hat{t}_{\mathrm{init}},\hat{x})
			+12\pi\phi^2(\hat{t}_{\mathrm{init}},\hat{x}),
			\\
			u_{1,1}(\hat{x})
			&:=\hat{\del}_0u_1(\hat{t}_{\mathrm{init}},\hat{x})
			\\
			&=(1-\mSch)^{-1}
			\bigl(
			\eta^{\alpha\beta}\del_th_{\alpha\beta}
			+24\pi\phi\del_t\phi
			\bigr)(\hat{t}_{\mathrm{init}},\hat{x}).
		\end{aligned}
	\end{equation}
	The initial-data assumptions and the exterior Schwarzschild form imply
	that $u_{1,0}$ and $u_{1,1}$ are spherically symmetric for
	$\hat r\geq1$, and satisfy the local bounds required in
	\eqref{eq2-21-aug-2026-B}. More precisely, in the Schwarzschild
	region,
	\begin{equation}
		u_{1,0}(\hat r)
		=
		\frac{4\mSch^2}{\hat r^2-\mSch^2}
		+\frac{4\mSch}{\hat r}
		+\frac{2\mSch^2}{\hat r^2},
		\qquad
		u_{1,1}(\hat r)=0.
	\end{equation}
	Combining this identity with the local initial estimates, we obtain
	\begin{equation}\label{eq1-21-aug-2026-u1}
		|u_{1,0}(\hat r\omega)|
		\lesssim C_1\eps\langle\hat r\rangle^{-1},
		\qquad
		\left|
		\hat\del_r(\hat r u_{1,0}(\hat r\omega))
		+\hat r u_{1,1}(\hat r\omega)
		\right|
		\lesssim C_1\eps\langle\hat r\rangle^{-2}.
	\end{equation}
	In particular, the initial-data assumptions in
	Corollary~\ref{cor1-02-aug-2026} hold for every $0<a\leq1$.

	We now apply \eqref{eq1-01-august-2026}. In $\Kcal^m$, since
	$\hat q=\hat t-\hat r>1/2$,
	\begin{equation}
		|\mathbb G|
		\lesssim
		(C_1\eps)^2
		\hat t^{-3+7\delta}\hat q^{-1+\delta}.
	\end{equation}
	On $\{\hat q>1/2\}\cap(\Kcal^m)^c$, the Schwarzschild exterior estimate,
	together with $1/2<\hat q\leq1$ and $\hat r\simeq\hat t$, gives the same
	bound; for $\hat q\leq1/2$ one has
	$|\mathbb G|\lesssim(C_1\eps)^2\hat r^{-4}$. Thus the source assumptions
	of Corollary~\ref{cor1-02-aug-2026} hold with
	\[
		a=1-7\delta,
		\qquad
		b=1-\delta,
		\qquad
		d=3.
	\]
	That corollary therefore gives
	\begin{equation}
		|u_1|
		\lesssim_{\delta}
		C_1\eps\hat t^{-1}
		\langle\hat t-\hat r\rangle^{-1+8\delta}.
	\end{equation}

	It remains to estimate the radiation field. In $\Kcal^m$,
	\eqref{eq1-01-august-2026} gives
	\begin{equation}\label{eq2-21-aug-2026-u1}
		|\mathbb G|+|\Delta_\omega\mathbb G|
		\lesssim |\mathbb G|_2
		\lesssim_{\delta}
		(C_1\eps)^2
		\hat t^{-3+8\delta}\hat q^{-1+\delta}.
	\end{equation}
	Outside $\Kcal^m$, $\mathbb G$ is spherically symmetric, so
	$\Delta_\omega\mathbb G=0$, and the transition region is handled as above.
	Hence the source assumptions of Corollary~\ref{cor1-03-aug-2026} hold with
	\[
		a=1-8\delta,
		\qquad
		b=1-\delta,
		\qquad
		d=3.
	\]
	Set
	\[
		z:=\Delta_\omega u_1.
	\]
	Its two initial values are
	\begin{equation}\label{eq3-21-aug-2026-u1}
		\begin{aligned}
			z(\hat t_{\mathrm{init}},\hat x)
			&=\Delta_\omega u_{1,0}(\hat x)=:z_0(\hat x),
			\\
			\hat\del_0z(\hat t_{\mathrm{init}},\hat x)
			&=\Delta_\omega u_{1,1}(\hat x)=:z_1(\hat x).
		\end{aligned}
	\end{equation}
	The exterior Schwarzschild data show that the two quantities in
	\eqref{eq3-21-aug-2026-u1} are supported in
	$\{\hat r\leq\hat t_{\mathrm{init}}-1/2\}$. The initial high-order
	estimates also give
	\begin{equation}
		\sum_{|I|\leq1}|\hat\del^I z_0(\hat x)|
		+|z_1(\hat x)|
		\lesssim C_1\eps.
	\end{equation}
	Thus all the initial-data assumptions in
	Corollary~\ref{cor1-03-aug-2026} are satisfied. Applying that
	corollary, we obtain the existence of $\hat{\Rcal}[u_1]$ and
	\begin{equation}
		\big|\hat{\Rcal}[u_1](\omega,\hat q)\big|
		\lesssim_{\delta}
		C_1\eps\langle\hat q\rangle^{-1+9\delta}.
	\end{equation}
	This completes the proof.
\end{proof}

\section{Klein--Gordon asymptotics and profile estimates}
\label{sec:appendix-kg-profile-technical}
We prove Theorem~\ref{thm:refined-kg-asymptotics}, the coefficient
derivative bounds used in Proposition~\ref{prop:kg-energy-profile}, and
the scalar-source remainder estimate \eqref{eq:modified-source-remainder-bound}.
We begin with two ODE lemmas along hyperboloidal rays.

\subsection{Two ray-wise ODE lemmas}

\begin{lemma}[Hyperboloidal ray reduction]
\label{prop1-11-05-2021}
Let $v$ be a sufficiently regular and spatially compactly supported solution to
\begin{equation}\label{starting KG-sharp}
-\Boxt_g v+c^2v=f
\end{equation}
in $\Hcal_{[s_0,s_1]}$, where $c>0$. Set
\begin{equation}\label{eq4-11-05-2021}
g^{\alpha\beta}=\eta^{\alpha\beta}+h^{\alpha\beta},
\qquad
\bar h^{00}:=\left(\frac{t}{s}\right)^2\hu^{00}.
\end{equation}
For fixed $(t,x)\in\Hcal_{[s_0,s_1]}$, let
\[
s:=\sqrt{t^2-r^2},
\qquad
\gamma_{t,x}(\lambda):=
\left(\frac{\lambda t}{s},\frac{\lambda x}{s}\right).
\]
Assume that $|\bar h^{00}|\leq1/3$ on the ray segment under consideration,
and set
\[
w_{t,x}(\lambda)
:=\lambda^{3/2}\left.v\right|_{\gamma_{t,x}(\lambda)}.
\]
Then
\begin{equation}\label{eq1-10-05-2021}
\begin{aligned}
\frac{d^2}{d\lambda^2}w_{t,x}(\lambda)
&+\frac{c^2}
{1-\left.\bar h^{00}\right|_{\gamma_{t,x}(\lambda)}}
 w_{t,x}(\lambda)
\\
&=
\left(1-\left.\bar h^{00}\right|_{\gamma_{t,x}(\lambda)}\right)^{-1}
\left.
\big(R_1[v]+R_2[v]+R_3[v]+s^{3/2}f\big)
\right|_{\gamma_{t,x}(\lambda)}.
\end{aligned}
\end{equation}
Here, the repeated spatial indices are summed from $1$ to $3$, and
\begin{subequations}\label{eq:CMP-remainders}
\begin{align}
R_1[v]:={}&
 s^{3/2}\sum_a\delu_a\delu_a v
 +\frac{x^ax^b}{s^{1/2}}\delu_a\delu_bv
 +\frac{3}{4s^{1/2}}v
 +\frac{3x^a}{s^{1/2}}\delu_av,
\label{eq:CMP-R1}
\\
R_2[v]:={}&
-\bar h^{00}\left(
\frac{3v}{4s^{1/2}}+\frac{3s^{3/2}}{t}\delu_0v
\right)
+ s^{3/2}\Bigl(
\hu^{0b}(\delu_0\delu_b+\delu_b\delu_0)v
+\hu^{ab}\delu_a\delu_bv
\notag\\
&\hspace{35mm}
+h^{\alpha\beta}\del_\alpha\Psiu_\beta^{\ \beta'}
\delu_{\beta'}v
\Bigr)
-\frac{r^2}{s^{1/2}t}\hu^{00}\delu_0v,
\label{eq:CMP-R2}
\\
R_3[v]:={}&
-\bar h^{00}\left(
\frac{2x^as^{3/2}}{t}\delu_0\delu_av
+\frac{3x^a}{s^{1/2}}\delu_av
+\frac{x^ax^b}{s^{1/2}}\delu_a\delu_bv
\right).
\label{eq:CMP-R3}
\end{align}
\end{subequations}
\end{lemma}

\begin{proof}
The decomposition in \cite[Lemma~3.4]{PLF-YM-CMP} is written in the
hyperboloidal frame. We use
$\bar\partial_0=(s/t)\delu_0$, $\bar\partial_a=\delu_a$,
$\bar h^{00}=(t/s)^2\hu^{00}$,
$\bar h^{0b}=(t/s)\hu^{0b}$, $\bar h^{ab}=\hu^{ab}$, and
\[
\bar h^{00}\bar\partial_0^2v
=\hu^{00}\del_t^2v
+\frac{r^2}{s^2t}\hu^{00}\del_tv.
\]
Moreover, \cite{PLF-YM-CMP} uses
$g^{\alpha\beta}=m^{\alpha\beta}-h^{\alpha\beta}$, whereas here
$g^{\alpha\beta}=\eta^{\alpha\beta}+h^{\alpha\beta}$. Substitution into
\cite[proof of Lemma~3.4]{PLF-YM-CMP} gives
\eqref{eq1-10-05-2021} and \eqref{eq:CMP-remainders}.
\end{proof}

\begin{lemma}[First-order diagonalization]
\label{lem1-2026-7-31}
Let $v$ be a sufficiently regular function defined on $[s_0,\infty)$ and
satisfying
\begin{equation}\label{eq1-2026-7-31}
v''(s)+c^2(1+\mathfrak q(s))v(s)=f(s),
\end{equation}
where $c>0$, $\mathfrak q$ is real-valued,
$|\mathfrak q|\leq1/2$, and
\[
|\mathfrak q'v'|+|f|\in L^1([s_0,\infty)).
\]
Set $a(s):=\sqrt{1+\mathfrak q(s)}$ and
\begin{equation}\label{eq:general-vpm}
\begin{aligned}
v^{\pm}(s)
&:=e^{\pm ci\int_{s_0}^{s}a(\sigma)\diff\sigma}
\left(\mp\frac{v'(s)}{2ci\,a(s)}+\frac12v(s)\right),
\\
F^{\pm}(s)
&:=\pm e^{\pm ci\int_{s_0}^{s}a(\sigma)\diff\sigma}
\left(
\frac{\mathfrak q'(s)v'(s)}{4ci\,a(s)^3}
-\frac{f(s)}{2ci\,a(s)}
\right).
\end{aligned}
\end{equation}
Then
\begin{equation}\label{eq1-2026-8-22}
v^{\pm}(s)=b^{\pm}+\mathfrak r^{\pm}(s),
\qquad
b^{\pm}:=v^{\pm}(s_0)+\int_{s_0}^{\infty}F^{\pm}(\sigma)\diff\sigma,
\qquad
\mathfrak r^{\pm}(s):=-\int_s^{\infty}F^{\pm}(\sigma)\diff\sigma.
\end{equation}
Writing
\[
\Phi(s):=c\int_{s_0}^{s}a(\sigma)\diff\sigma,
\]
one has
\begin{equation}\label{eq:ode-v-vprime-reconstruction}
\begin{aligned}
v(s)
&=e^{-i\Phi(s)}\bigl(b^++\mathfrak r^+(s)\bigr)
 +e^{i\Phi(s)}\bigl(b^-+\mathfrak r^-(s)\bigr),
\\
v'(s)
&=-ci\,a(s)e^{-i\Phi(s)}\bigl(b^++\mathfrak r^+(s)\bigr)
 +ci\,a(s)e^{i\Phi(s)}\bigl(b^-+\mathfrak r^-(s)\bigr).
\end{aligned}
\end{equation}
If $v$ and $f$ are real-valued, then
\[
b^- = \overline{b^+},
\qquad
\mathfrak r^-(s)=\overline{\mathfrak r^+(s)}.
\]
Moreover,
\begin{equation}\label{eq2-2026-7-31}
\begin{aligned}
|b^{\pm}|
&\leq |v^{\pm}(s_0)|
 +Cc^{-1}\int_{s_0}^{\infty}
 \bigl(|\mathfrak q'(\sigma)v'(\sigma)|+|f(\sigma)|\bigr)\diff\sigma,
\\
|\mathfrak r^{\pm}(s)|
&\leq Cc^{-1}\int_s^{\infty}
 \bigl(|\mathfrak q'(\sigma)v'(\sigma)|+|f(\sigma)|\bigr)\diff\sigma.
\end{aligned}
\end{equation}
\end{lemma}

\begin{proof}
As in \cite[proof of Lemma~3.5]{PLF-YM-CMP},
\[
a'=\frac{\mathfrak q'}{2a},
\]
and direct differentiation of \eqref{eq:general-vpm}, together with
\eqref{eq1-2026-7-31}, gives
\[
\bigl(v^{\pm}\bigr)'=F^{\pm}.
\]
Since $|\mathfrak q|\leq1/2$, one has $a\simeq1$ and hence
$F^\pm\in L^1([s_0,\infty))$. Integration gives
\eqref{eq1-2026-8-22} and \eqref{eq2-2026-7-31}; adding and subtracting
\eqref{eq:general-vpm} gives \eqref{eq:ode-v-vprime-reconstruction}.
\end{proof}

\subsection{Verification of the ODE assumptions}

Set
\[
s_0:=\hat t_{\mathrm{init}},
\qquad
H:=\hu^{00}+4\mSch t^{-1}.
\]

\begin{lemma}[Pointwise estimates]
\label{lem:pointwise-inputs-for-ode}
Let $(h_{\alpha\beta},\phi)$ be the global solution to
\eqref{eq3-29-nov-2025} together with initial data
\eqref{eq main initial-data} satisfying \eqref{eq1-07-dec-2025} and
\eqref{eq main smallness-condition}. Then, in $\Kcal^m$,
\begin{subequations}\label{eq:pointwise-inputs-for-ode}
\begin{align}
|h|_{N-7}
&\lesssim C_1\eps\hat t^{-1+\delta},
\label{eq:pointwise-h}
\\[1mm]
|H|_{N-9}
&\lesssim_{\delta,N}
C_1\eps\hat t^{-2+7\delta/2}
\big(\hat t-\hat r+\hat s^\delta\big),
\label{eq:pointwise-H}
\\[1mm]
|\del H|_{N-9}
&\lesssim_{\delta,N}
C_1\eps\hat t^{-2+7\delta/2},
\label{eq:pointwise-del-H}
\\[1mm]
|\phi|_{N-6}
+\frac{\hat s}{\hat t}|\del\phi|_{N-6}
+|\delu\phi|_{N-6}
&\lesssim
C_1\eps\hat t^{-3/2+\delta},
\label{eq:pointwise-phi-high}
\\[1mm]
|\phi|+|\del_t\phi|
&\lesssim_{\delta}
C_1\eps\hat t^{-5/2+\delta}
\big(\hat t-\hat r+\hat s^\delta\big)^{1+\delta/2}.
\label{eq:pointwise-phi-sharp}
\end{align}
\end{subequations}
\end{lemma}

\begin{proof}
By \eqref{eq1-14-dec-2025}, \eqref{eq-2026-8-4},
\eqref{eq:del-hu00-N9}, \eqref{eq5-20-june-2026}, and
$\mSch\lesssim C_1\eps$, together with
\[
t\simeq\hat t,
\qquad
\hat t-\hat r\leq t-r
\lesssim_\delta \hat t-\hat r+\hat s^\delta,
\qquad
s^\delta\lesssim_\delta\hat t^\delta,
\qquad
\frac{\hat s}{\hat t}\leq\frac{s}{t},
\]
we obtain the first four estimates.

For the last estimate, we use
\begin{equation}\label{eq:wave-decomposition-pointwise-kg}
\begin{aligned}
-\Box\phi={}&(s/t)^2\del_t\del_t\phi
+t^{-1}\sum_a\left(2(x^a/t)\del_tL_a\phi-\delu_aL_a\phi\right)
\\
&-t^{-1}\sum_a(x^a/t)\delu_a\phi
+t^{-1}\left((r/t)^2+3\right)\del_t\phi.
\end{aligned}
\end{equation}
Since $\Boxt_g\phi-c^2\phi=0$,
\begin{equation}\label{eq1-21-jan-2026}
c^2|\phi|
\lesssim
(s/t)^2|\del\del\phi|
+t^{-1}|\del\phi|_1
+|h|\,|\del\del\phi|.
\end{equation}
Differentiating the Klein--Gordon equation with respect to $t$ and applying
\eqref{eq:wave-decomposition-pointwise-kg} to $\del_t\phi$ gives
\begin{equation}\label{eq:dt-kg-improved-pointwise}
c^2|\del_t\phi|
\lesssim
(s/t)^2|\del\del\del\phi|
+t^{-1}|\del\del\phi|_1
+|h|\,|\del\del\del\phi|
+|\del h|\,|\del\del\phi|.
\end{equation}
By \eqref{eq5-20-june-2026}, \eqref{eq1-14-dec-2025}, and
\eqref{eq3-07-dec-2025}, and using
$s^2\lesssim t(t-r)$ and $t-r\geq1$, the right-hand sides are bounded by
\[
C_1\eps t^{-5/2+\delta}(t-r)^{1+\delta/2}.
\]
The last estimate follows from the above coordinate comparisons.
\end{proof}

\begin{proposition}[ODE estimates along hyperboloidal rays]
\label{prop:ray-wise-kg-bounds}
Let $(h_{\alpha\beta},\phi)$ be the global solution to
\eqref{eq3-29-nov-2025} together with initial data
\eqref{eq main initial-data} satisfying \eqref{eq1-07-dec-2025} and
\eqref{eq main smallness-condition}. For fixed $\hat y\in B_1(0)$, set
\[
\vartheta:=\sqrt{1-|\hat y|^2},
\qquad
\gamma_{\hat y}(\lambda):=
\left(\frac{\lambda}{\vartheta},\frac{\lambda\hat y}{\vartheta}\right),
\qquad
B:=1-2\mSch t^{-1}.
\]
Define
\begin{equation}\label{eq:def-vtx}
\begin{aligned}
v(\lambda,\hat y)
&:=\lambda^{3/2}\left.\phi\right|_{\gamma_{\hat y}(\lambda)},
\\
\mathfrak q(\lambda,\hat y)
&:=\vartheta^{-2}\left.HB^2\right|_{\gamma_{\hat y}(\lambda)}.
\end{aligned}
\end{equation}
Let
\begin{equation}\label{eq:effective-initial-lambda}
\lambda_0(\hat y):=
\max\left\{
 s_0,
 \frac12\left(\frac{1+|\hat y|}{1-|\hat y|}\right)^{1/2}
\right\}.
\end{equation}
For $C_1\eps$ sufficiently small, on $[\lambda_0(\hat y),\infty)$,
\begin{equation}\label{eq2:12-4-2026}
v''+c^2(1+\mathfrak q)v=f_*,
\qquad
|\mathfrak q|\leq\frac13,
\end{equation}
and
\begin{equation}\label{eq:ray-ode-input-bounds}
\begin{aligned}
|v|+|v'|
&\lesssim
C_1\eps\,\vartheta^{5/2-2\delta}\lambda^{2\delta},
\\
|\mathfrak q'v'|+|f_*|
&\lesssim
C_1\eps\,\vartheta^{3/2+7\delta/2}\lambda^{-2+9\delta}.
\end{aligned}
\end{equation}
\end{proposition}

\begin{proof}
From \eqref{eq3-29-nov-2025} and
\eqref{eq1-2026-8-6}--\eqref{eq2-2026-8-6},
\[
\Boxt_g\phi-c^2\phi=0,
\qquad
\hat\Box\phi
=-(1+4\mSch t^{-1})\del_t^2\phi
+\Delta_x\phi+\mSch t^{-2}\mathbb A[\phi],
\]
where
\[
\mathbb A[\phi]
=4\mSch\frac{4\mSch t^{-1}-3}{B^2}\del_t^2\phi
+\frac{2}{B^3}\del_t\phi.
\]
Writing $h^{\mu\nu}\del_\mu\del_\nu\phi$ in the semi-hyperboloidal frame
and using \eqref{eq3-04-june-2026}, we obtain
\begin{equation}\label{eq:kg-hat-coordinate-detailed}
-\hat\Box\phi-HB^2\hat\del_0^2\phi+c^2\phi
=-\mathbb R_1[h,\phi],
\end{equation}
where
\begin{equation}\label{eq:def-mathbb-R1-detailed}
\begin{aligned}
\mathbb R_1[h,\phi]:={}&
-\sum_{(\mu,\nu)\neq(0,0)}
\hu^{\mu\nu}\delu_\mu\delu_\nu\phi
-h^{\mu\nu}\del_\mu\Psiu_\nu^{\ \nu'}\delu_{\nu'}\phi
\\
&+\mSch t^{-2}\mathbb A[\phi]
-2\mSch t^{-2}H\hat\del_0\phi .
\end{aligned}
\end{equation}
In $\Kcal^m$, by \eqref{eq:pointwise-h}, \eqref{eq:pointwise-H},
\eqref{eq:pointwise-phi-high}, \eqref{eq11-20-june-2026},
\eqref{eq7-03-aug-2026}, and $\mSch\lesssim C_1\eps$,
\begin{equation}\label{eq:R1-hat-pointwise-bound}
|\mathbb R_1[h,\phi]|
\lesssim
(C_1\eps)^2\hat t^{-7/2+2\delta}.
\end{equation}

Along $\gamma_{\hat y}$, $\hat x/\hat t=\hat y$, $\hat s=\lambda$, and
$\vartheta$ is constant. For any spacetime quantity $F$,
\begin{equation}\label{eq:ray-derivative-identity}
\del_\lambda\left.F\right|_{\gamma_{\hat y}(\lambda)}
=
\left.
\left(
\vartheta\hat\del_0F
+\frac{\hat y^a}{\lambda}\hat L_aF
\right)
\right|_{\gamma_{\hat y}(\lambda)}.
\end{equation}

We first treat the Schwarzschild exterior. At an exterior point of
$\gamma_{\hat y}$ with $\lambda\geq s_0$, one has $r\geq1$ and
$0<\hat q\leq1$: indeed $t\geq r$ along the ray, while
$t-r-2m_{\mathrm S}\ln r\leq1$ in the exterior. In particular,
$r\simeq t\simeq\lambda/\vartheta$ and, with $d:=t-r$,
$d=\hat q+2m_{\mathrm S}\ln t\lesssim1+\ln t$. Rearranging
\eqref{eq3-20-june-2026} gives the exact cancellation
\[
H=-\frac{2m_{\mathrm S}}{t^2}
\left(\frac{(d+m_{\mathrm S})^2}{r-m_{\mathrm S}}
      +\frac{m_{\mathrm S}^2}{r+m_{\mathrm S}}\right).
\]
Along the ray, $t'=B^{-1}/\vartheta$ and $r'=r/\lambda$. Thus
$|t'|/t+|r'|/r\lesssim\lambda^{-1}$, and
\[
d'=\frac{\hat q}{\lambda}+\frac{B^{-1}-1}{\vartheta},
\qquad |d'|\lesssim\lambda^{-1}.
\]
The formula for $H$ therefore gives
$|H|+\lambda|H'|\lesssim m_{\mathrm S}t^{-3}(1+\ln t)^2$.
Since $B\simeq1$, $\lambda|B'|\lesssim m_{\mathrm S}/t$, and
$\vartheta^{-2}t^{-2}\simeq\lambda^{-2}$, we obtain
\begin{equation}\label{eq:exterior-phase-coefficient}
|\mathfrak q|+\lambda|\mathfrak q'|
\lesssim m_{\mathrm S}\lambda^{-2}t^{-1}(1+\ln t)^2
\lesssim m_{\mathrm S}\lambda^{-2}.
\end{equation}
This estimate is uniform in $\hat y$, including as $|\hat y|\uparrow1$.
The scalar field and its derivatives vanish on the exterior part of
the ray.

On the part of the ray in $\Kcal^m$, we instead use the commuted
interior estimates.
Since
\[
\hat L_a
=L_a
+\frac{2\mSch x^a}{t-2\mSch}\del_t
-2\mSch\ln t\,\del_a,
\]
we have
\[
|\hat L H|_{N-10}
\lesssim_N
|H|_{N-9}
+\mSch(1+\ln t)|\del H|_{N-9}.
\]
By \eqref{eq:pointwise-H}, \eqref{eq:pointwise-del-H},
$\mSch\lesssim C_1\eps$, and
\[
\ln t\lesssim_\delta t^{\delta/2}
\lesssim_\delta \hat s^\delta,
\]
we obtain
\begin{equation}\label{eq:effective-metric-estimates}
\begin{aligned}
|\hat L H|_{N-10}
&\lesssim_{\delta,N}
C_1\eps\hat t^{-2+7\delta/2}
\bigl(\hat t-\hat r+\hat s^\delta\bigr)
\\
&\lesssim_{\delta,N}
C_1\eps\left(
 t^{-2+7\delta/2}(t-r)
+t^{-2+9\delta/2}
\right).
\end{aligned}
\end{equation}
On the part of the ray with $\hat t-\hat r\geq1/2$,
\begin{equation}\label{eq:ray-coordinate-comparison}
\begin{aligned}
\hat s&=\lambda,
&
\hat t&=\lambda\vartheta^{-1},
&
\hat t-\hat r&=\frac{\lambda\vartheta}{1+|\hat y|},
\\
t&\simeq\hat t,
&
\lambda\vartheta&\geq\frac12 .
\end{aligned}
\end{equation}
Moreover,
\[
\left|
\del_\lambda\left.B\right|_{\gamma_{\hat y}(\lambda)}
\right|
\lesssim \mSch\vartheta\lambda^{-2}.
\]
Thus \eqref{eq:def-vtx} and \eqref{eq:ray-derivative-identity}, together with
$B\simeq1$, give
\begin{equation}\label{eq:ode-structural-bounds}
\begin{aligned}
|\mathfrak q|
&\lesssim
\vartheta^{-2}
\left.|H|\right|_{\gamma_{\hat y}(\lambda)},
\\
|\mathfrak q'|
&\lesssim
\vartheta^{-2}
\left.
\left(
\vartheta|\hat\del_0H|
+\lambda^{-1}|\hat L H|
+\mSch\vartheta\lambda^{-2}|H|
\right)
\right|_{\gamma_{\hat y}(\lambda)},
\\
|v|+|v'|
&\lesssim
\lambda^{3/2}
\left.
\left(
|\phi|+\vartheta|\hat\del_0\phi|
+\lambda^{-1}|\hat L\phi|
\right)
\right|_{\gamma_{\hat y}(\lambda)}.
\end{aligned}
\end{equation}
By \eqref{eq1-2026-8-5},
\[
\left.|\hat L\phi|\right|_{\gamma_{\hat y}(\lambda)}
\lesssim
\left.t^\delta|\phi|_1\right|_{\gamma_{\hat y}(\lambda)}.
\]
Since
$|\hat\del_0H|\lesssim|\del H|$ and
$|\hat\del_0\phi|\lesssim|\del_t\phi|$,
\eqref{eq:pointwise-H}, \eqref{eq:pointwise-del-H},
\eqref{eq:pointwise-phi-high}, \eqref{eq:pointwise-phi-sharp},
\eqref{eq:effective-metric-estimates}, and
\eqref{eq:ray-coordinate-comparison} give
\begin{equation}\label{eq:ode-quantitative-bounds}
\begin{aligned}
|\mathfrak q|
&\lesssim
C_1\eps\left(
\vartheta^{1-7\delta/2}\lambda^{-1+7\delta/2}
+\vartheta^{-7\delta/2}\lambda^{-2+9\delta/2}
\right),
\\
|\mathfrak q'|
&\lesssim
C_1\eps\left(
\vartheta^{1-7\delta/2}\lambda^{-2+7\delta/2}
+\vartheta^{-7\delta/2}\lambda^{-3+9\delta/2}
\right),
\\
|v|+|v'|
&\lesssim
C_1\eps\,\vartheta^{5/2-2\delta}\lambda^{2\delta}.
\end{aligned}
\end{equation}
The same estimates hold on the exterior part of the ray for
$\lambda\geq\lambda_0(\hat y)$: $v$ and $v'$ vanish there, and the bounds
for $\mathfrak q$ and $\mathfrak q'$ follow from
\eqref{eq:exterior-phase-coefficient}, since $\lambda\geq s_0\geq1$.
Since $\lambda_0(\hat y)\simeq\vartheta^{-1}$, the first estimate gives
$|\mathfrak q|\leq1/3$ on $[\lambda_0(\hat y),\infty)$ when
$C_1\eps$ is sufficiently small.

Define
\begin{equation}\label{eq:def-effective-hat-h00}
\hat h^{00}:=HB^2,
\qquad
\hat h^{0a}=\hat h^{ab}:=0,
\qquad
\hat g^{\alpha\beta}:=\eta^{\alpha\beta}+\hat h^{\alpha\beta}.
\end{equation}
Then \eqref{eq:kg-hat-coordinate-detailed} is
\[
-\Boxt_{\hat g}\phi+c^2\phi=-\mathbb R_1[h,\phi],
\qquad
\Boxt_{\hat g}:=\hat g^{\alpha\beta}\hat\del_\alpha\hat\del_\beta.
\]
Since $\hat h^{00}$ is the only non-zero Cartesian component,
$\underline{\hat h}^{00}=\hat h^{00}$ and
$\underline{\hat h}^{0a}=\underline{\hat h}^{ab}=0$. Hence the coefficient
$\bar h^{00}$ in Lemma~\ref{prop1-11-05-2021} is $\mathfrak q$ along
$\gamma_{\hat y}$.

Let $\widehat R_j[\phi]$ denote the remainders in
Lemma~\ref{prop1-11-05-2021} written in the hatted coordinates. For example,
with
\[
\widehat{\delu}_0:=\hat\del_0,
\qquad
\widehat{\delu}_a:=\hat t^{-1}\hat L_a,
\]
we have
\begin{equation}\label{eq:def-hat-R1}
\begin{aligned}
\widehat R_1[\phi]={}&
\hat s^{3/2}\sum_a\widehat{\delu}_a\widehat{\delu}_a\phi
+\frac{\hat x^a\hat x^b}{\hat s^{1/2}}
\widehat{\delu}_a\widehat{\delu}_b\phi
+\frac{3}{4\hat s^{1/2}}\phi
+\frac{3\hat x^a}{\hat s^{1/2}}\widehat{\delu}_a\phi.
\end{aligned}
\end{equation}
Lemma~\ref{prop1-11-05-2021} applied to
\eqref{eq:kg-hat-coordinate-detailed} gives
\[
v''+\frac{c^2}{1-\mathfrak q}v
=\frac{1}{1-\mathfrak q}
\left(
\sum_{j=1}^3
\left.\widehat R_j[\phi]\right|_{\gamma_{\hat y}(\lambda)}
-\lambda^{3/2}
\left.\mathbb R_1[h,\phi]\right|_{\gamma_{\hat y}(\lambda)}
\right).
\]
Using
\[
(1-\mathfrak q)^{-1}
=1+\mathfrak q+\mathfrak q^2(1-\mathfrak q)^{-1},
\]
we obtain \eqref{eq2:12-4-2026}, where
\begin{equation}\label{eq:def-f-star}
f_*:=
\frac{1}{1-\mathfrak q}
\left(
\sum_{j=1}^3
\left.\widehat R_j[\phi]\right|_{\gamma_{\hat y}(\lambda)}
-\lambda^{3/2}
\left.\mathbb R_1[h,\phi]\right|_{\gamma_{\hat y}(\lambda)}
-c^2\mathfrak q^2v
\right).
\end{equation}
In \eqref{eq:CMP-remainders}, every term containing $\hat\del_0$ in
$\widehat R_2$ and $\widehat R_3$ carries a factor $\mathfrak q$. Thus,
by \eqref{eq1-2026-8-5} and $\lambda\vartheta\geq1/2$,
\[
\begin{aligned}
\sum_{j=1}^3
\left.|\widehat R_j[\phi]|\right|_{\gamma_{\hat y}(\lambda)}
\lesssim{}&
\lambda^{-1/2}
\left.
\bigl(|\phi|+t^\delta|\phi|_2\bigr)
\right|_{\gamma_{\hat y}(\lambda)}
\\
&+|\mathfrak q|\lambda^{1/2}\vartheta
\left.
t^\delta\bigl(|\phi|_2+|\del\phi|_1\bigr)
\right|_{\gamma_{\hat y}(\lambda)}.
\end{aligned}
\]
Together with \eqref{eq:R1-hat-pointwise-bound}, this gives the following
bound in $\Kcal^m$. It also holds in the exterior, where all the scalar
terms on the left-hand side vanish:
\begin{equation}\label{eq:ray-source-terms-bound}
\sum_{j=1}^3
\left.|\widehat R_j[\phi]|\right|_{\gamma_{\hat y}(\lambda)}
+\lambda^{3/2}
\left.|\mathbb R_1[h,\phi]|\right|_{\gamma_{\hat y}(\lambda)}
+|\mathfrak q^2v|
\lesssim
C_1\eps\,\vartheta^{3/2+7\delta/2}\lambda^{-2+9\delta}.
\end{equation}
By \eqref{eq:def-f-star}, \eqref{eq:ode-quantitative-bounds}, and
\eqref{eq:ray-source-terms-bound},
\begin{equation}\label{eq:ode-forcing-bound}
|\mathfrak q'v'|+|f_*|
\lesssim
C_1\eps\,\vartheta^{3/2+7\delta/2}\lambda^{-2+9\delta}.
\end{equation}
This proves \eqref{eq:ray-ode-input-bounds}.
\end{proof}

\begin{proof}[Proof of Theorem~\ref{thm:refined-kg-asymptotics}]
Fix $\hat y\in B_1(0)$. By Proposition~\ref{prop:ray-wise-kg-bounds},
$|\mathfrak q|\leq1/3$ on $[\lambda_0(\hat y),\infty)$. If the interval
$[s_0,\lambda_0(\hat y))$ is nonempty, then $\hat q<1/2$ there by the
definition of $\lambda_0(\hat y)$. By Corollary~\ref{coro1-22-june-2026}
this part of the ray lies in the Schwarzschild exterior.
The uniform bound \eqref{eq:exterior-phase-coefficient} gives
$|\mathfrak q|\leq1/3$ there as well for sufficiently small $C_1\eps$. Hence
\begin{equation}\label{eq:def-a-application}
a(\lambda,\hat y):=\sqrt{1+\mathfrak q(\lambda,\hat y)}
\end{equation}
is defined for $\lambda\geq s_0$. Below $s_0$ we use the real continuous
extension stated in Theorem~\ref{thm:refined-kg-asymptotics}; no ODE or
metric square-root identity is asserted there. Since $0<\delta<1/9$,
\eqref{eq:ode-forcing-bound} is integrable in $\lambda$. Since
$\lambda_0\simeq\vartheta^{-1}$, \eqref{eq:ode-quantitative-bounds} and
\eqref{eq:ode-forcing-bound} give
\[
|v^\pm(\lambda_0,\hat y)|
\lesssim C_1\eps\vartheta^{5/2-4\delta},
\qquad
\int_{\lambda_0}^{\infty}
\bigl(|\mathfrak q'v'|+|f_*|\bigr)\diff\lambda
\lesssim C_1\eps\vartheta^{5/2-11\delta/2}.
\]
Thus Lemma~\ref{lem1-2026-7-31} applies on
$[\lambda_0(\hat y),\infty)$. Absorbing the phase accumulated on
$[s_0,\lambda_0(\hat y)]$ into $b^\pm$ and $\mathfrak r^\pm$, we set
\[
\Phi(\lambda,\hat y):=c\int_{s_0}^{\lambda}a(\sigma,\hat y)\diff\sigma.
\]
It gives
\begin{equation}\label{eq:applied-v-vprime-reconstruction}
\begin{aligned}
v&=e^{-i\Phi}(b^++\mathfrak r^+)+e^{i\Phi}(b^-+\mathfrak r^-),
\\
v'&=-ci\,a e^{-i\Phi}(b^++\mathfrak r^+)
+ci\,a e^{i\Phi}(b^-+\mathfrak r^-).
\end{aligned}
\end{equation}
Thus
\begin{equation}\label{eq:applied-profile-remainder-bounds}
\begin{aligned}
|b^\pm(\hat y)|
&\lesssim C_1\eps\vartheta^{5/2-11\delta/2},
\\
|\mathfrak r^\pm(\lambda,\hat y)|
&\lesssim
C_1\eps\vartheta^{3/2+7\delta/2}\lambda^{-1+9\delta}.
\end{aligned}
\end{equation}
The first estimate in \eqref{eq:applied-profile-remainder-bounds} gives
\eqref{eq:kg-profile-bound-main}. Since $\phi$ is real-valued,
$b^-=\overline{b^+}$. The first identity in
\eqref{eq:applied-v-vprime-reconstruction} gives the Klein--Gordon expansion
and \eqref{eq:simultaneous-phi-error}. The representation of $b^\pm$ and the
uniform integrability of $F^\pm$ on compact subsets of $B_1(0)$ give the
continuity of $b^\pm$.

For \eqref{eq:simultaneous-dt-error}, the second identity in
\eqref{eq:applied-v-vprime-reconstruction} gives
\[
\begin{aligned}
v'-\mathcal D
={}&-ci(a-1)e^{-i\Phi}b^+
+ci(a-1)e^{i\Phi}b^-
\\
&-ci a e^{-i\Phi}\mathfrak r^+
+ci a e^{i\Phi}\mathfrak r^-.
\end{aligned}
\]
Since $|a-1|\lesssim|\mathfrak q|$, \eqref{eq:ode-quantitative-bounds} and
\eqref{eq:applied-profile-remainder-bounds}, together with
$\lambda\vartheta\geq1/2$, give
\begin{equation}\label{eq:vprime-leading-error}
|v'-\mathcal D|
\lesssim
C_1\eps\vartheta^{3/2+7\delta/2}\lambda^{-1+9\delta}.
\end{equation}
By \eqref{eq:ray-derivative-identity}, \eqref{eq:def-vtx}, and
$\hat L_a\hat s=0$,
\begin{equation}\label{eq:direct-dt-phi-identity}
\vartheta\hat\del_0\phi
=
\lambda^{-3/2}\left(
 v'-\lambda^{-1}\left(\frac32v+\hat y^a\hat L_av\right)
\right).
\end{equation}
Since $\hat L_a\lambda=0$,
\[
\lambda^{-1}|\hat L_av|=\lambda^{1/2}|\hat L_a\phi|.
\]
By \eqref{eq1-2026-8-5}, \eqref{eq:pointwise-phi-high}, and
\eqref{eq:ray-coordinate-comparison},
\begin{equation}\label{eq:angular-term-direct-dt}
\lambda^{-1}|\hat y^a\hat L_av|
\lesssim
C_1\eps\vartheta^{3/2-\delta}\lambda^{-1+3\delta},
\end{equation}
and similarly
\[
\lambda^{-1}|v|
\lesssim
C_1\eps\vartheta^{3/2-\delta}\lambda^{-1+3\delta}.
\]
Together with \eqref{eq:vprime-leading-error} and
\eqref{eq:direct-dt-phi-identity}, these estimates give
\eqref{eq:simultaneous-dt-error} for $\lambda\geq\lambda_0(\hat y)$.

If $\lambda<\lambda_0(\hat y)$ and $\hat t-\hat r\geq1/2$, then
\[
\frac12\left(\frac{1+|\hat y|}{1-|\hat y|}\right)^{1/2}
\leq\lambda<s_0,
\]
so $\vartheta\gtrsim_{s_0}1$ and the corresponding spacetime region is
bounded. More explicitly,
\[
\lambda\vartheta=(\hat t-\hat r)(1+|\hat y|)\geq\frac12.
\]
Hence $\vartheta>1/(2s_0)$, $\hat t<2s_0^2$, and $\lambda\geq1/2$.
The phase is taken with the real extension specified above. Therefore
$|\mathcal A|+|\mathcal D|\lesssim_c |b^+|$ on this region.
The pointwise estimate \eqref{eq:pointwise-phi-sharp},
\eqref{eq3-04-june-2026}, and the first estimate in
\eqref{eq:applied-profile-remainder-bounds} give
\eqref{eq:simultaneous-phi-error} and \eqref{eq:simultaneous-dt-error}
there, after enlarging the fixed constant. No metric square root below
$s_0$ is needed for this bounded-region argument.
\end{proof}

\subsection{Derivatives of the phase coefficient}
The following estimate is used in the proof of
Proposition~\ref{prop:kg-energy-profile}.

\begin{lemma}[Profile derivatives of the ODE coefficient]
\label{lem:app-q-y-derivatives}
Let $\mathfrak q(\lambda,\hat y)$ be defined by \eqref{eq:def-vtx}.
For $0<\rho,\eta<1/4$, set
\[
K_{\rho,\eta}:=\{\hat y\in B_1(0):|\hat y|\geq\rho,
\sqrt{1-|\hat y|^2}\geq\eta\}.
\]
Then, for $k=1,2$,
\begin{equation}\label{eq:app-q-y-derivatives}
|\del_{\hat y}^k \mathfrak q(\lambda,\hat y)|
\lesssim_{\rho,\eta}C_1\eps\,\lambda^{-1+7\delta/2},
\qquad \lambda\geq\hat t_{\mathrm{init}},\quad \hat y\in K_{\rho,\eta}.
\end{equation}
\end{lemma}

\begin{proof}
At fixed $\lambda$, differentiation along the family of hyperboloidal
rays gives
\begin{equation}\label{eq:app-y-derivative-vectorfields}
\del_{\hat y_i}
=\frac1{1-|\hat y|^2}
\left(\hat\Omega_{0i}+\hat y_j\Omega_{ij}\right).
\end{equation}
On $K_{\rho,\eta}$ the factor $(1-|\hat y|^2)^{-1}$ and all of its first
two $\hat y$-derivatives are bounded by constants depending only on
$(\rho,\eta)$. Moreover $\hat\Omega_{0i}=\hat L_i$, while the spatial
rotations are bounded combinations of the hatted Lorentz boosts; indeed,
\[
\Omega_{ij}=\hat y_i\hat L_j-\hat y_j\hat L_i.
\]
Thus \eqref{eq:app-y-derivative-vectorfields} expresses the first two
$\hat y$-derivatives in terms of at most two boosts or rotations.
Write
\[
\hat L_i=L_i+A_i^\mu\del_\mu,\qquad
A_i^0=\frac{2m_{\mathrm S}x^i}{t-2m_{\mathrm S}},\qquad
A_i^a=-2m_{\mathrm S}\ln t\,\delta_i^a.
\]
On $K_{\rho,\eta}$, the coefficients satisfy
\[
|A_i|+|L_jA_i|\lesssim_{\rho,\eta}
m_{\mathrm S}(1+\ln t),\qquad
|\del A_i|\lesssim_{\rho,\eta}m_{\mathrm S}t^{-1}.
\]
The second hatted boost expands as
\[
\begin{aligned}
\hat L_i\hat L_jH={}&L_iL_jH+(L_iA_j^\nu)\del_\nu H
+A_j^\nu L_i\del_\nu H+A_i^\mu\del_\mu L_jH\\
&+A_i^\mu(\del_\mu A_j^\nu)\del_\nu H
+A_i^\mu A_j^\nu\del_\mu\del_\nu H.
\end{aligned}
\]
For fixed $K_{\rho,\eta}$, choose $\Lambda_{\rho,\eta}\geq s_0$
so that every corresponding point with $\lambda>\Lambda_{\rho,\eta}$
lies in $\Kcal^m$. Using \eqref{eq:pointwise-H} and \eqref{eq:pointwise-del-H}, together
with $t-r\simeq_{\rho,\eta}t$ on this compact set of rays, we obtain
\[
|\hat L_iH|+|\hat L_i\hat L_jH|
\lesssim_{\rho,\eta}C_1\eps
\left(t^{-1+7\delta/2}
+(1+\ln t)^2t^{-2+7\delta/2}\right)
\lesssim_{\rho,\eta}C_1\eps t^{-1+7\delta/2}.
\]
All derivatives on the right are covered by the stated commuted norms,
since $N\geq14$. Differentiating
\eqref{eq:app-y-derivative-vectorfields} once more only produces bounded
coefficient derivatives on $K_{\rho,\eta}$. Hence, for $k=1,2$,
\[
|\del_{\hat y}^kH|
\lesssim_{\rho,\eta}C_1\eps\lambda^{-1+7\delta/2},
\]
because $t\simeq_{\rho,\eta}\lambda$ along these rays.

The remaining values $s_0\leq\lambda\leq\Lambda_{\rho,\eta}$ form
a bounded spacetime region. Local smooth small-data bounds, together
with the explicit Schwarzschild formula wherever the ray is exterior,
give the same estimate after enlarging its fixed constant.
Finally, $\mathfrak q=\vartheta^{-2}HB^2$. The first two derivatives
of $\vartheta^{-2}$ are bounded on
$K_{\rho,\eta}$, while derivatives of
$B=1-2m_{\mathrm S}t^{-1}$ are lower order. Using also
$m_{\mathrm S}\lesssim C_1\eps$ proves \eqref{eq:app-q-y-derivatives}.
\end{proof}

\subsection{Scalar-source remainder}

\begin{lemma}[Scalar-source remainder]
\label{lem:scalar-source-remainder}
Let $R_0$ be the remainder in \eqref{eq:modified-source-decomposition}. Then
\[
|R_0(\hat t,\hat x)|
\lesssim_\delta (C_1\eps)^2\hat t^{-4+9\delta},
\qquad
\hat t-\hat r\geq\frac12.
\]
\end{lemma}

\begin{proof}
Recall from the proof of Proposition~\ref{prop:scalar-induced-source-profile} that
\[
\lambda=\hat s,
\qquad
\vartheta=\frac{\hat s}{\hat t},
\qquad
B=1-2\mSch t^{-1},
\qquad
\left(\frac{s}{t}\right)^2=\vartheta^2+\mathcal E_m,
\]
and
\[
\lambda^{3/2}\phi=\mathcal A+\rho_\phi,
\qquad
\lambda^{3/2}\vartheta\hat\del_0\phi=\mathcal D+\rho_t.
\]
By \eqref{eq:source-before-cancellation} and
\eqref{eq:leading-source-cancellation},
\begin{equation}\label{eq:scalar-source-remainder-expansion}
\begin{aligned}
R_0=8\pi\lambda^{-3}\Big(&
2c^2\mathcal A\rho_\phi
+2B^2\mathcal D\rho_t
+c^2\rho_\phi^2\\
&+B^2\rho_t^2
+(B^2-1)\mathcal D^2\Big)
+8\pi\mathcal E_m(\del_t\phi)^2.
\end{aligned}
\end{equation}
Theorem~\ref{thm:refined-kg-asymptotics} gives
\begin{equation}\label{eq:source-local-remainder-bounds}
\begin{aligned}
|\rho_\phi|
&\lesssim C_1\eps\vartheta^{3/2+7\delta/2}\lambda^{-1+9\delta},
\\
|\rho_t|
&\lesssim C_1\eps\left[
\vartheta^{3/2+7\delta/2}\lambda^{-1+9\delta}
+\vartheta^{3/2-\delta}\lambda^{-1+3\delta}\right].
\end{aligned}
\end{equation}
Moreover,
\[
|\mathcal A|+|\mathcal D|
\lesssim C_1\eps\vartheta^{5/2-11\delta/2}.
\]
Since $B\simeq1$, the $\rho_\phi$ term and the first term in $\rho_t$ satisfy
\[
\lambda^{-3}\left(
|\mathcal A||\rho_\phi|
+|\mathcal D|\,C_1\eps
\vartheta^{3/2+7\delta/2}\lambda^{-1+9\delta}
\right)
\lesssim
(C_1\eps)^2\hat t^{-4+9\delta}.
\]
For the second term in $\rho_t$,
\[
\begin{aligned}
\lambda^{-3}|\mathcal D|\,
C_1\eps\vartheta^{3/2-\delta}\lambda^{-1+3\delta}
&\lesssim
(C_1\eps)^2\lambda^{-4+3\delta}\vartheta^{4-13\delta/2}
\\
&=(C_1\eps)^2\hat t^{-4+3\delta}\vartheta^{-7\delta/2}
\\
&\lesssim (C_1\eps)^2\hat t^{-4+19\delta/4}
\lesssim (C_1\eps)^2\hat t^{-4+9\delta},
\end{aligned}
\]
where $\lambda\vartheta=\hat t\vartheta^2\geq1/2$. The quadratic remainder terms satisfy
\[
\begin{aligned}
\lambda^{-3}\bigl(|\rho_\phi|^2+|\rho_t|^2\bigr)
\lesssim
(C_1\eps)^2\Big(
&\vartheta^{3+7\delta}\lambda^{-5+18\delta}
+\vartheta^{3-2\delta}\lambda^{-5+6\delta}
\Big)
\lesssim
(C_1\eps)^2\hat t^{-4+9\delta}.
\end{aligned}
\]
Next, $|B^2-1|\lesssim\mSch t^{-1}$ and
\[
\lambda^{-3}|\mathcal D|^2
\lesssim
(C_1\eps)^2\hat t^{-3}\vartheta^{2-11\delta}
\]
give
\[
|B^2-1|\lambda^{-3}|\mathcal D|^2
\lesssim
(C_1\eps)^2\hat t^{-4+9\delta}.
\]
Finally, by \eqref{eq:pointwise-phi-sharp} and
$|\mathcal E_m|\lesssim\mSch t^{-1}\ln t$,
\[
|\mathcal E_m|\,|\del_t\phi|^2
\lesssim_\delta
(C_1\eps)^2t^{-4+9\delta}.
\]
Here we used $t-r\leq t$, $\mSch\lesssim C_1\eps$, and
$\ln t\lesssim_\delta t^\delta$. Combining these estimates in
\eqref{eq:scalar-source-remainder-expansion} proves the claim.
\end{proof}


\end{document}